\RequirePackage{fix-cm}
\documentclass[twocolumn]{svjour3}          
\smartqed  
\usepackage{graphicx}

\usepackage{svninfo}
\svnInfo $Id: preamble.tex 111 2010-01-29 09:56:52Z antonio $

\usepackage[us]{datetime}
\usepackage{graphicx,color}
\usepackage{subfig} 
\usepackage{caption}
\usepackage{url}

\usepackage{amsmath,amssymb,mathrsfs,amsthm}

\newcommand{\vect}[1]{\boldsymbol{#1}}
\newcommand{\mat}[1]{\boldsymbol{#1}}

\newcommand{\diffs}[3]{\frac{\partial^2 #1}{
\ifx#2#3 
\partial #2^2
\else
\partial #2 \partial #3
\fi
}}

\newcommand{\av}{\vect{a}}

\newcommand{\pv}{\vect{p}}
\newcommand{\dpv}{\dot{\vect{p}}}

\newcommand{\qv}{{\vect{q}}}
\newcommand{\dqv}{\dot{\vect{q}}}

\newcommand{\uv}{\vect{u}}
\newcommand{\vv}{\vect{v}}

\newcommand{\omegav}{\vect{\omega}}

\newcommand{\Cm}{\mat{C}}

\newcommand{\Dm}{\mat{D}}

\newcommand{\Jm}{\mat{J}}

\newcommand{\Rm}{\mat{R}}

\usepackage{Cris}

    \newcommand{\argmin}{\operatornamewithlimits{arg\, min}}

\def\figPolarThreeD{\centering\includegraphics[width=0.90\textwidth]{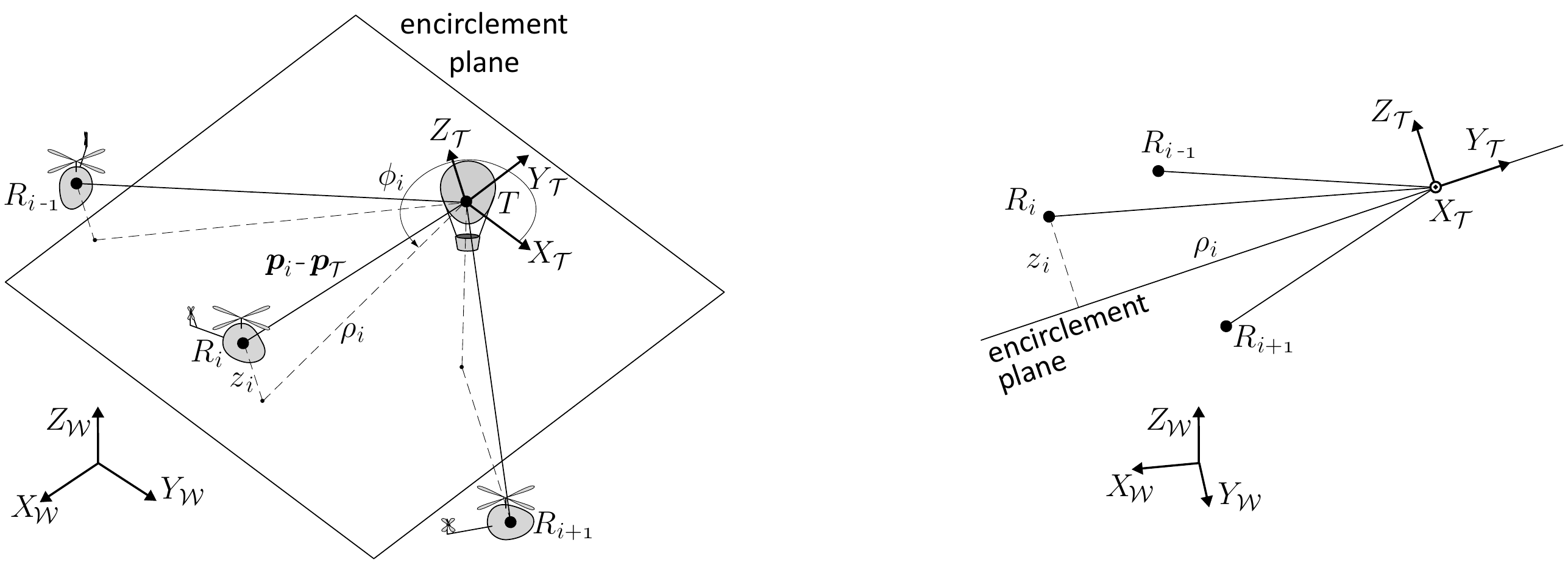}}
\def\figSafeDistances{\centering\includegraphics[width=0.90\textwidth]{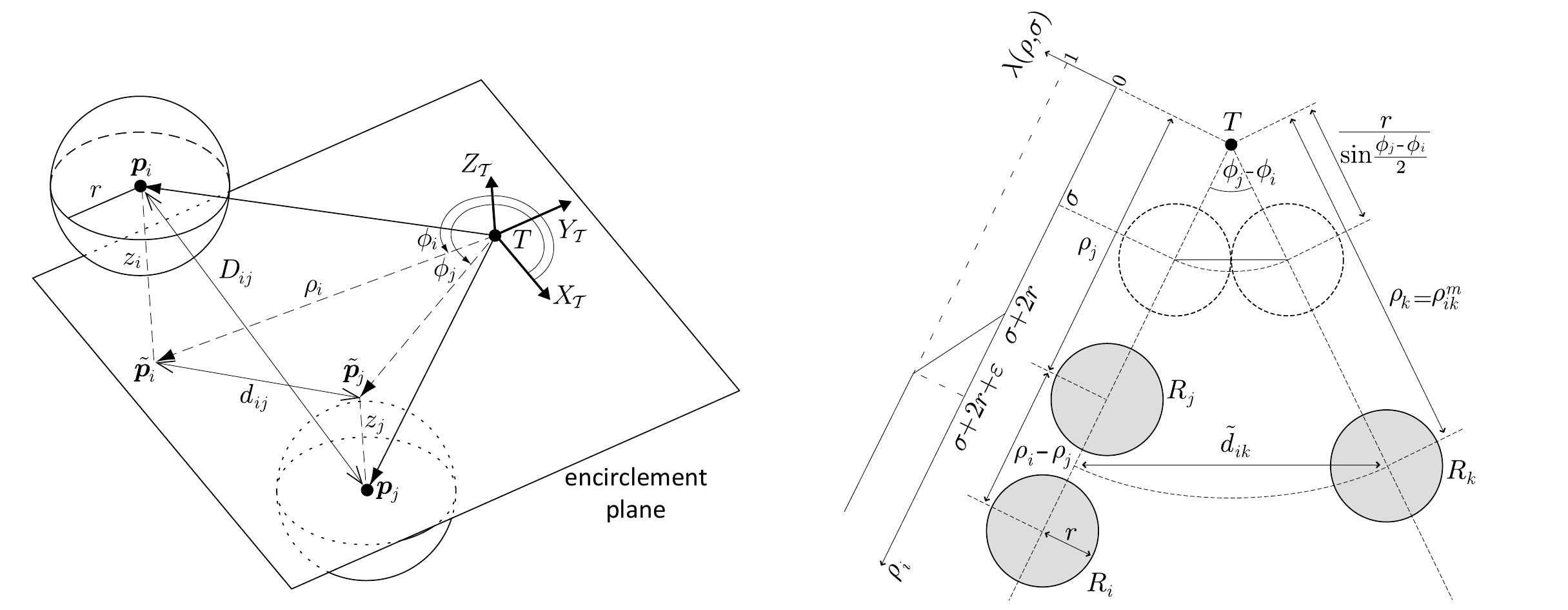}}
\def\figExpStatTarget{\centering\includegraphics[width=0.97\columnwidth]{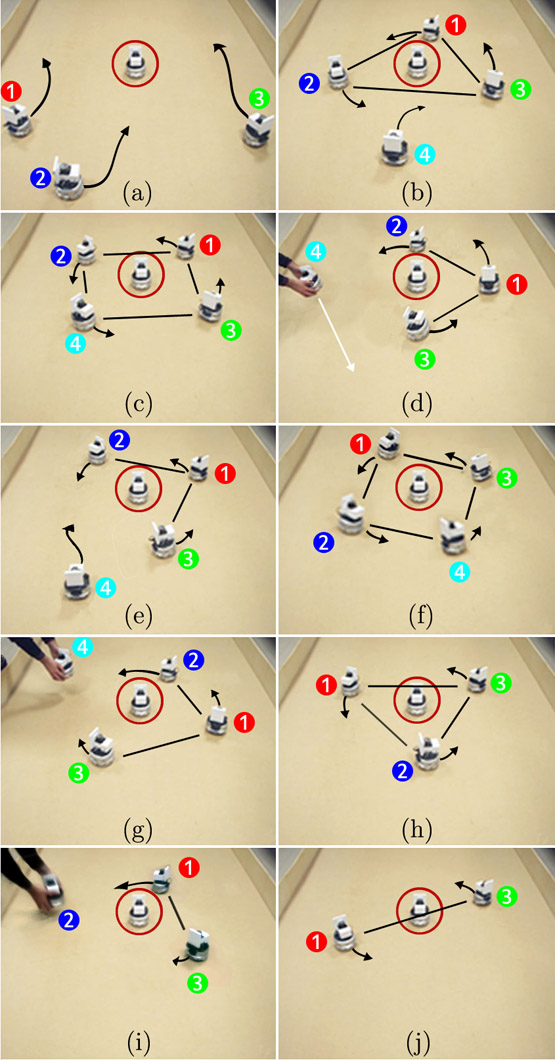}}
\def\figKhIIIExpPlots{\centering\includegraphics[width=0.99\columnwidth]{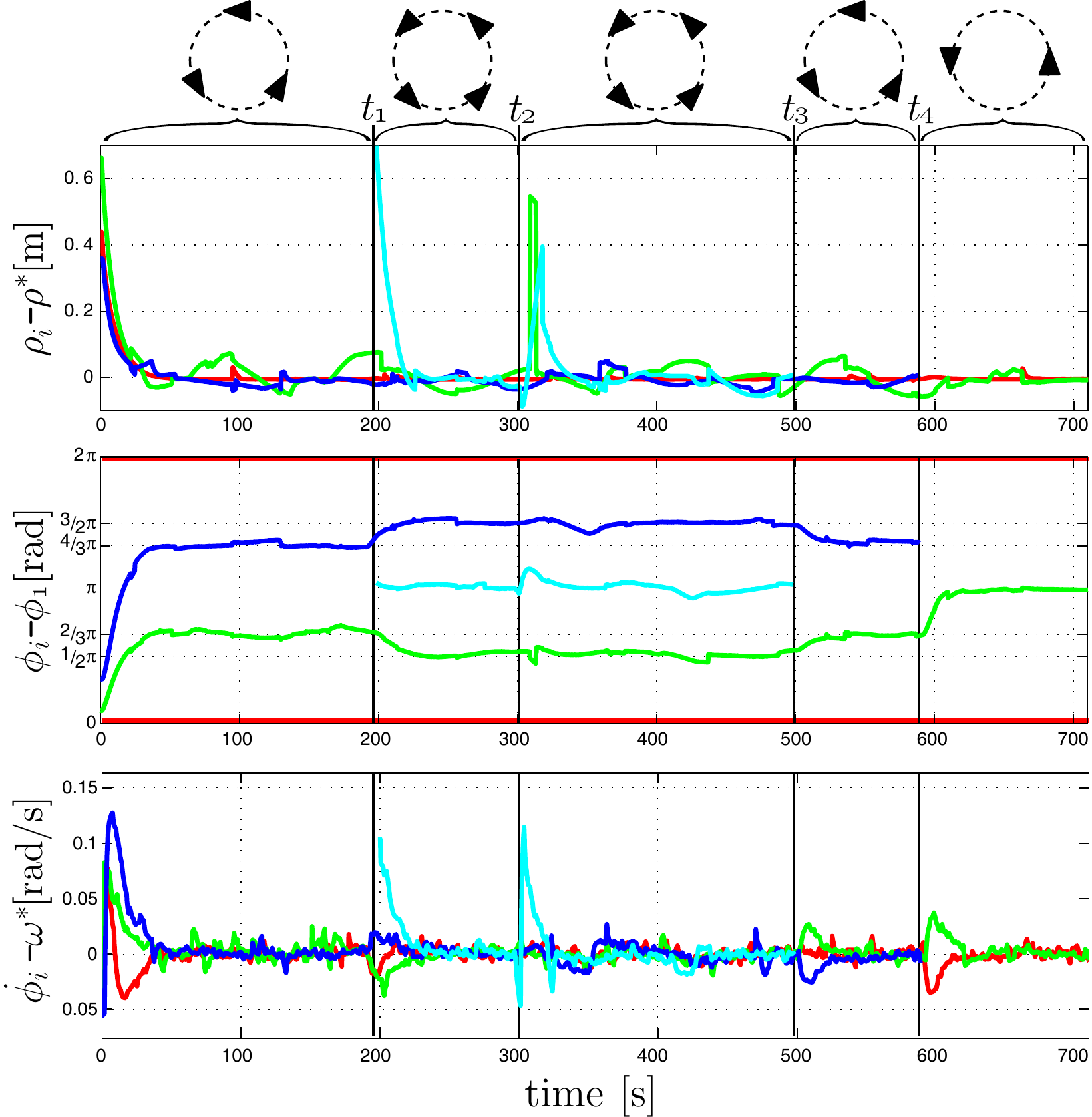}}
\def\figExpMovingTarget{\centering\includegraphics[width=0.99\columnwidth]{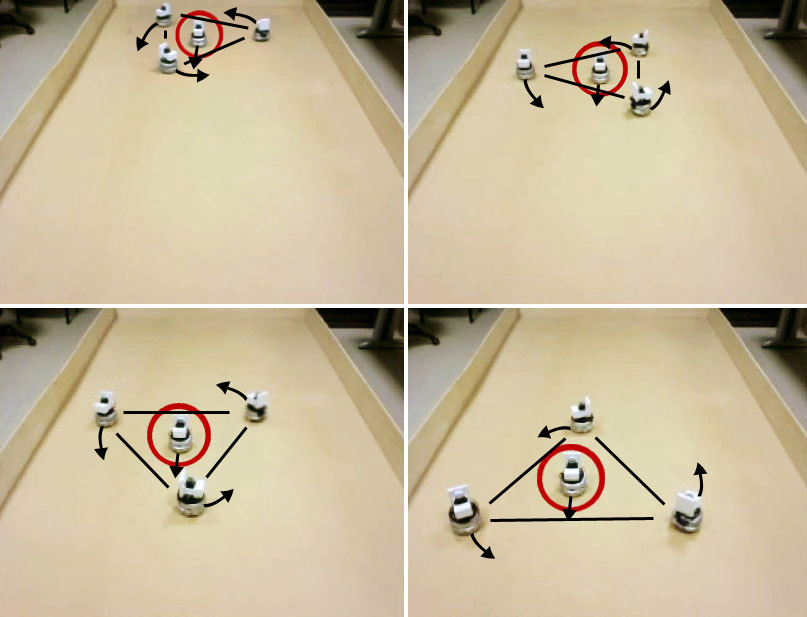}}

\newtheorem{prop}{Proposition}

\ifx\isfinal\undefined

\else
\newcommand{\todo}[1]{}
\newcommand{\afmargin}[1]{}
\newcommand{\gomargin}[1]{}

\fi%

\usepackage{paralist}

\usepackage{graphicx}          

\usepackage{flushend}

\usepackage[textsize=footnotesize]{todonotes}
\usepackage{soul}

\usepackage[round]{natbib}

\graphicspath{{./},{./figures/}}

\def\FigSimAErrRho{{\includegraphics[width=0.49\columnwidth]{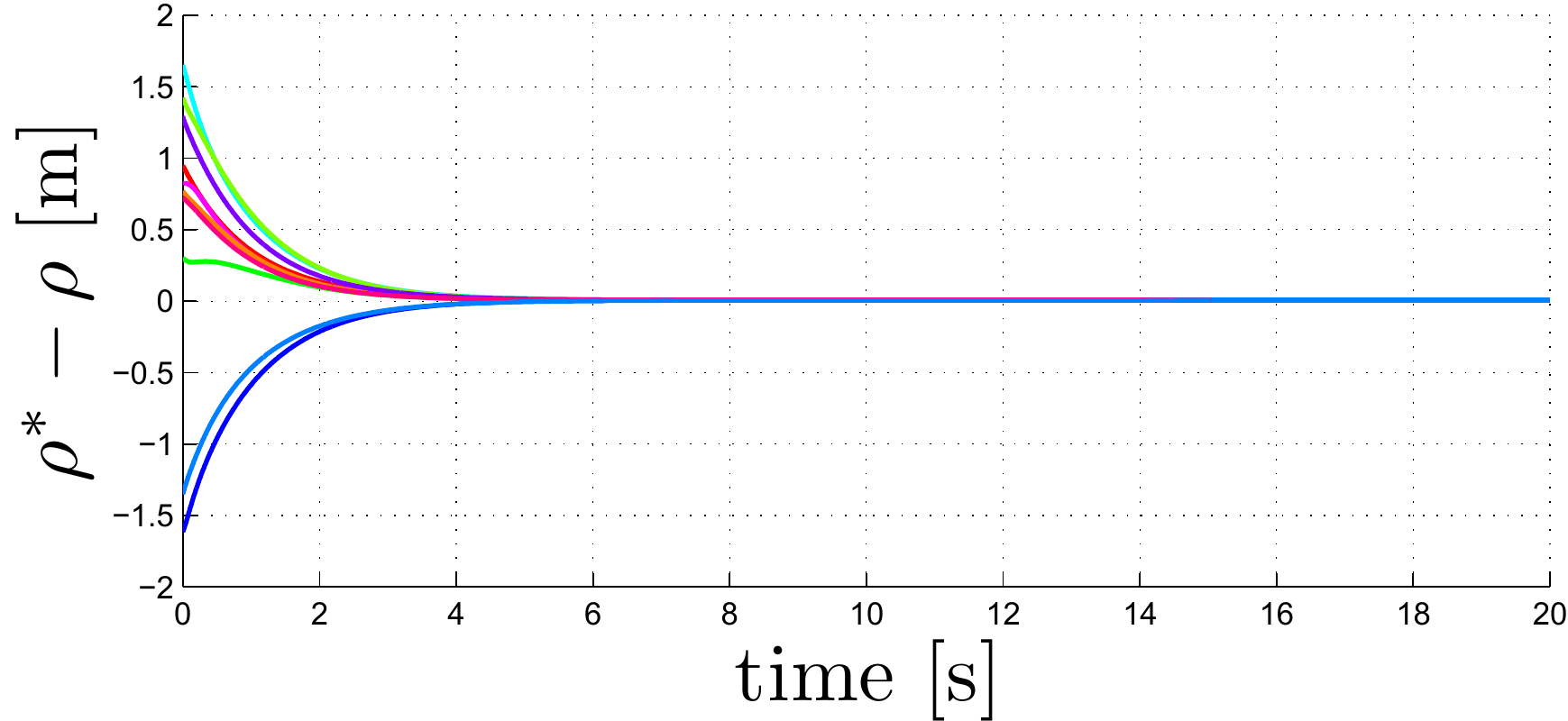}}}
\def\FigSimAErrPhi{{\includegraphics[width=0.49\columnwidth]{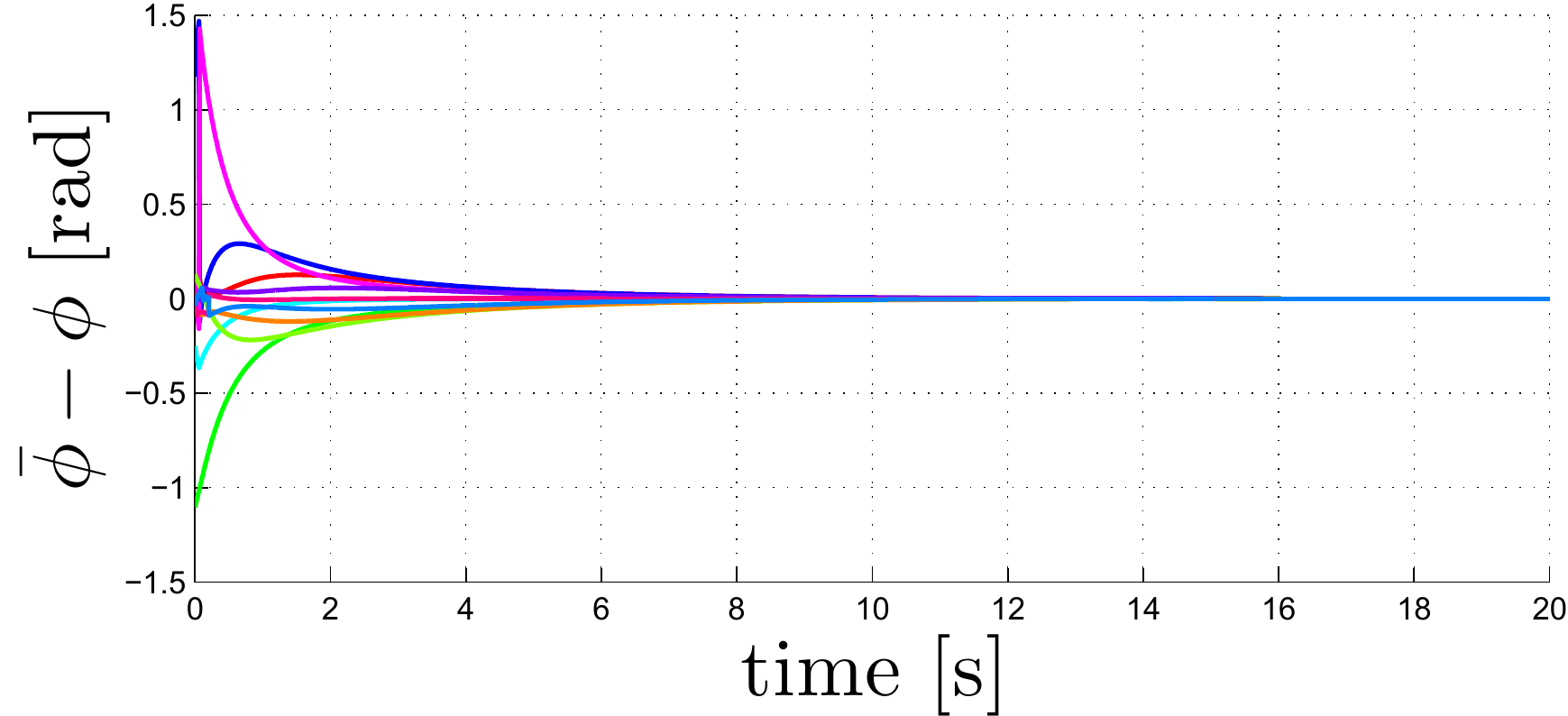}}}
\def\FigSimAErrDotPhi{{\includegraphics[width=0.49\columnwidth]{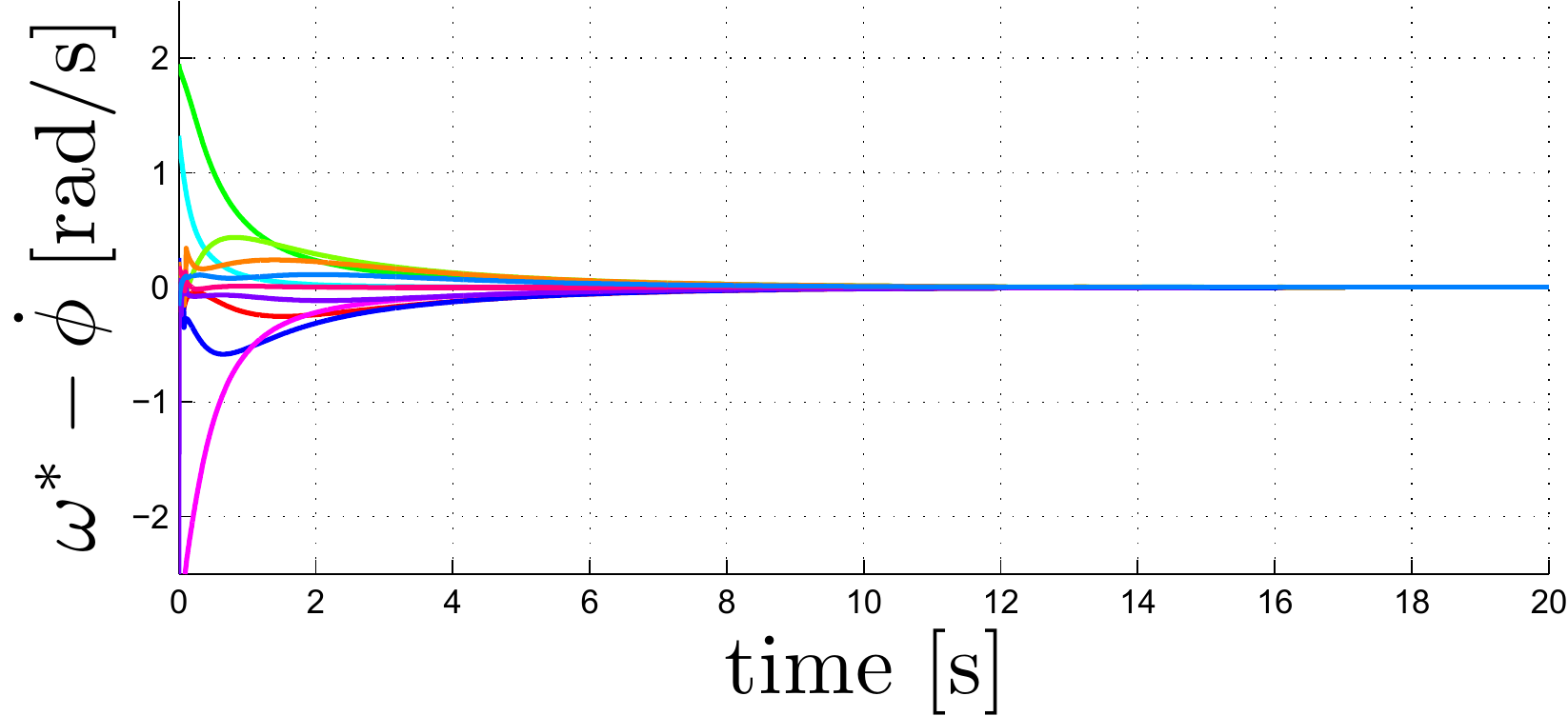}}}
\def\FigSimAErrZ{{\includegraphics[width=0.49\columnwidth]{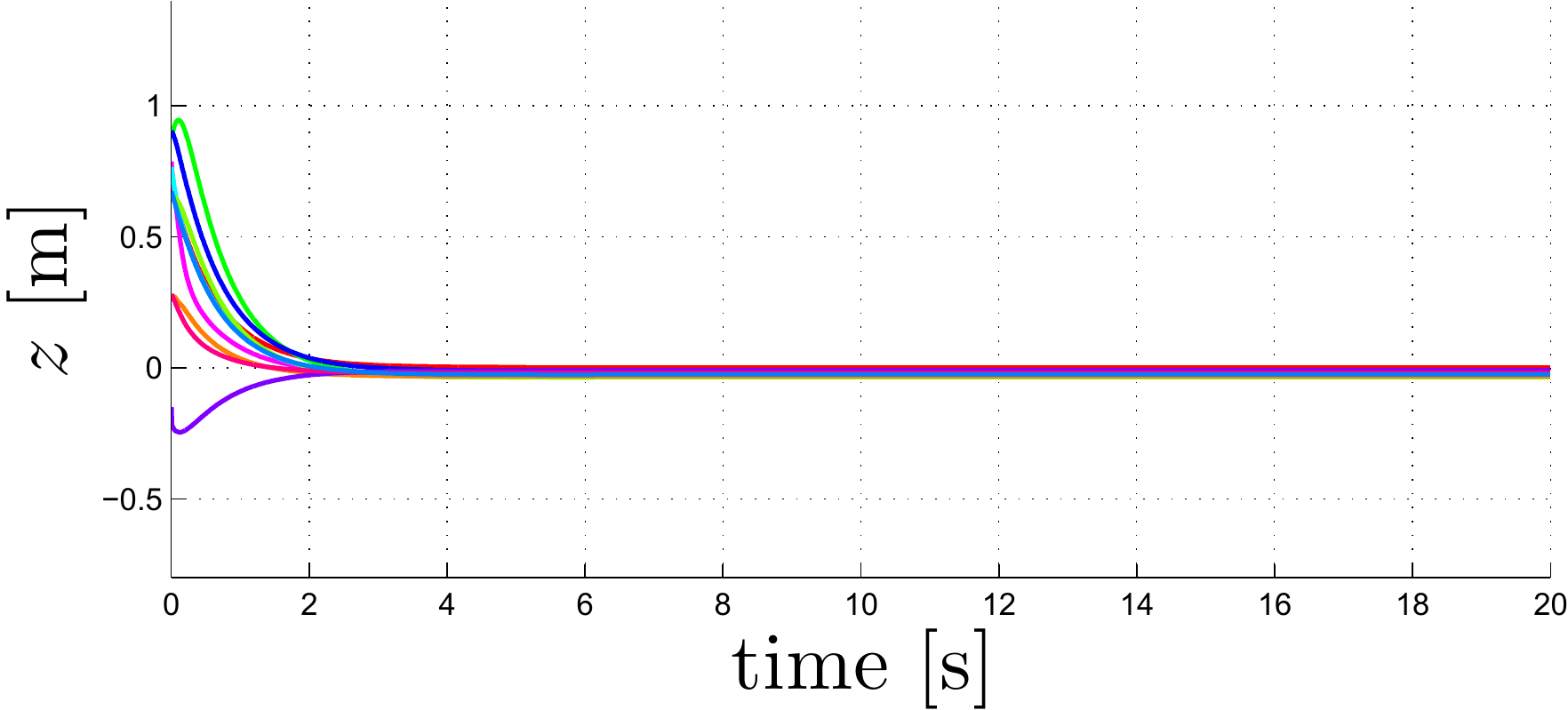}}}
\def\FigSimATrajXY{{\includegraphics[width=0.66\columnwidth]{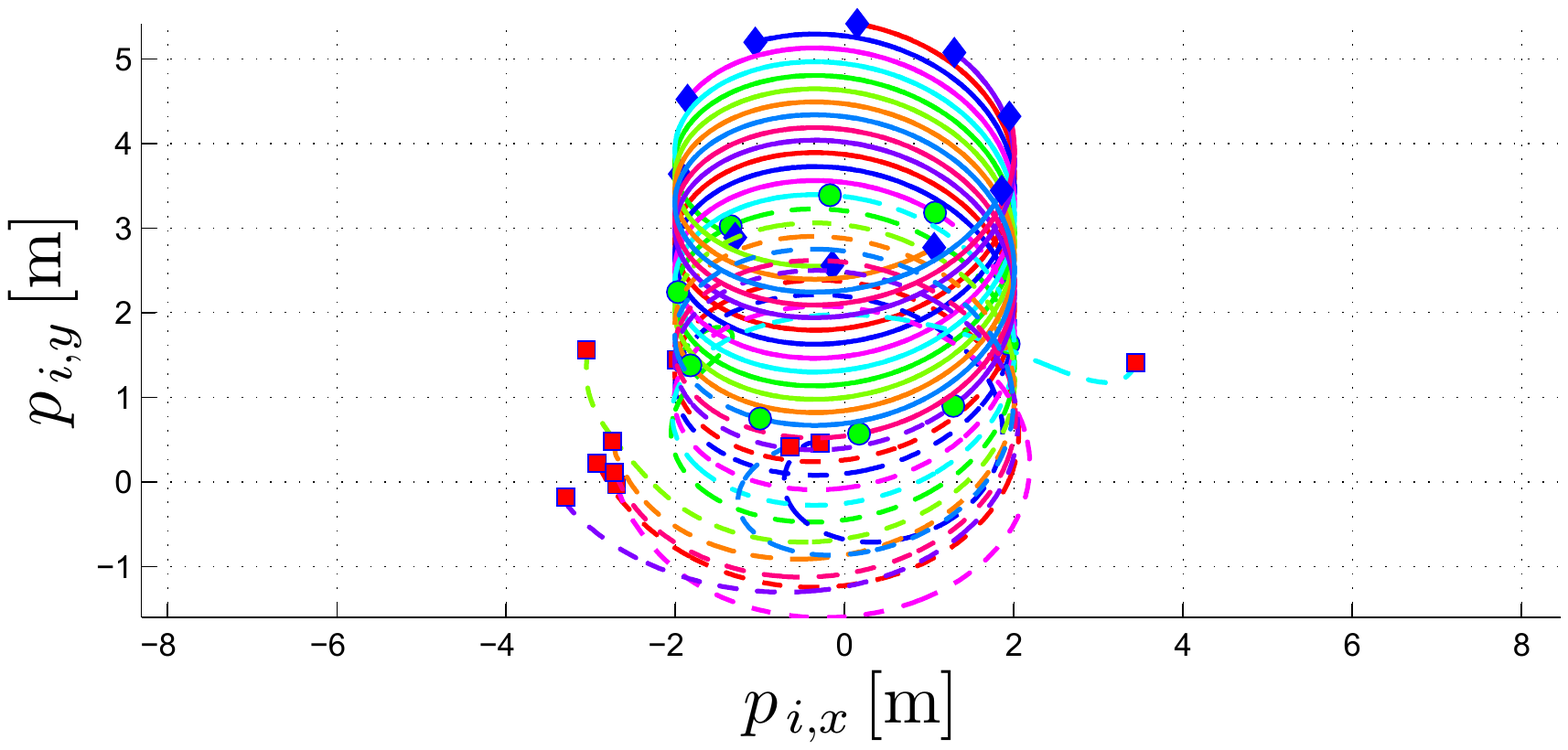}}}
\def\FigSimATrajXZ{{\includegraphics[width=0.66\columnwidth]{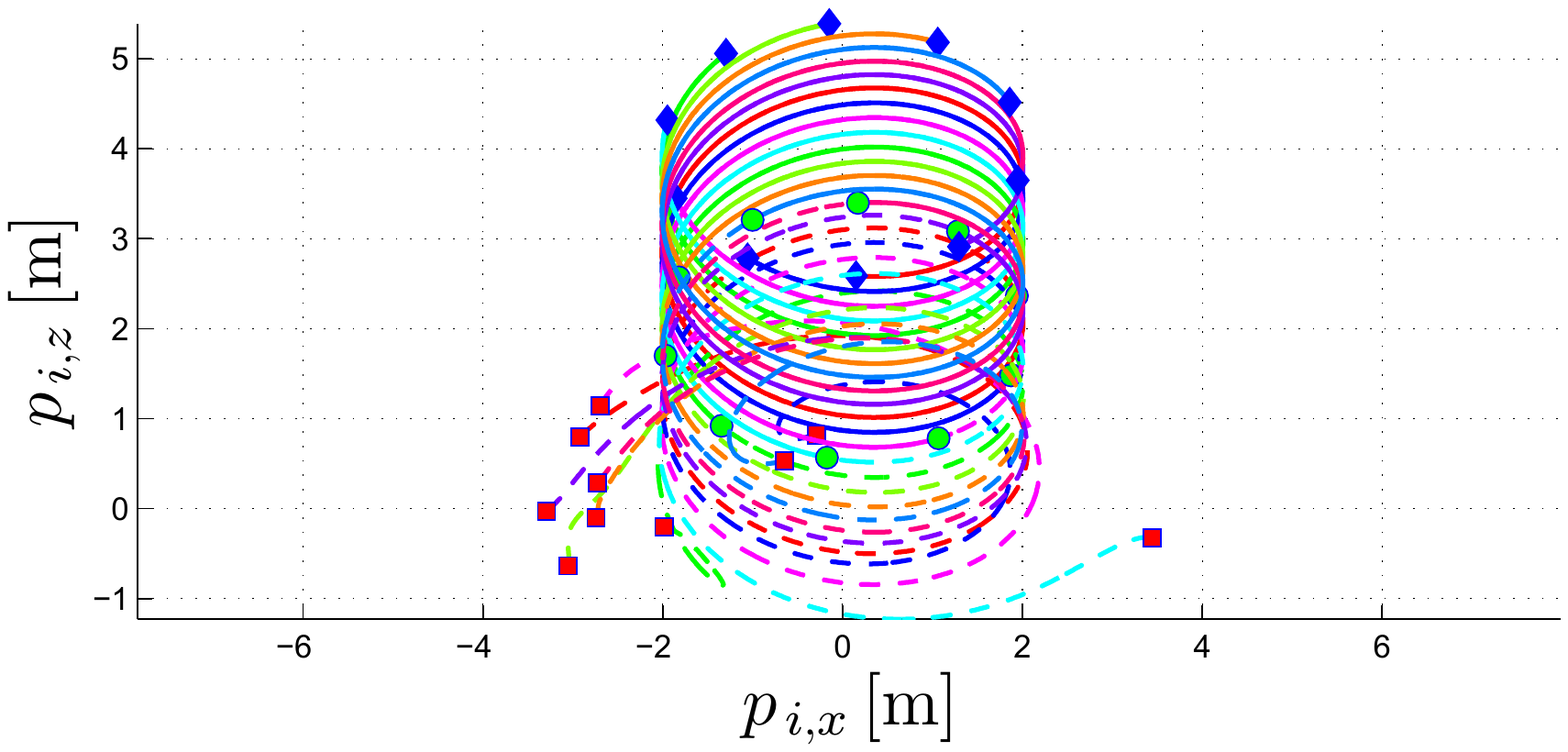}}}
\def\FigSimATrajYZ{{\includegraphics[width=0.66\columnwidth]{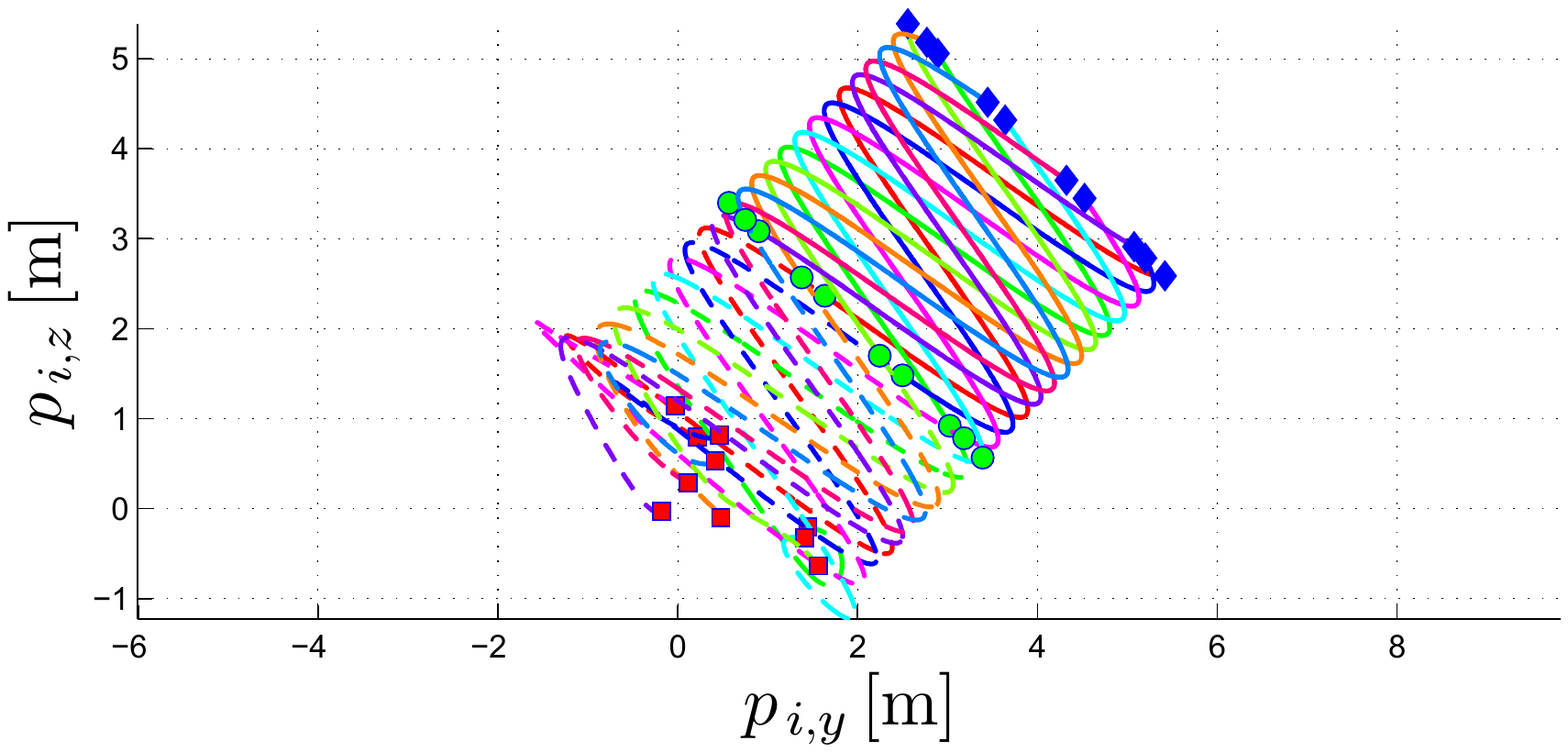}}}

\def\FigSimBErrRho{{\includegraphics[width=0.49\columnwidth]{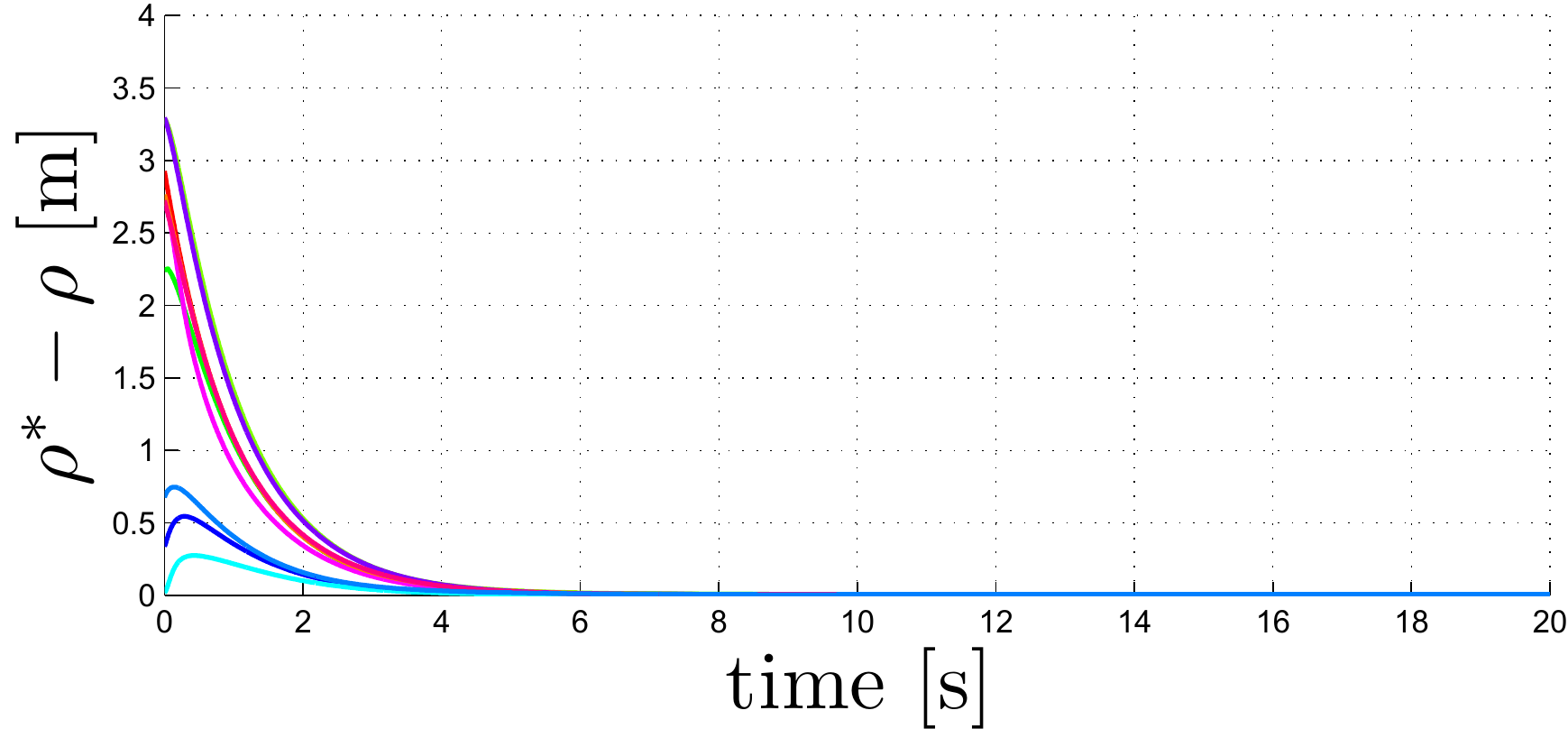}}}
\def\FigSimBErrPhi{{\includegraphics[width=0.49\columnwidth]{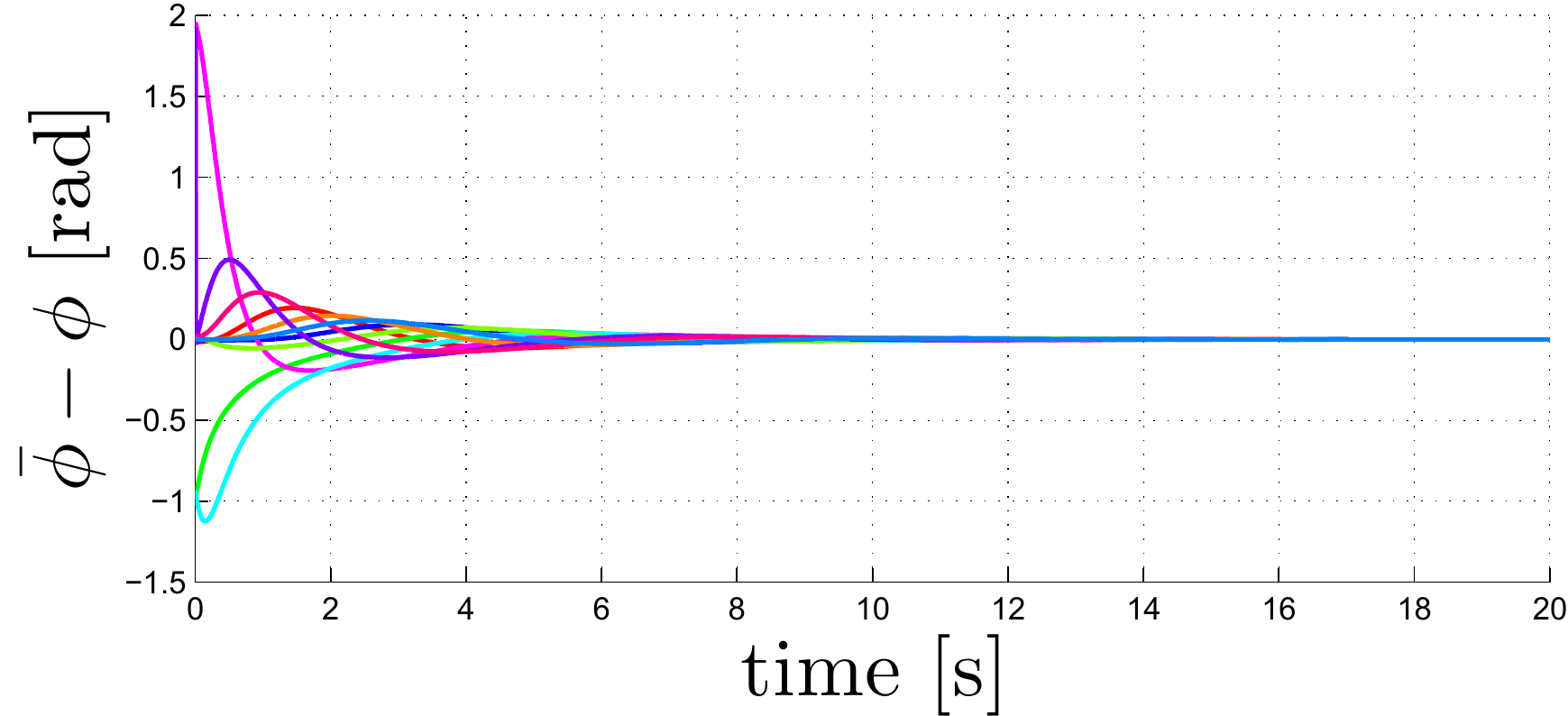}}}
\def\FigSimBErrDotPhi{{\includegraphics[width=0.49\columnwidth]{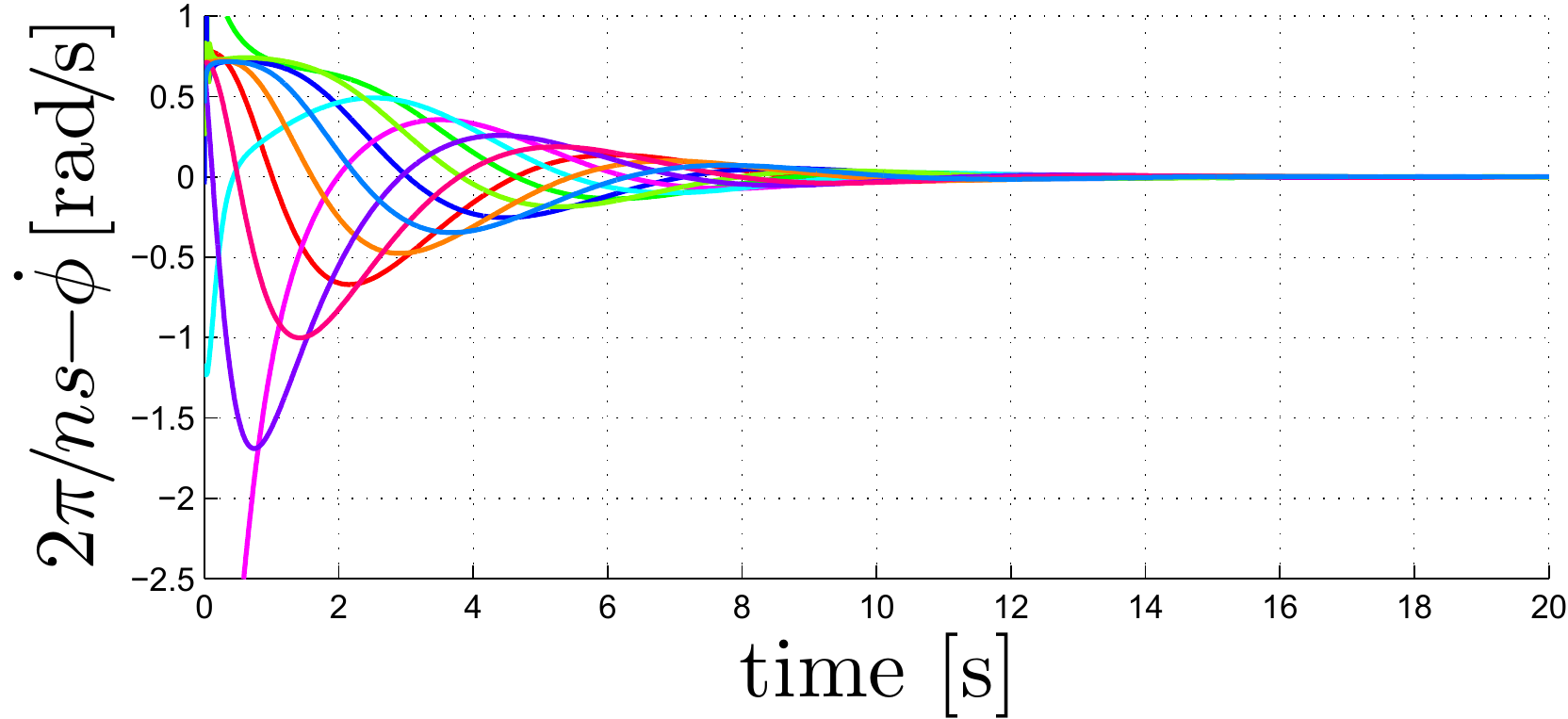}}}
\def\FigSimBErrZ{{\includegraphics[width=0.49\columnwidth]{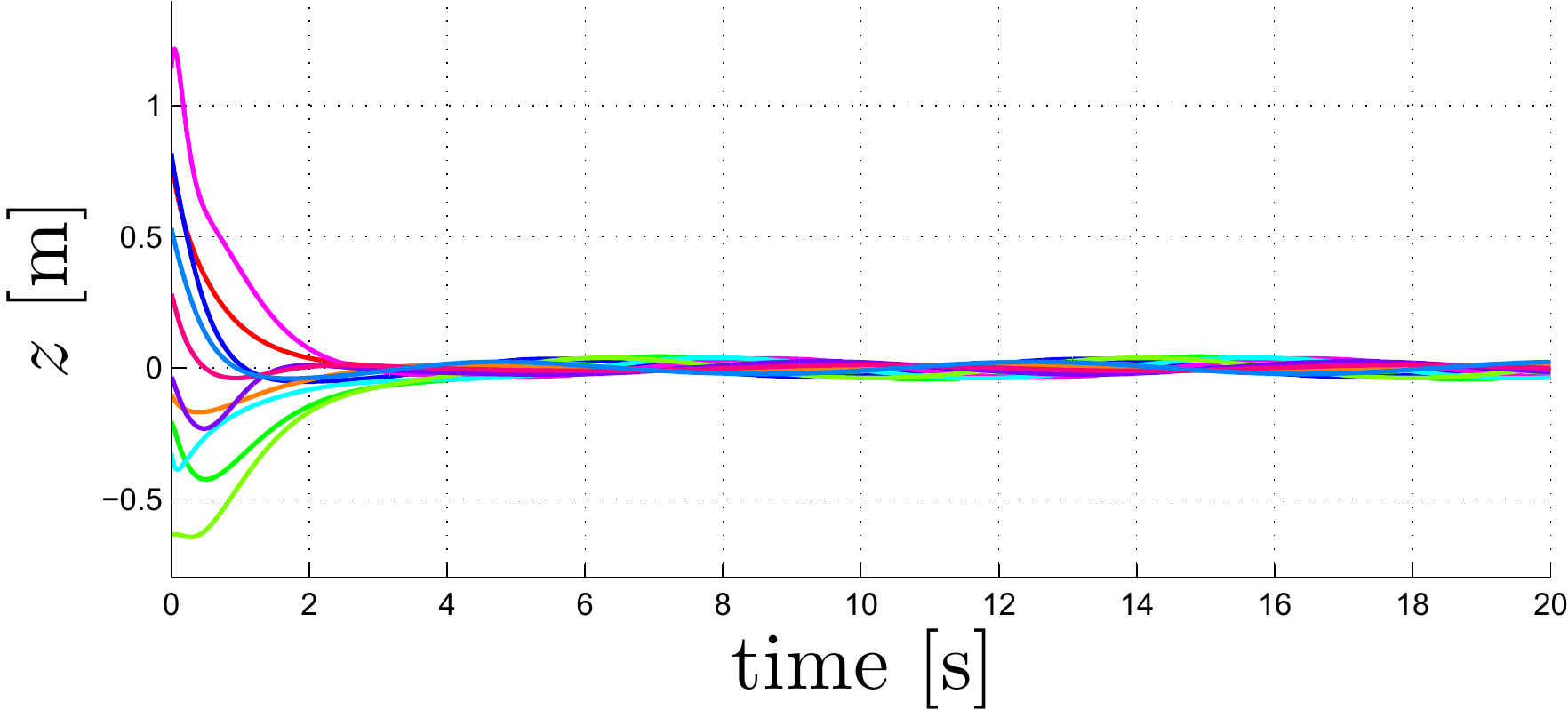}}}
\def\FigSimBTrajXY{{\includegraphics[width=0.66\columnwidth]{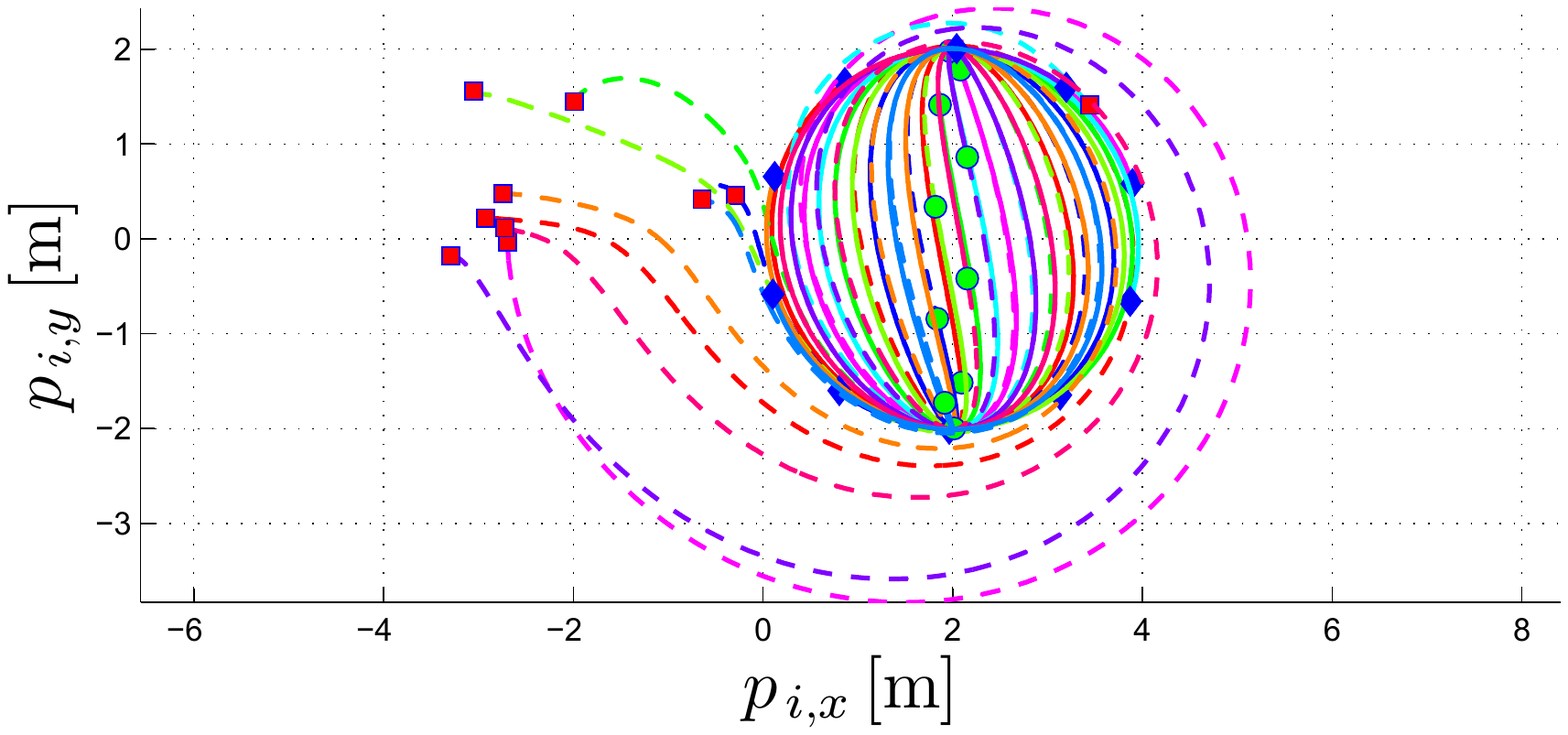}}}
\def\FigSimBTrajXZ{{\includegraphics[width=0.66\columnwidth]{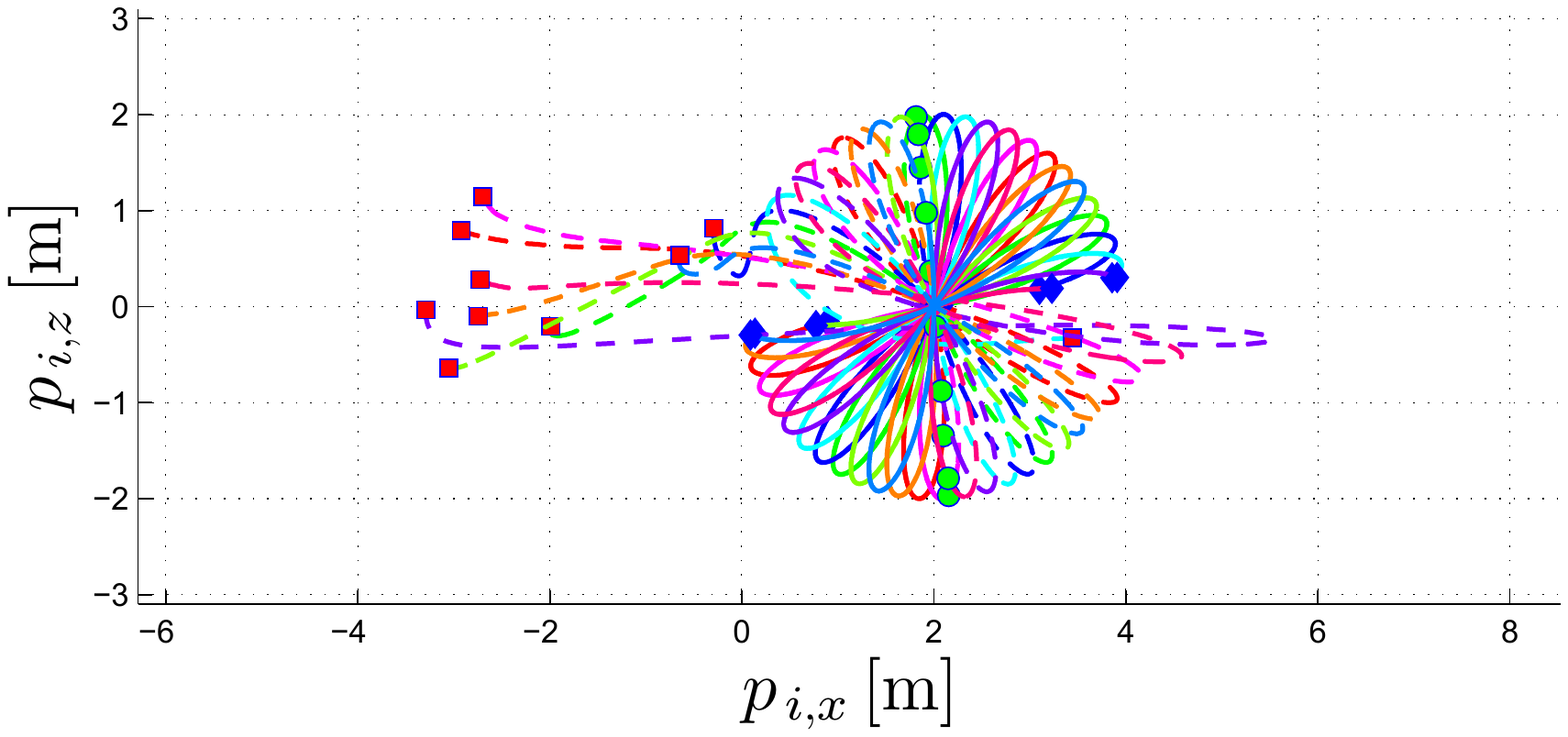}}}
\def\FigSimBTrajYZ{{\includegraphics[width=0.66\columnwidth]{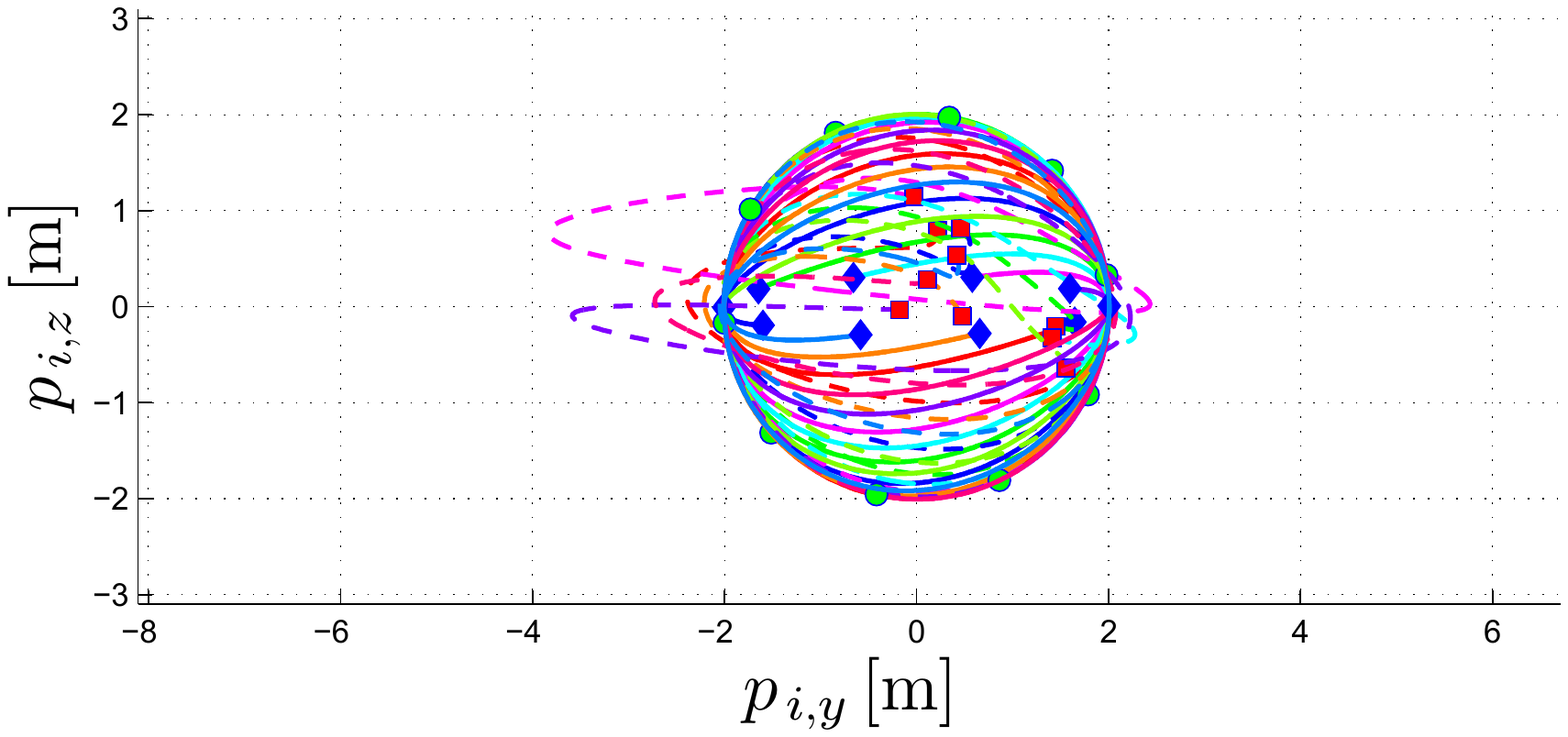}}}

\def\FigSimCErrRho{{\includegraphics[width=0.49\columnwidth]{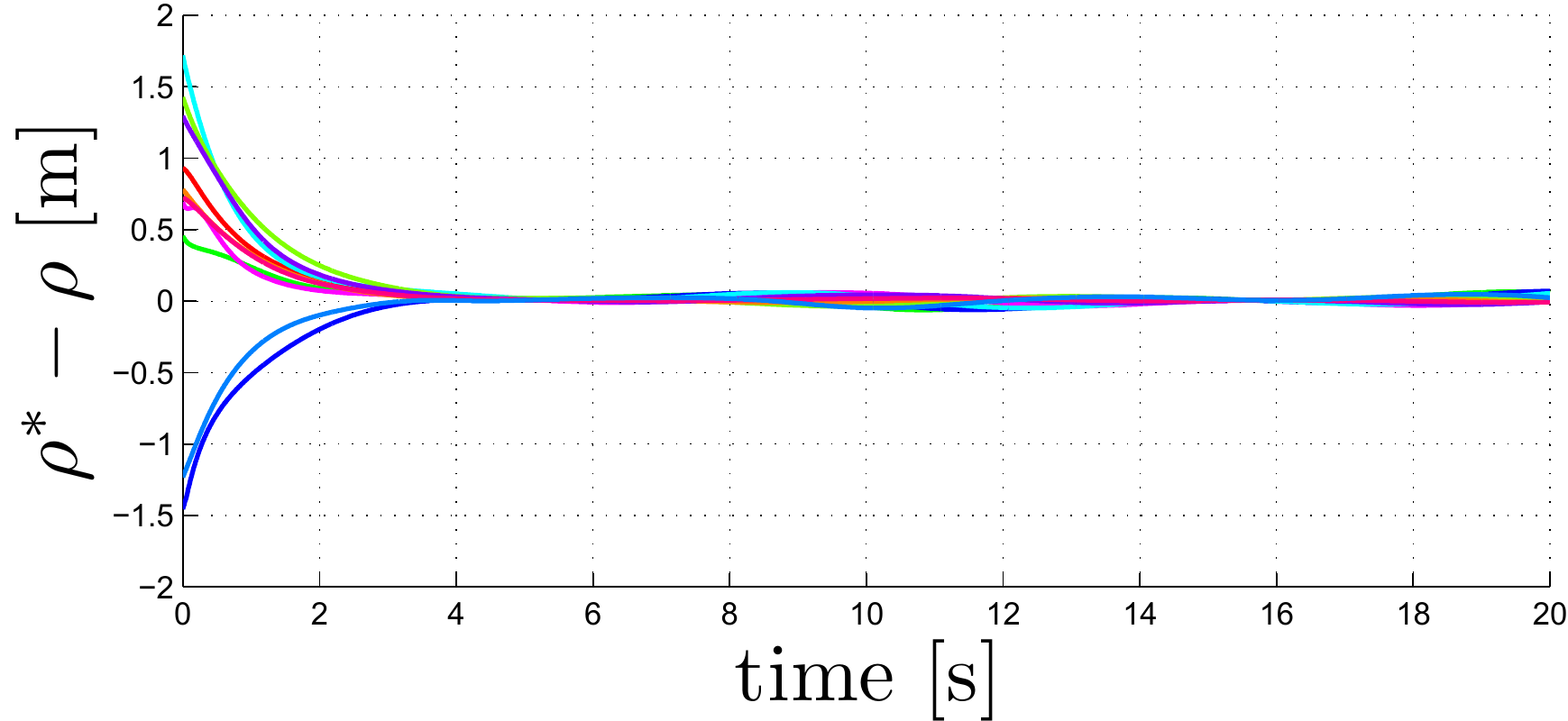}}}
\def\FigSimCErrPhi{{\includegraphics[width=0.49\columnwidth]{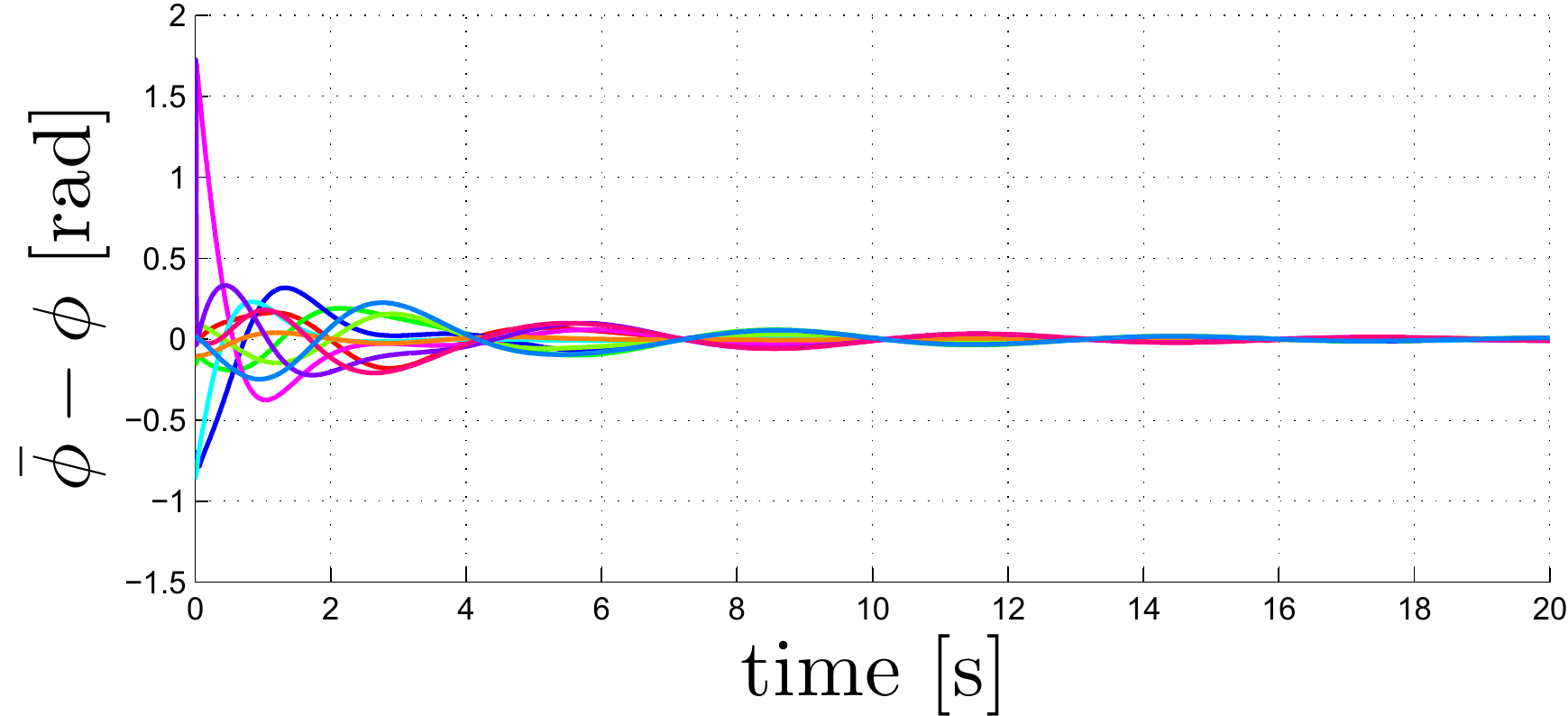}}}
\def\FigSimCErrDotPhi{{\includegraphics[width=0.49\columnwidth]{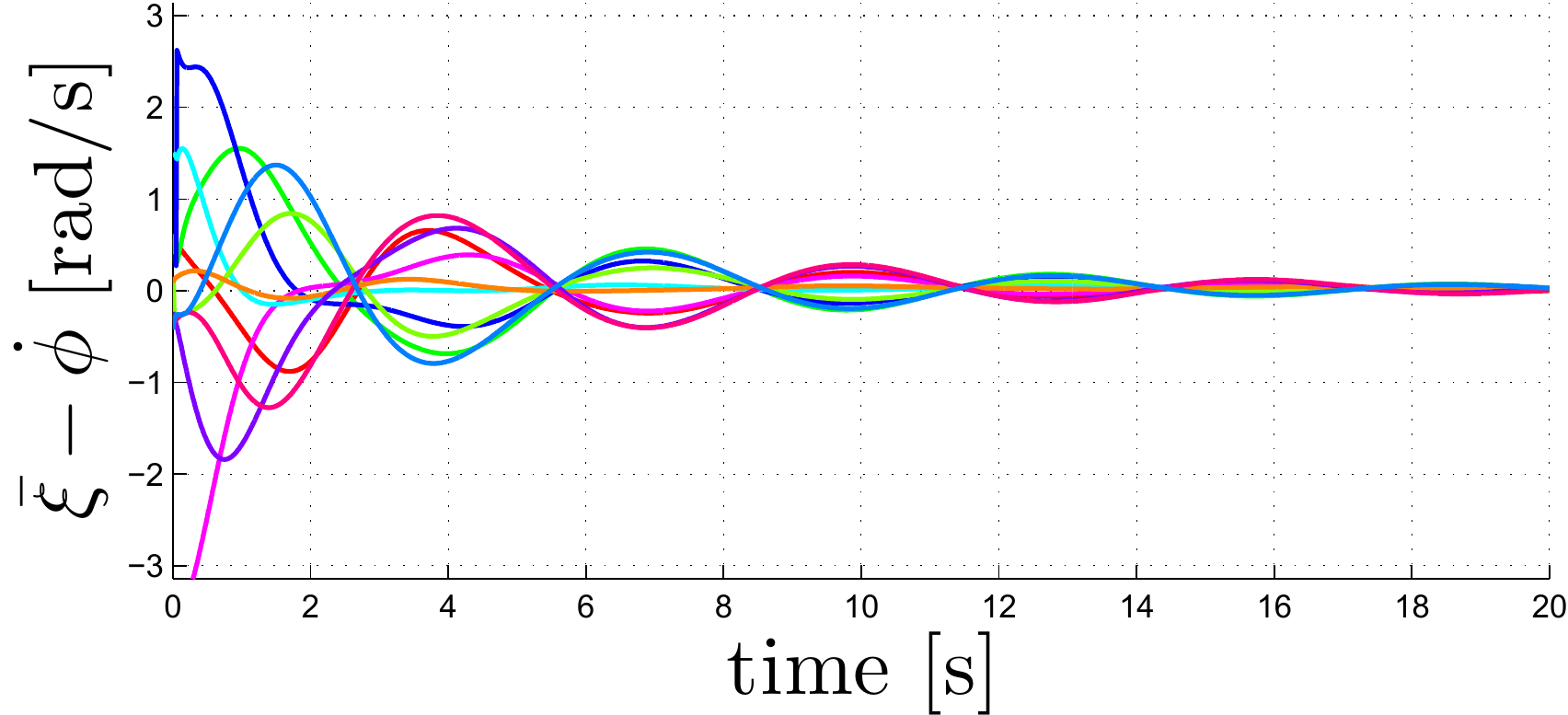}}}
\def\FigSimCErrZ{{\includegraphics[width=0.49\columnwidth]{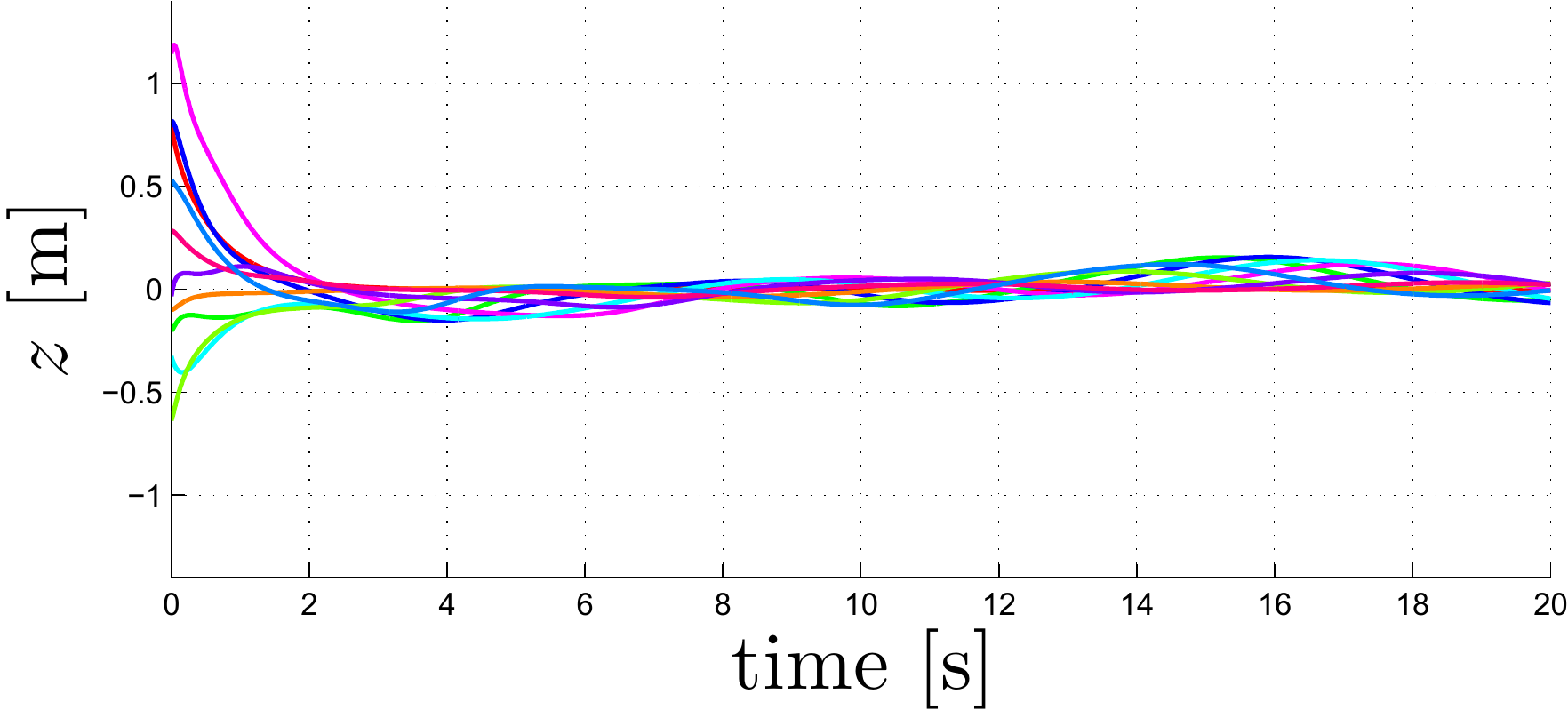}}}
\def\FigSimCTrajXY{{\includegraphics[width=0.66\columnwidth]{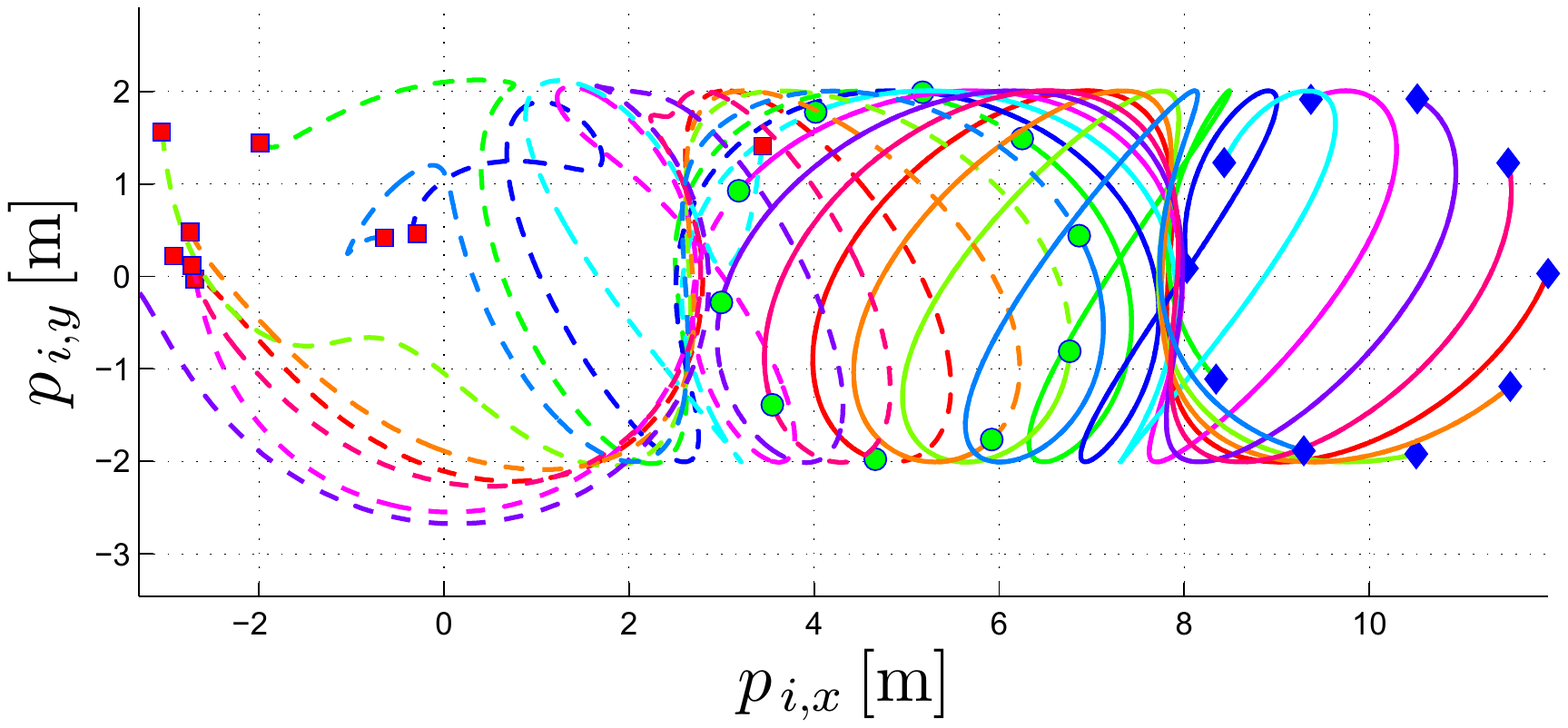}}}
\def\FigSimCTrajXZ{{\includegraphics[width=0.66\columnwidth]{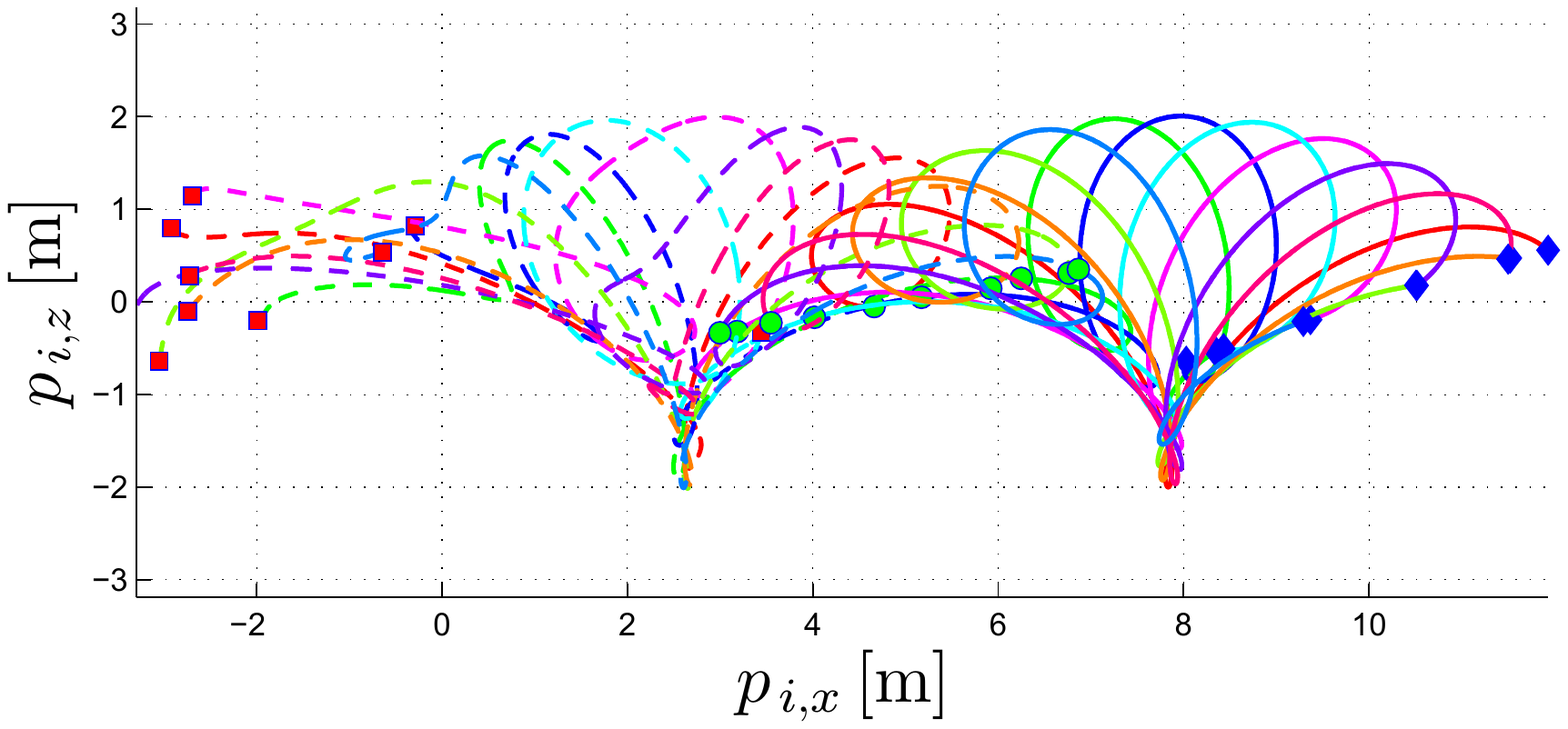}}}
\def\FigSimCTrajYZ{{\includegraphics[width=0.66\columnwidth]{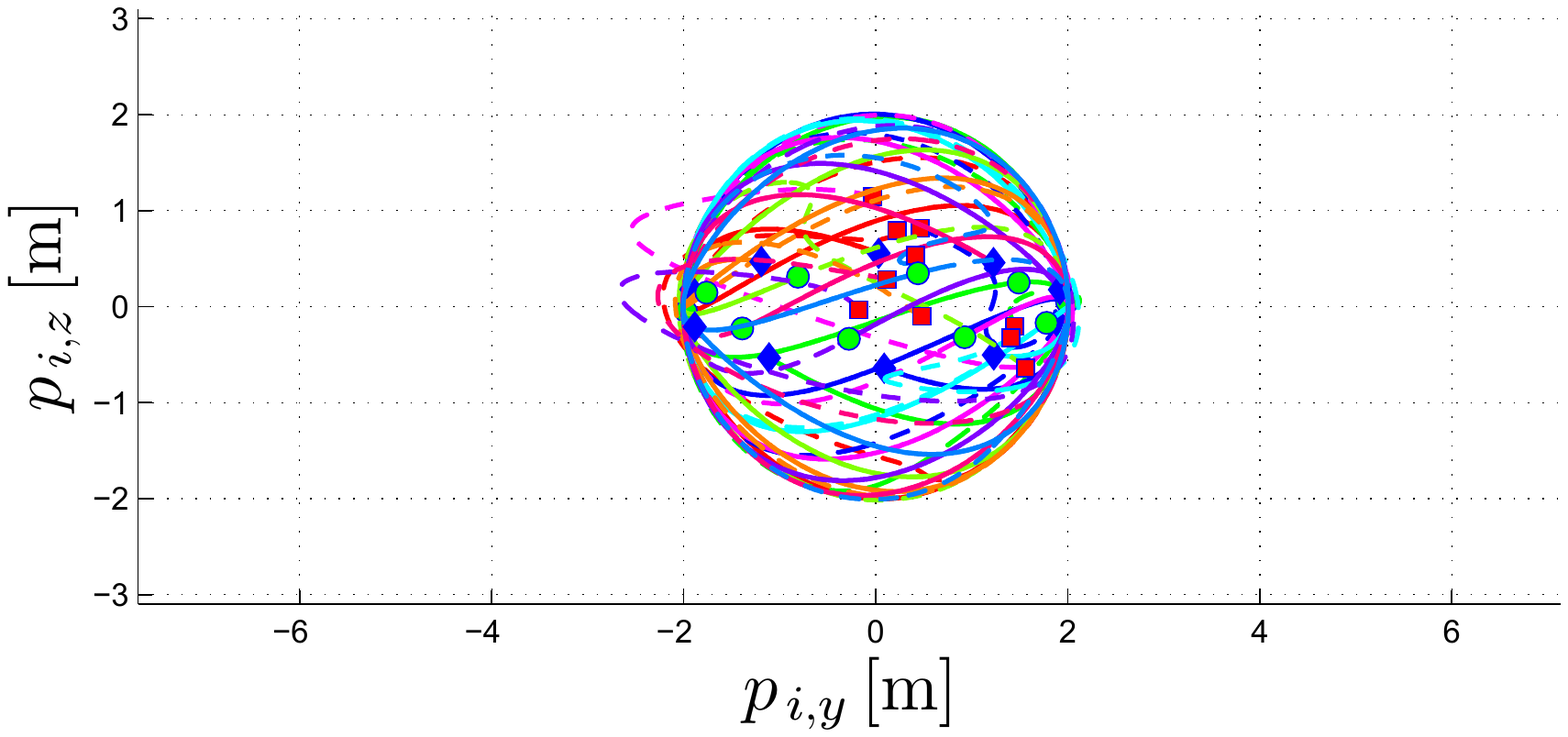}}}

\def\FigSimCOTrajXY{{\includegraphics[width=0.9\columnwidth]{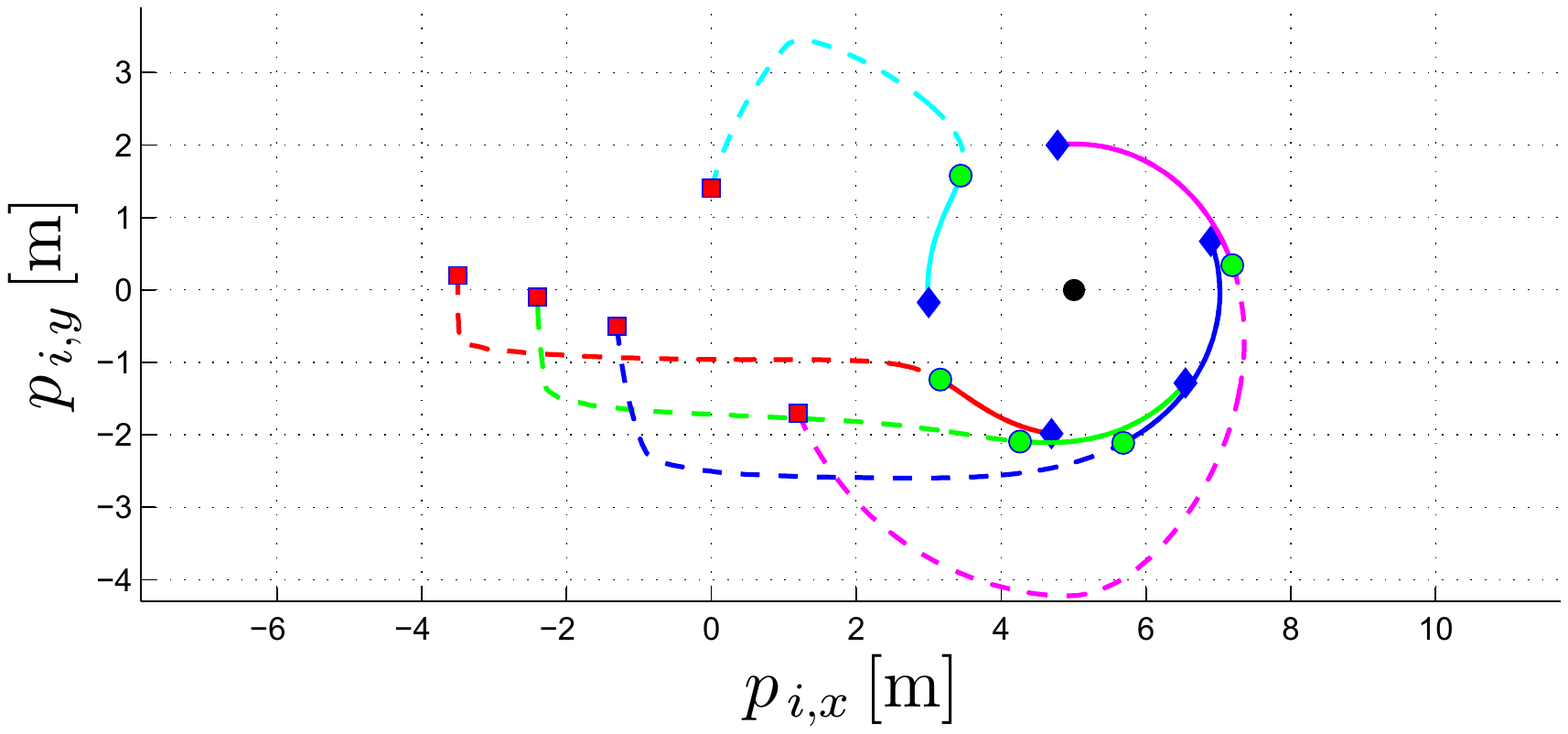}}}
\def\FigSimCODistances{{\includegraphics[width=0.9\columnwidth]{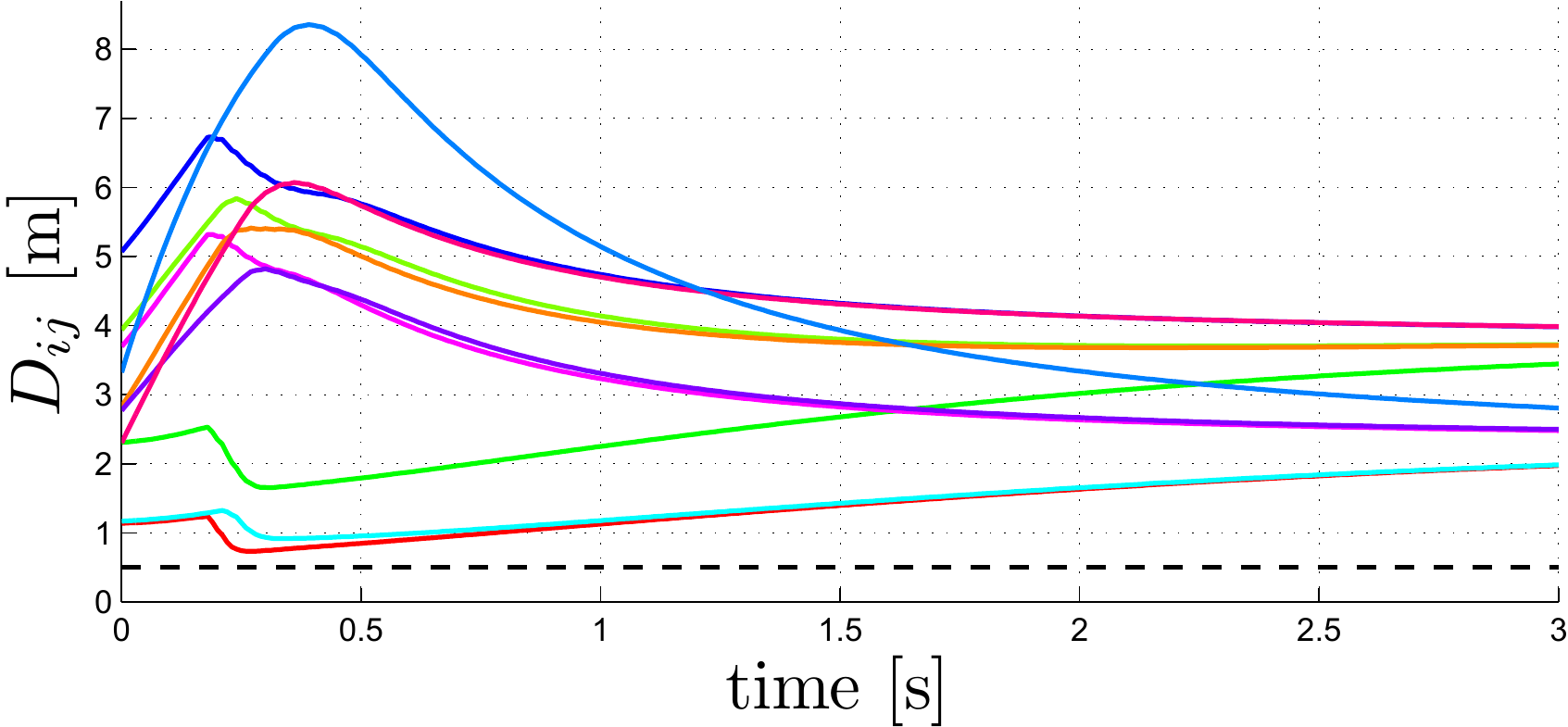}}}
\def\FigSimNOCOTrajXY{{\includegraphics[width=0.9\columnwidth]{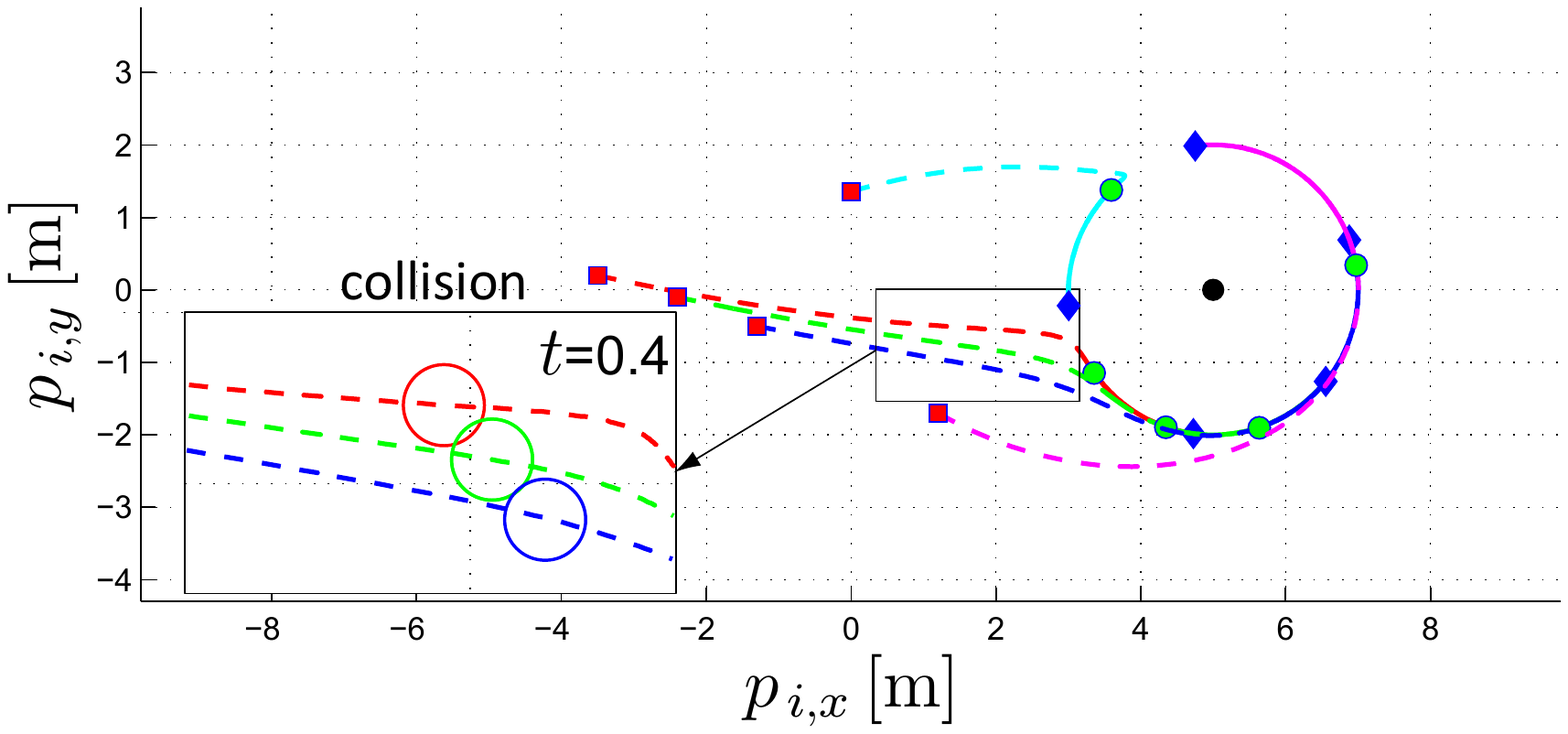}}}
\def\FigSimNOCODistances{{\includegraphics[width=0.9\columnwidth]{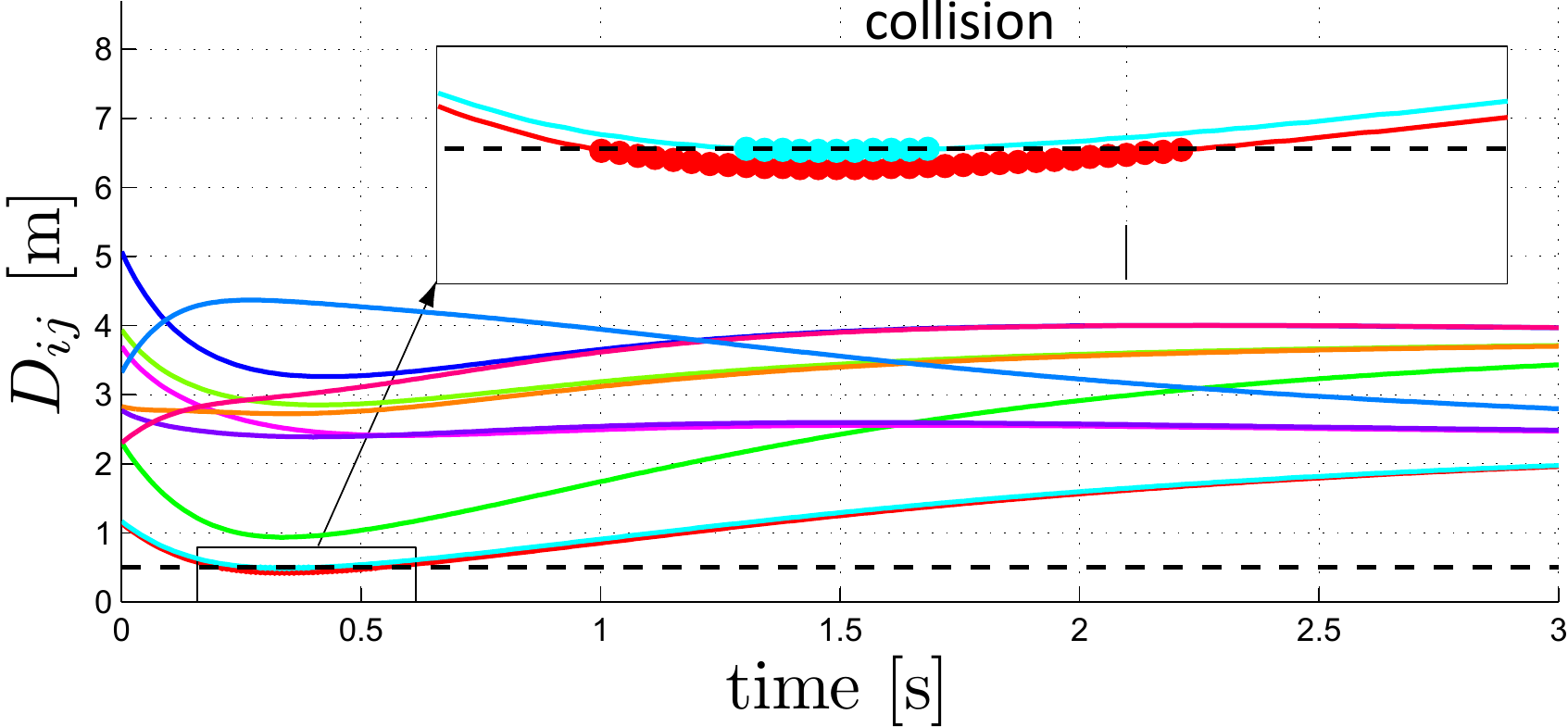}}}

\def\FigSimDAErrRho{{\includegraphics[width=0.59\columnwidth]{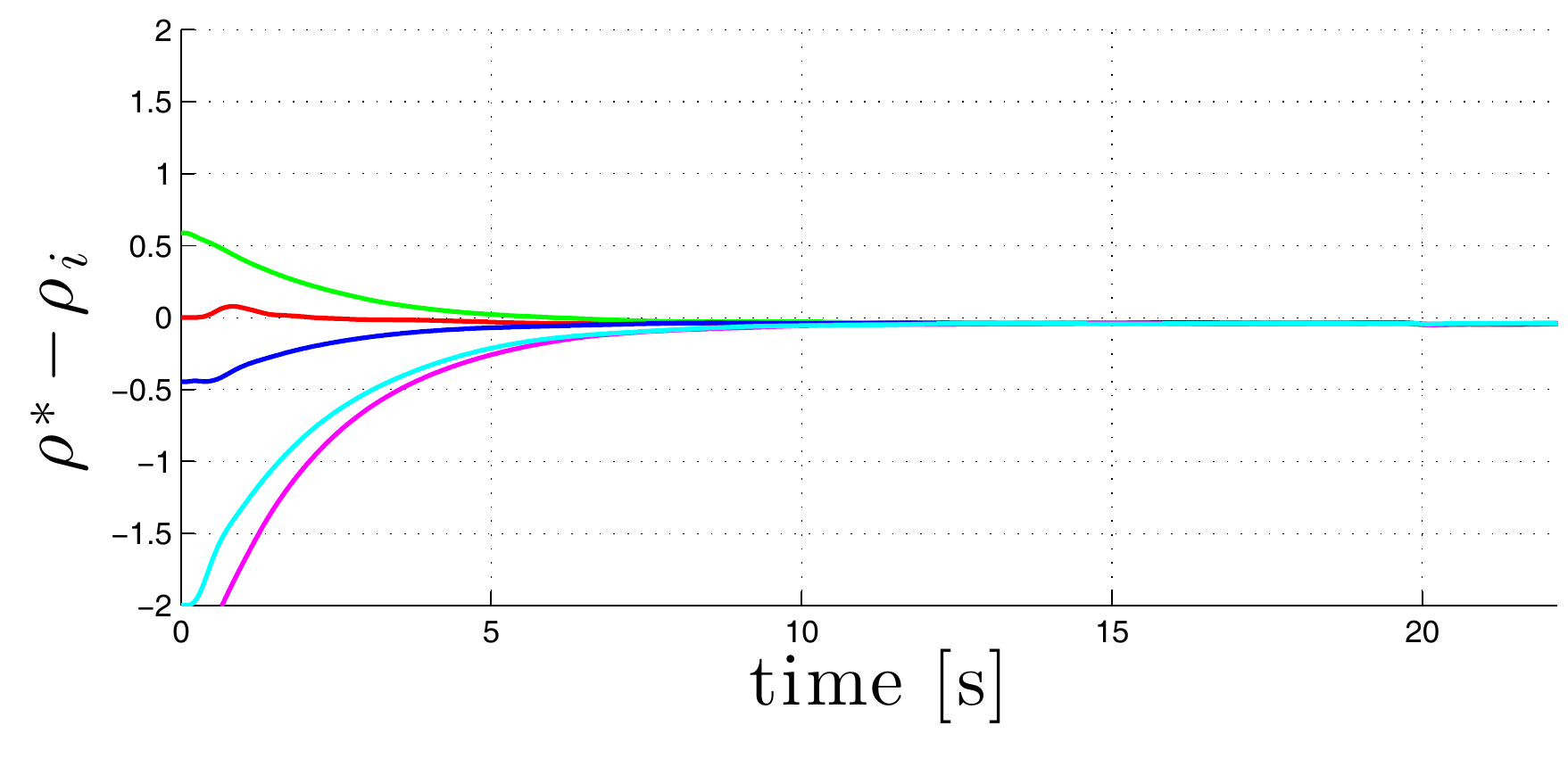}}}
\def\FigSimDAErrPhi{{\includegraphics[width=0.59\columnwidth]{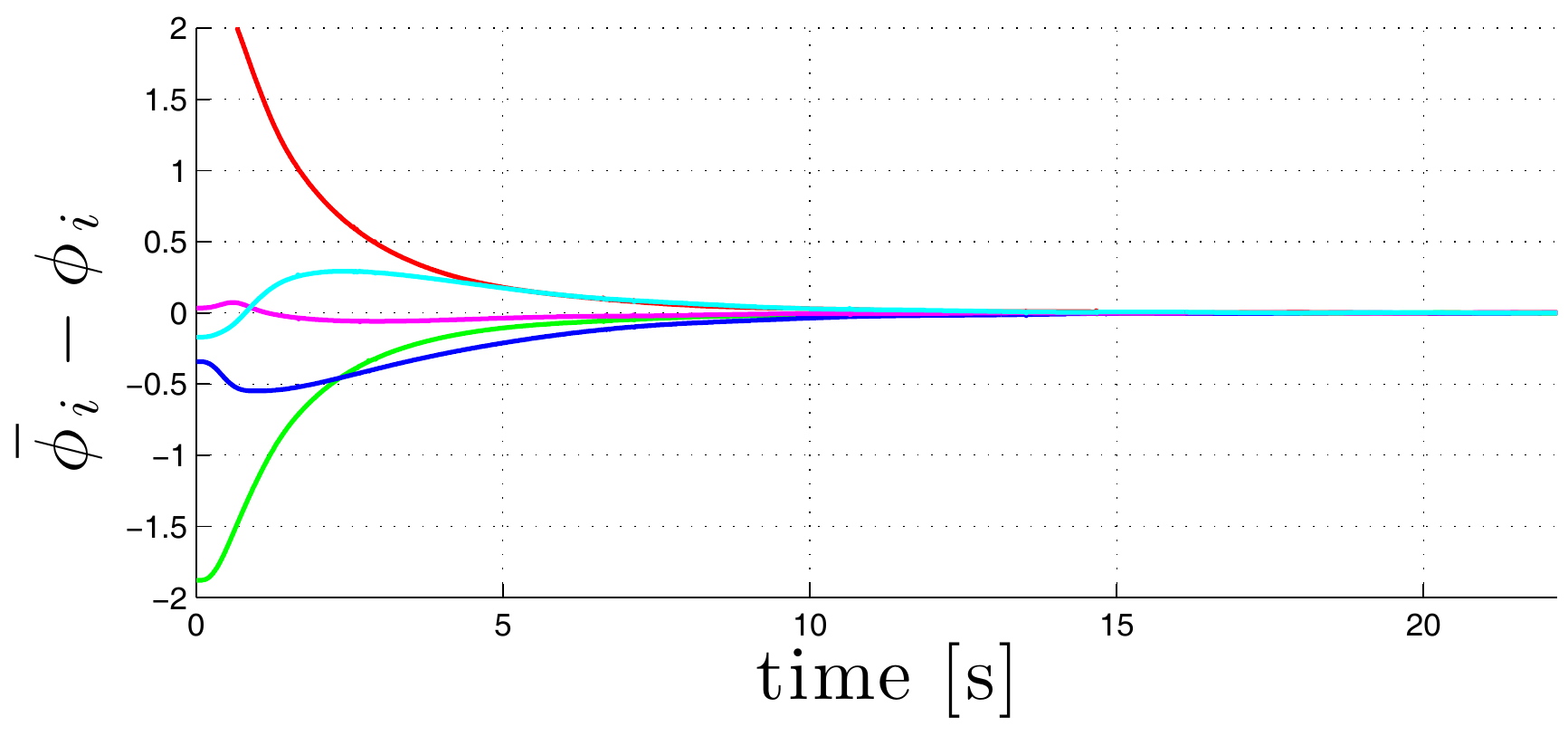}}}
\def\FigSimDAErrDotPhi{{\includegraphics[width=0.59\columnwidth]{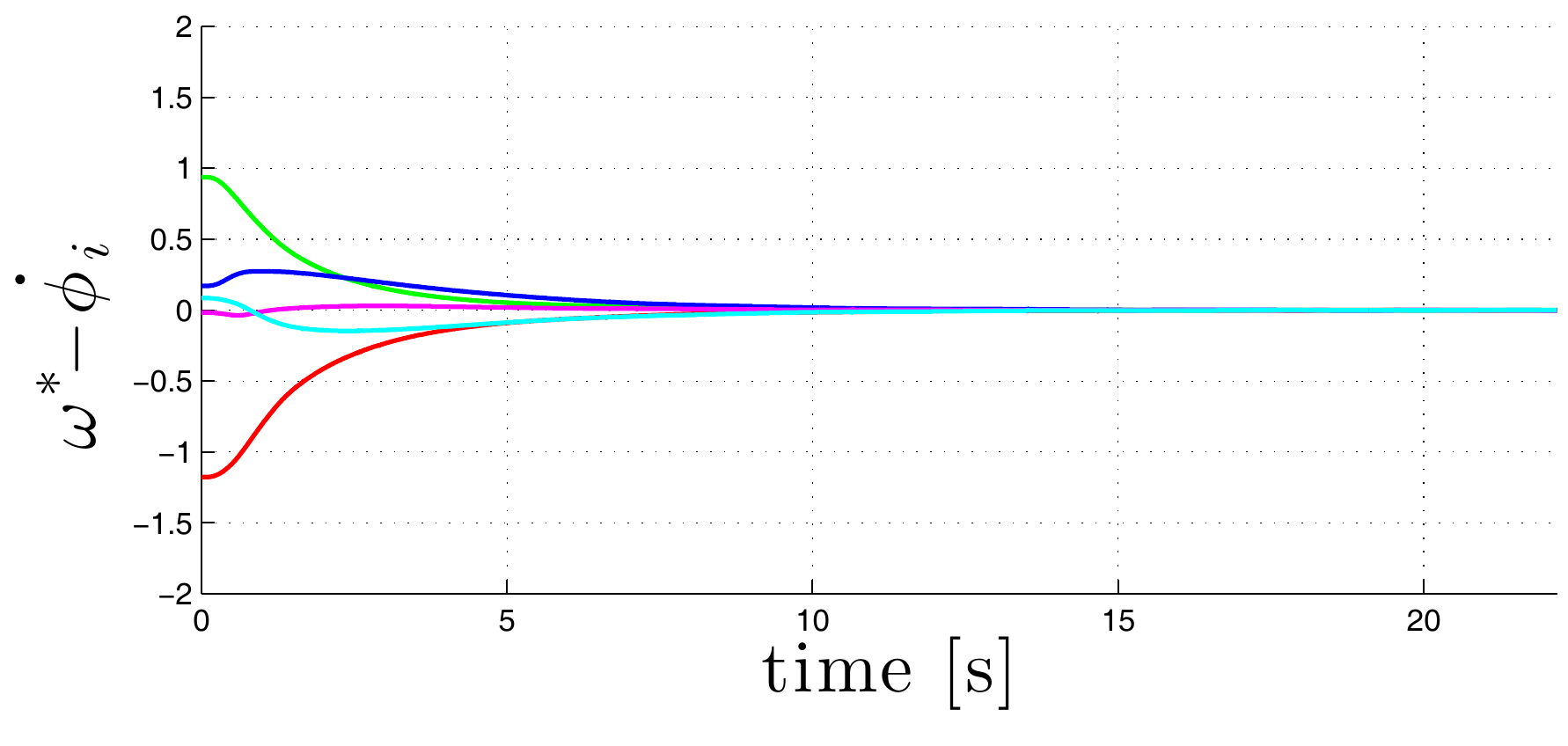}}}
\def\FigSimDAErrZ{{\includegraphics[width=0.58\columnwidth]{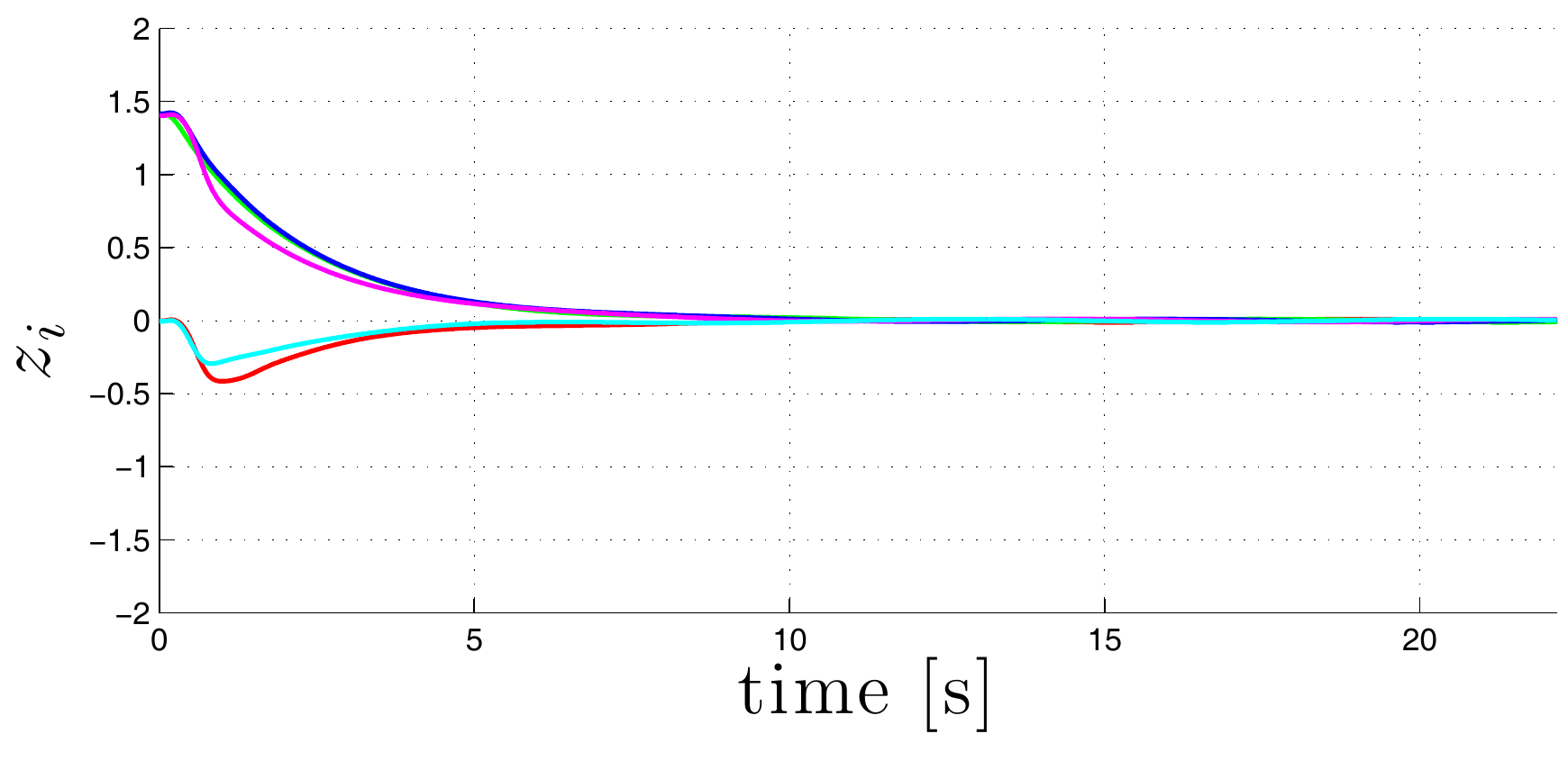}}}

\def\FigSimDADotPT{{\includegraphics[width=0.59\columnwidth]{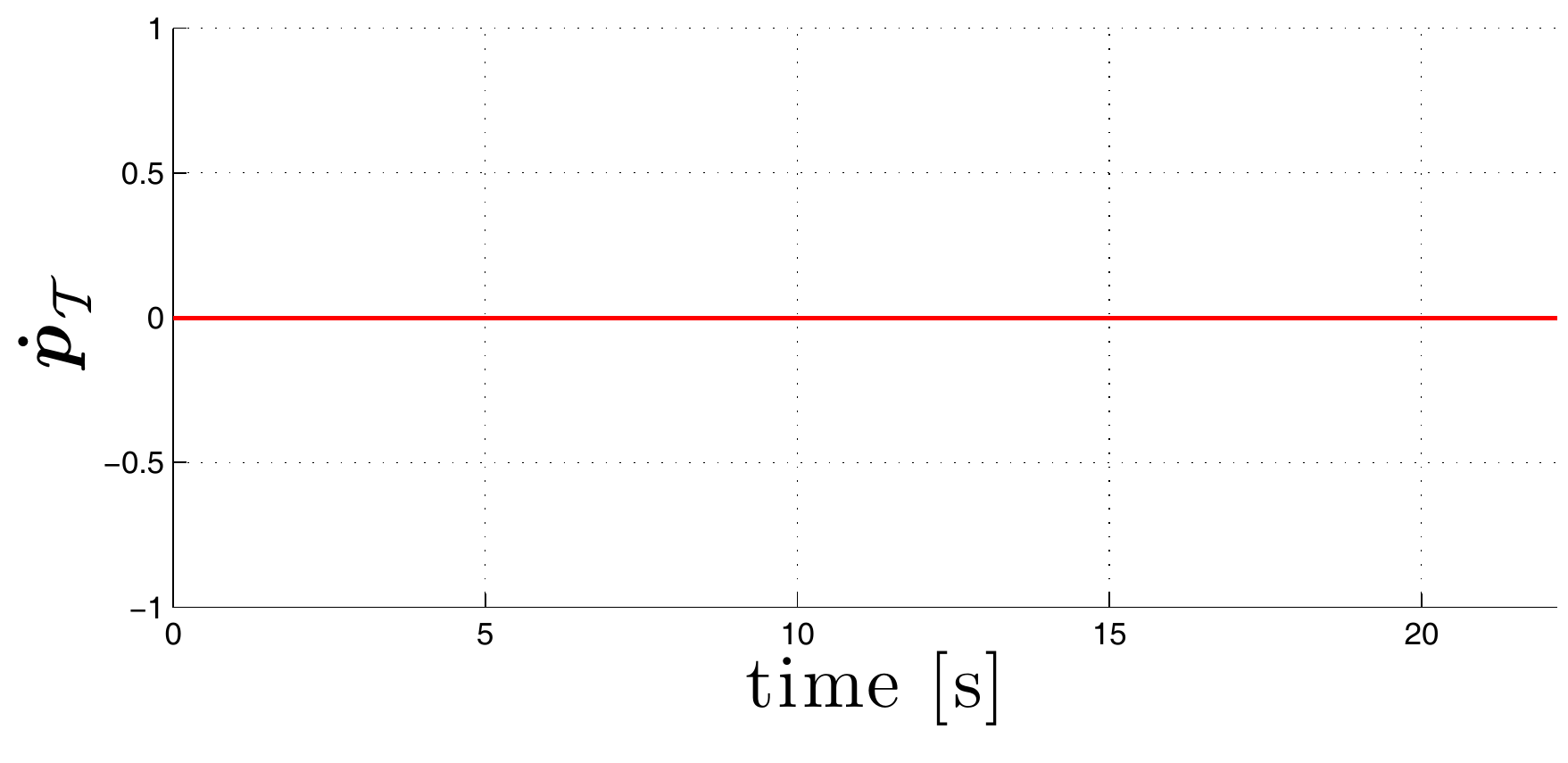}}}
\def\FigSimDAOmegaT{{\includegraphics[width=0.59\columnwidth]{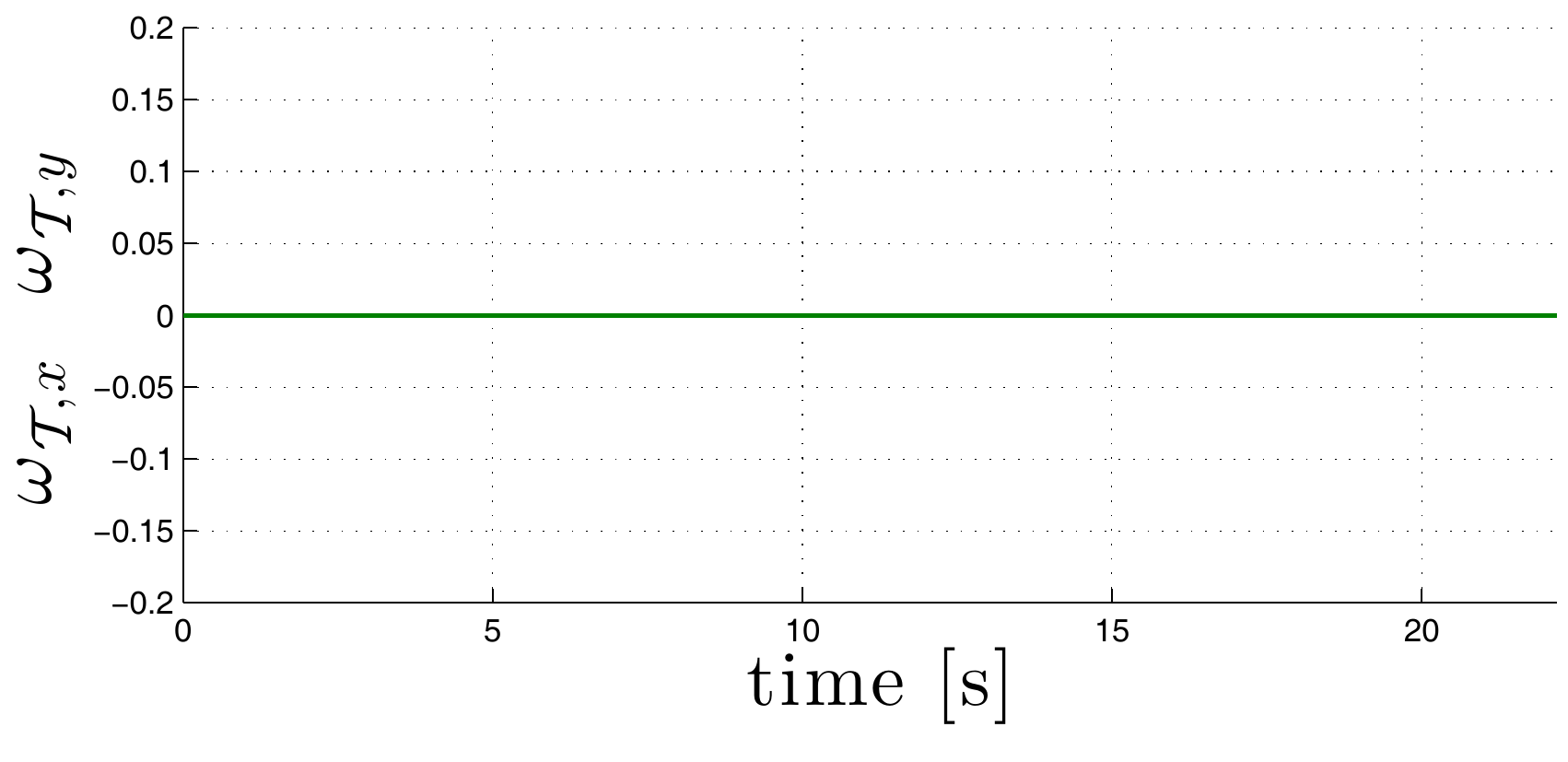}}}

\def\FigSimDEErrRho{{\includegraphics[width=0.59\columnwidth]{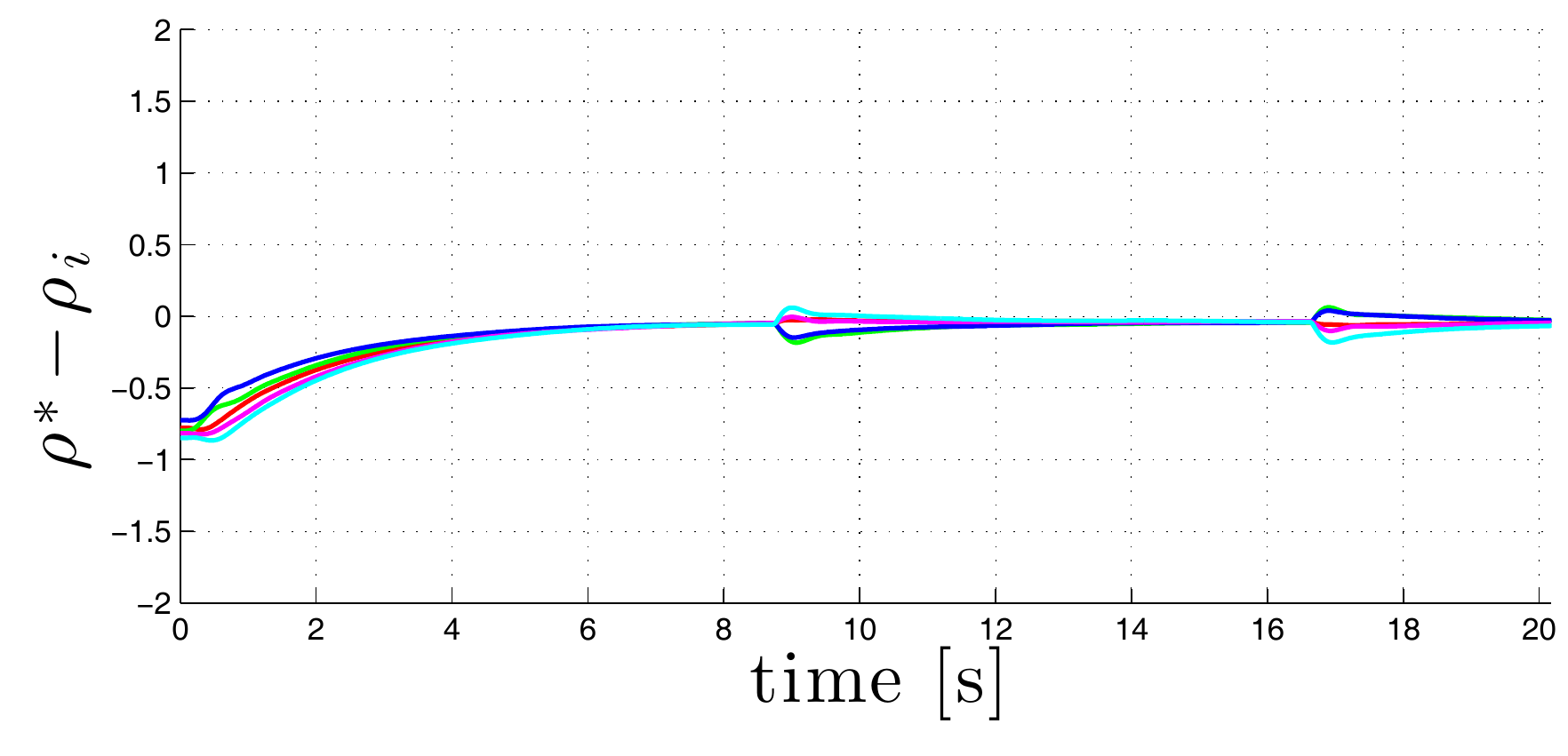}}}
\def\FigSimDEErrPhi{{\includegraphics[width=0.59\columnwidth]{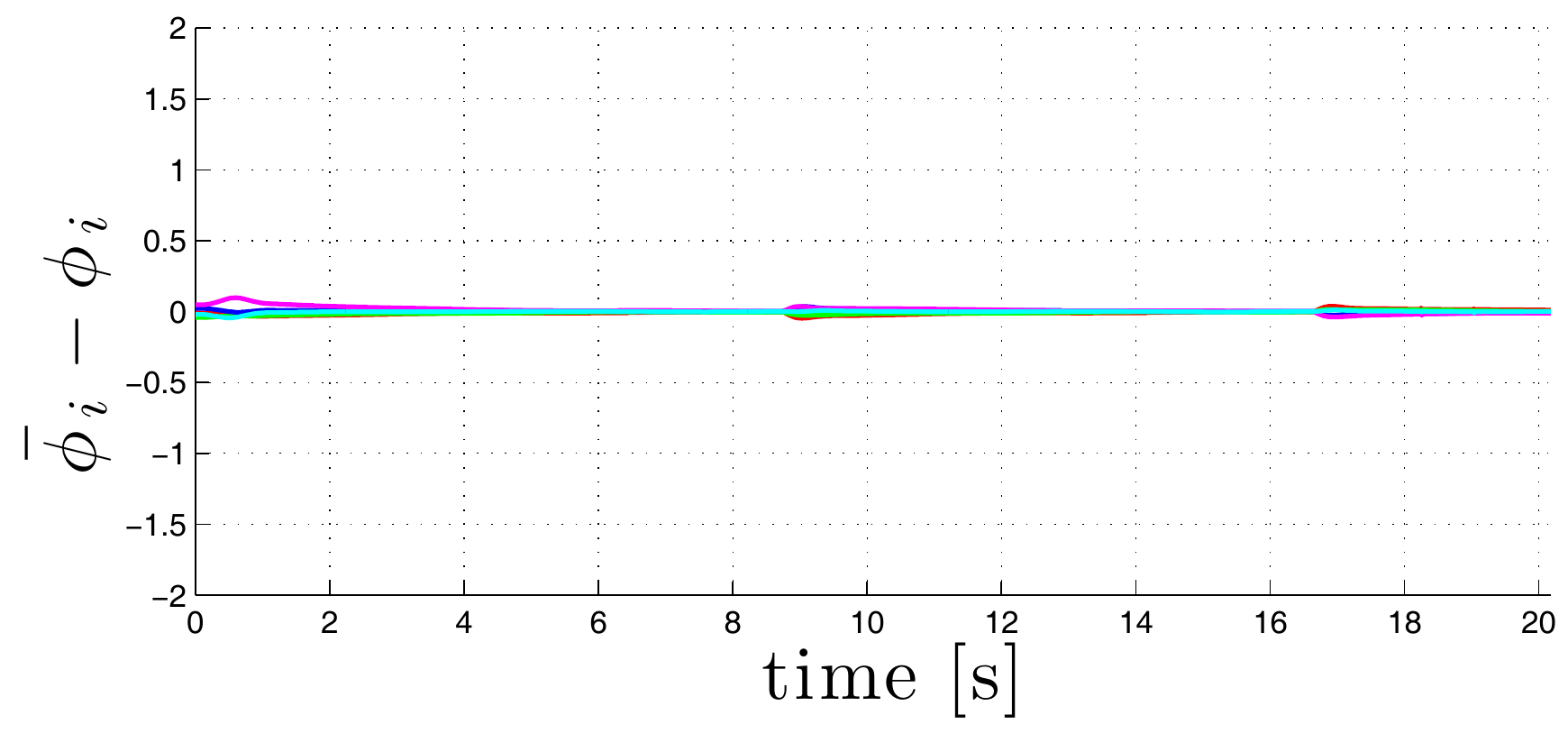}}}
\def\FigSimDEErrDotPhi{{\includegraphics[width=0.59\columnwidth]{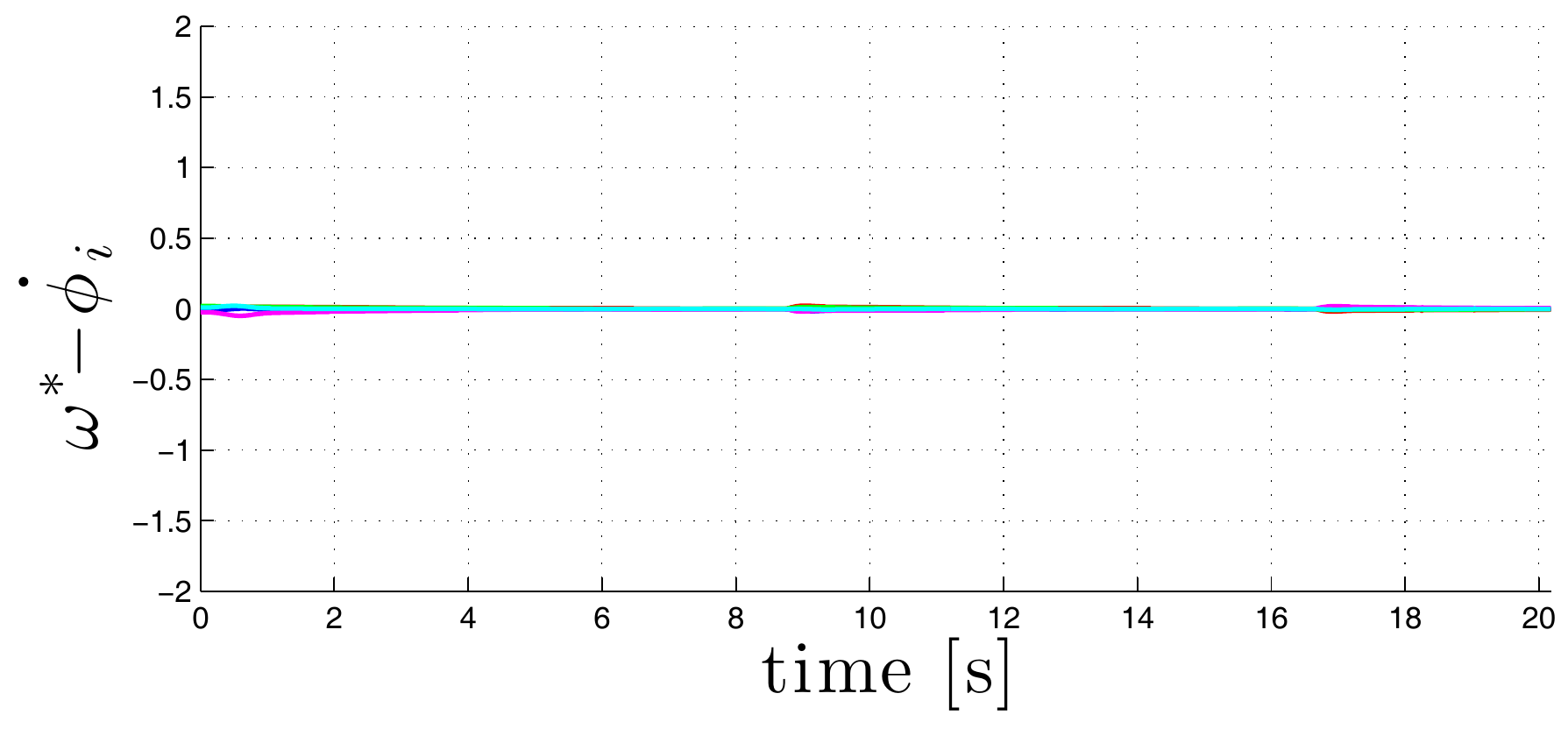}}}
\def\FigSimDEErrZ{{\includegraphics[width=0.58\columnwidth]{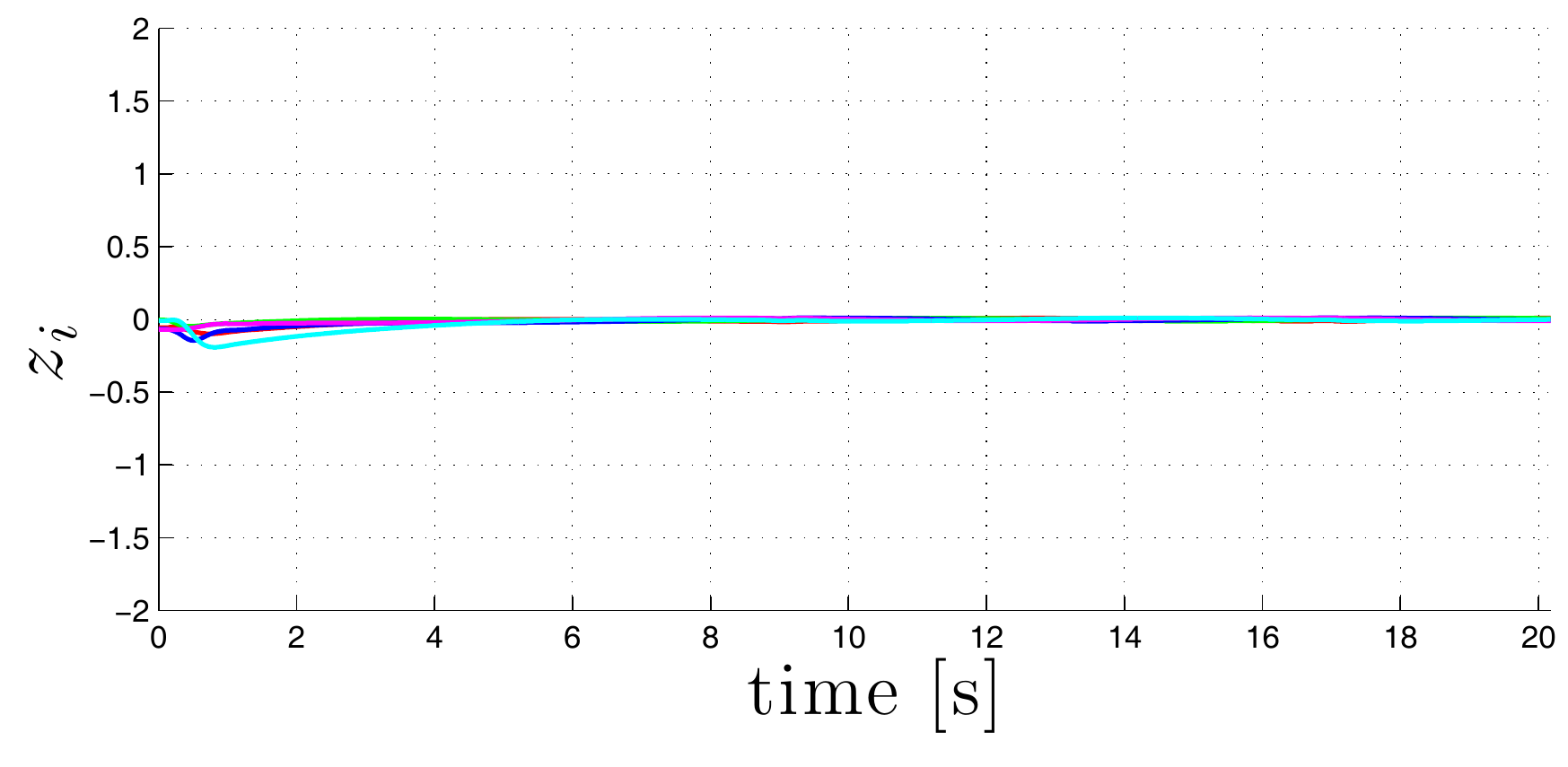}}}

\def\FigSimDEDotPT{{\includegraphics[width=0.59\columnwidth]{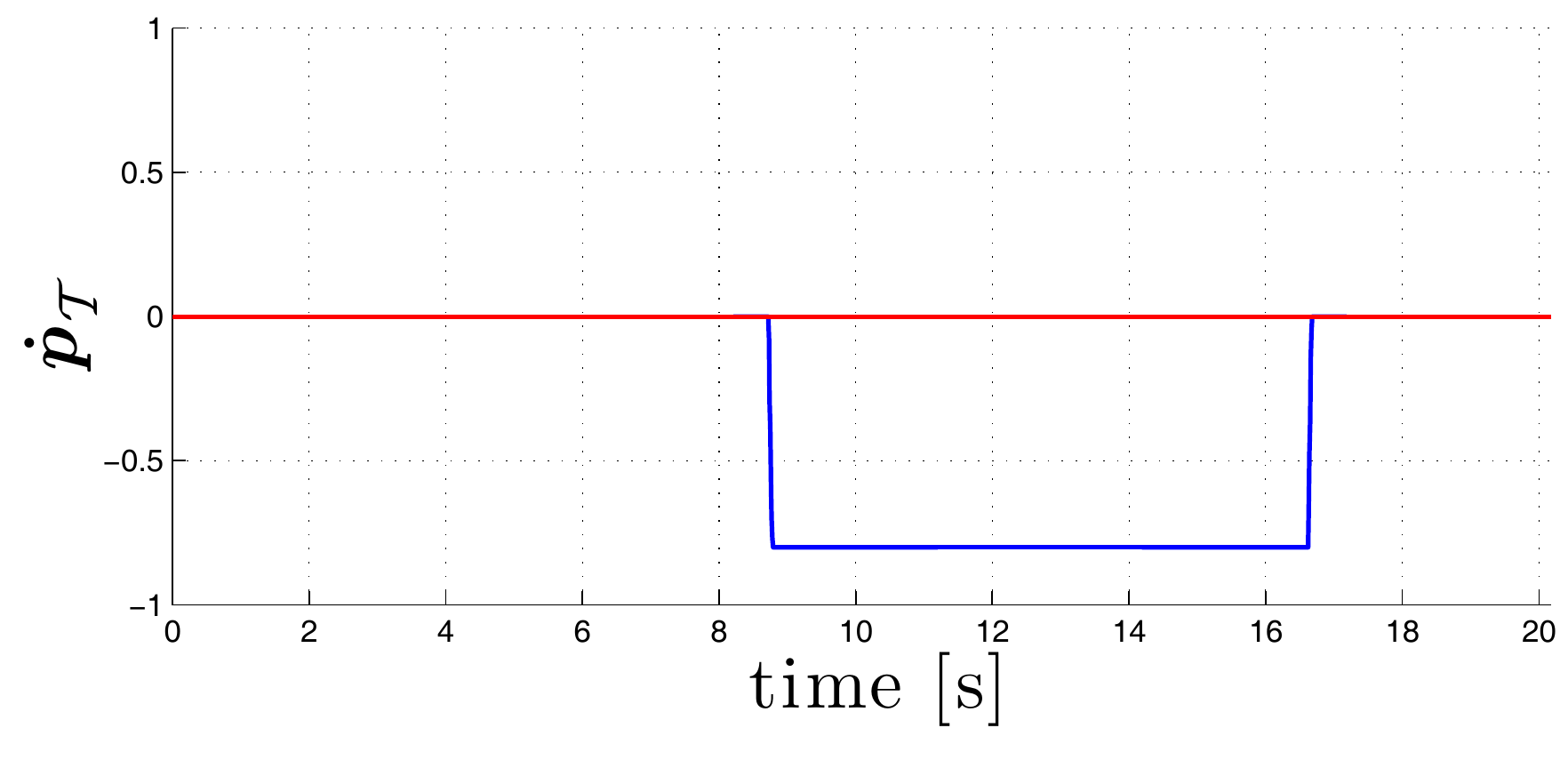}}}
\def\FigSimDEOmegaT{{\includegraphics[width=0.59\columnwidth]{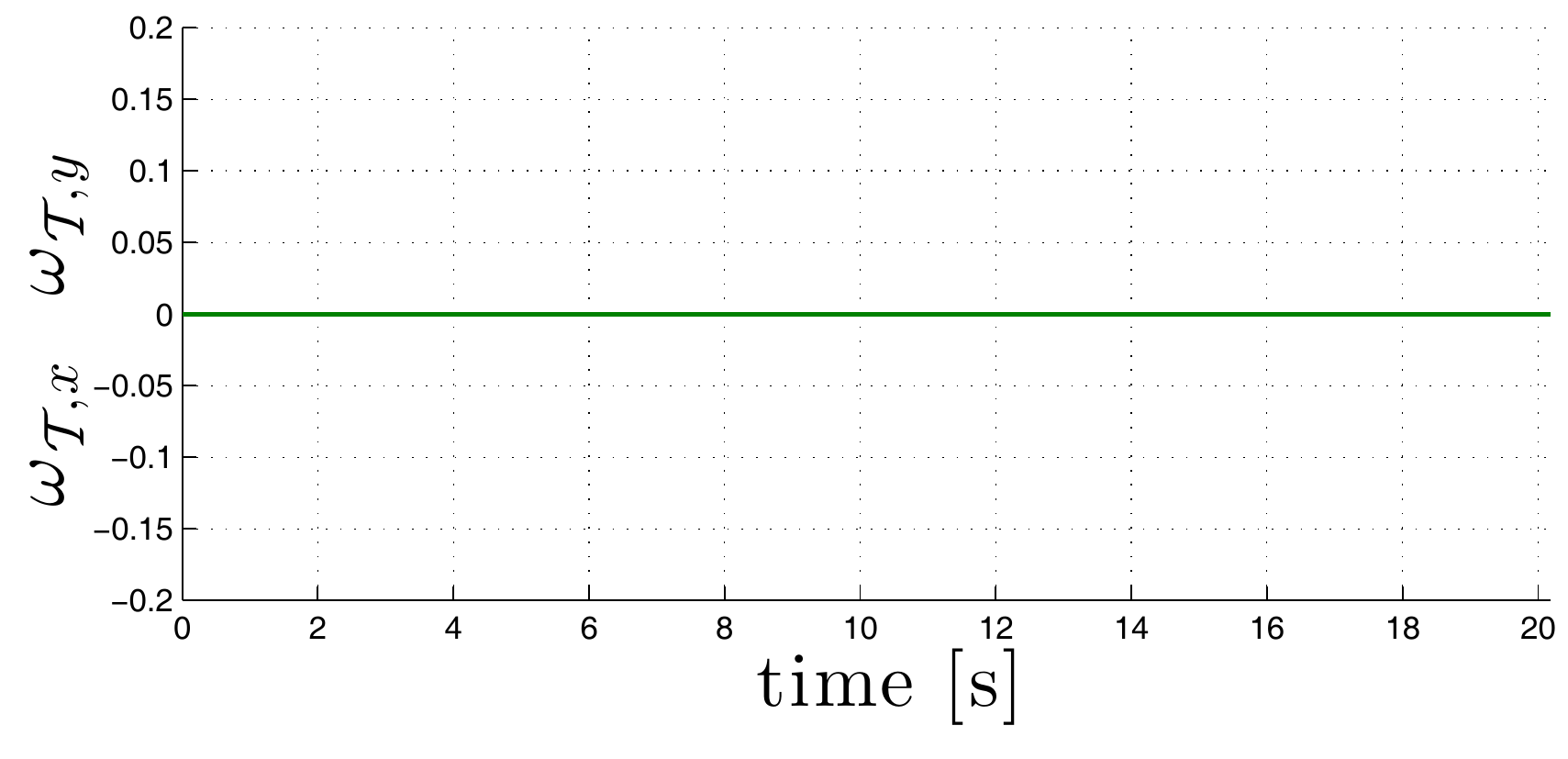}}}

\def\FigSimDDErrRho{{\includegraphics[width=0.59\columnwidth]{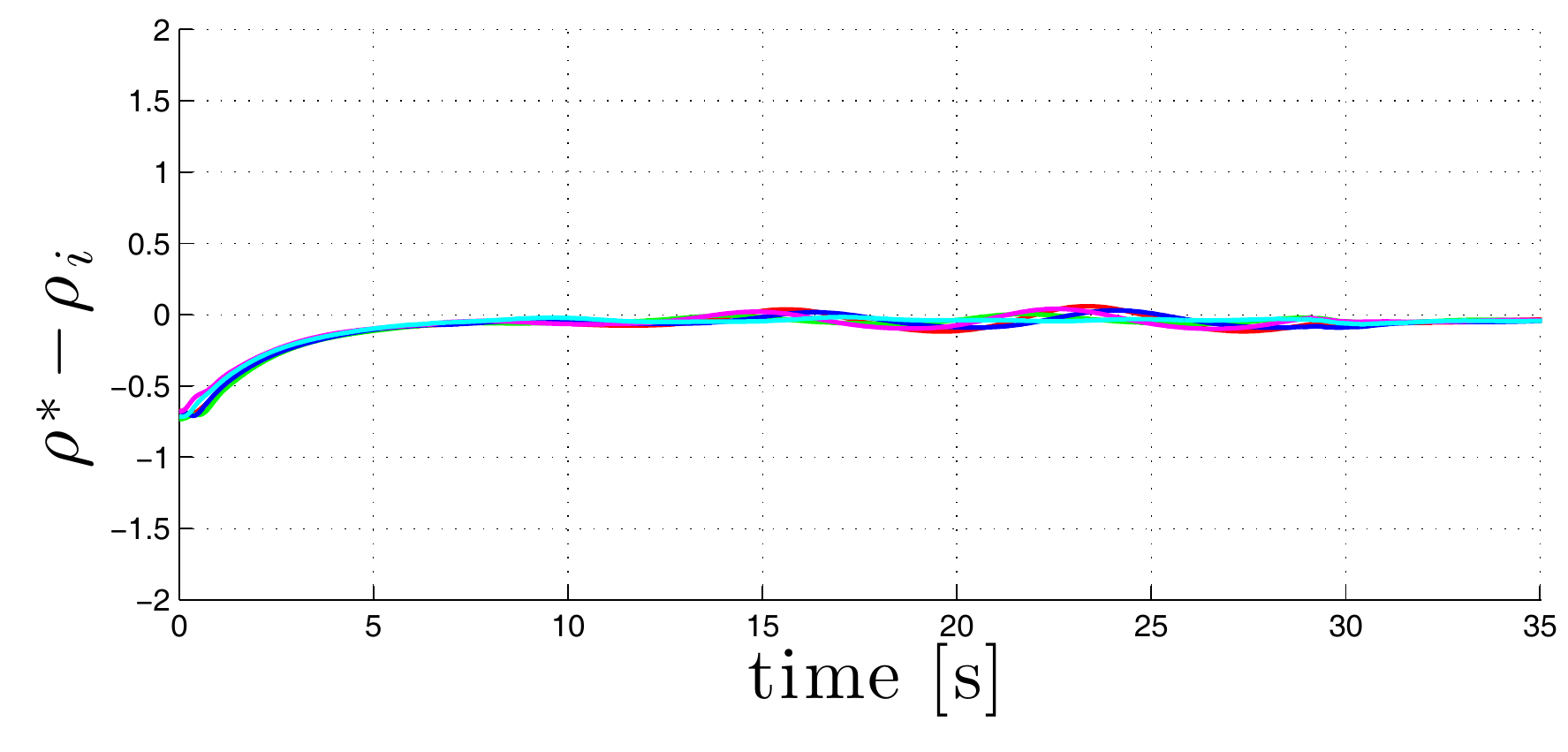}}}
\def\FigSimDDErrPhi{{\includegraphics[width=0.59\columnwidth]{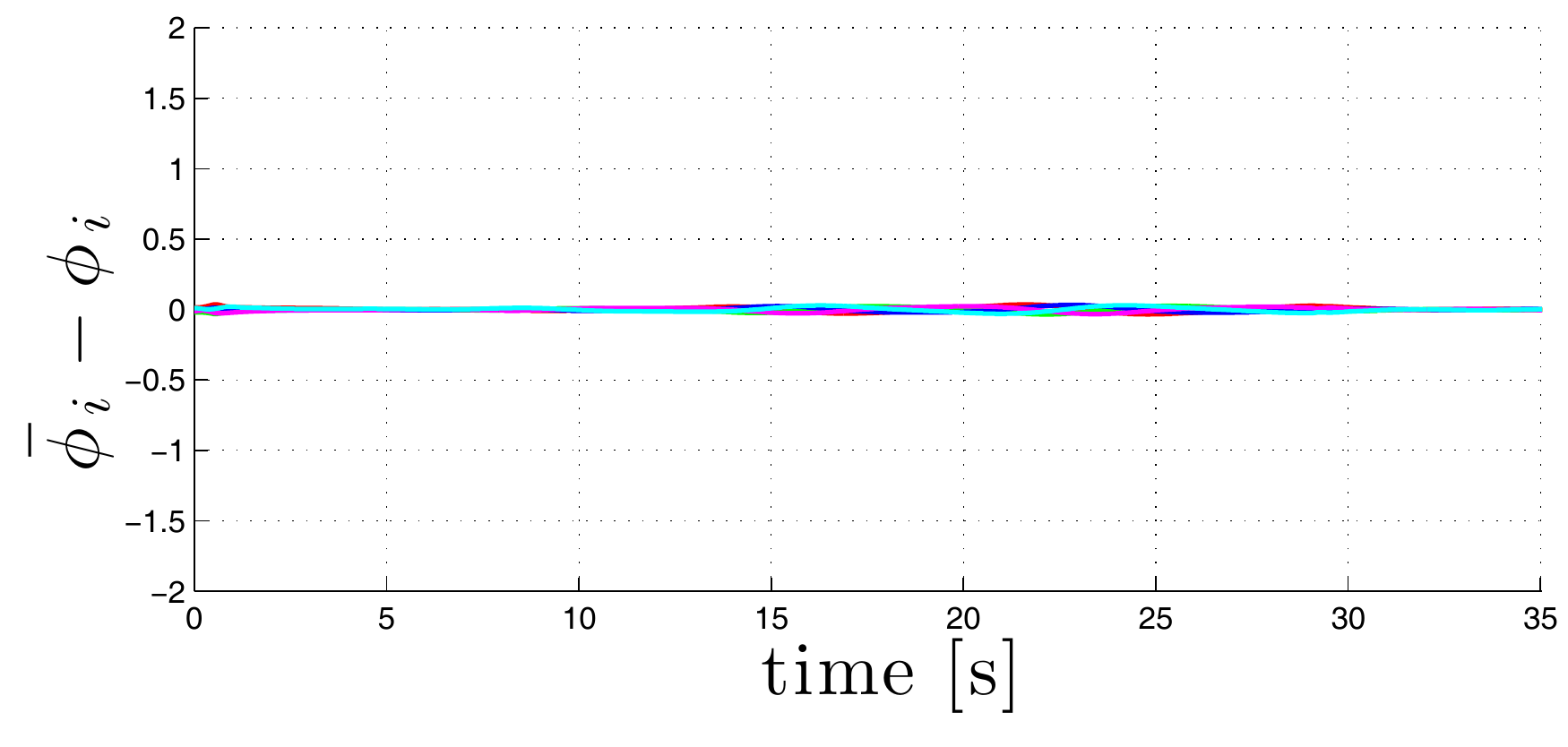}}}
\def\FigSimDDErrDotPhi{{\includegraphics[width=0.59\columnwidth]{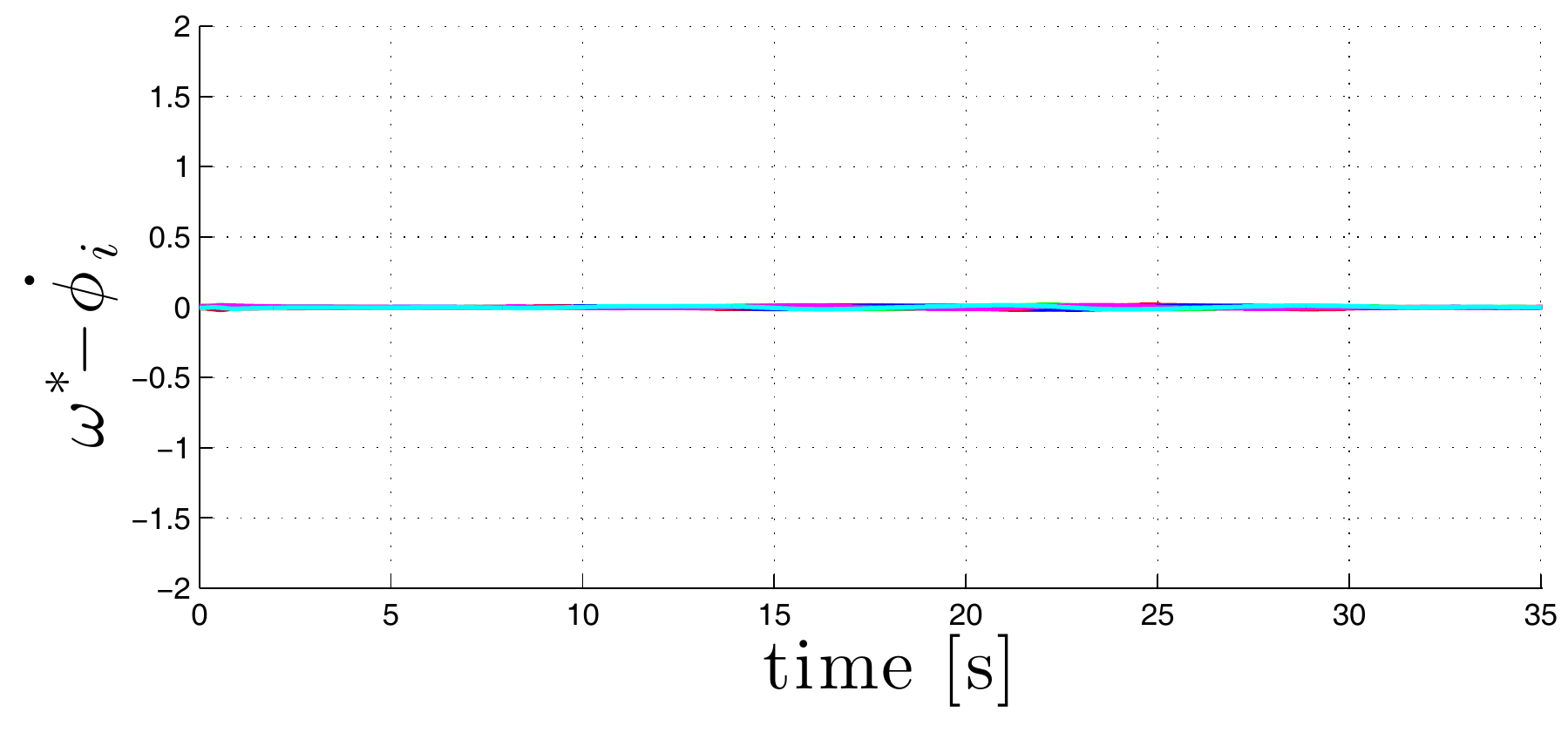}}}
\def\FigSimDDErrZ{{\includegraphics[width=0.58\columnwidth]{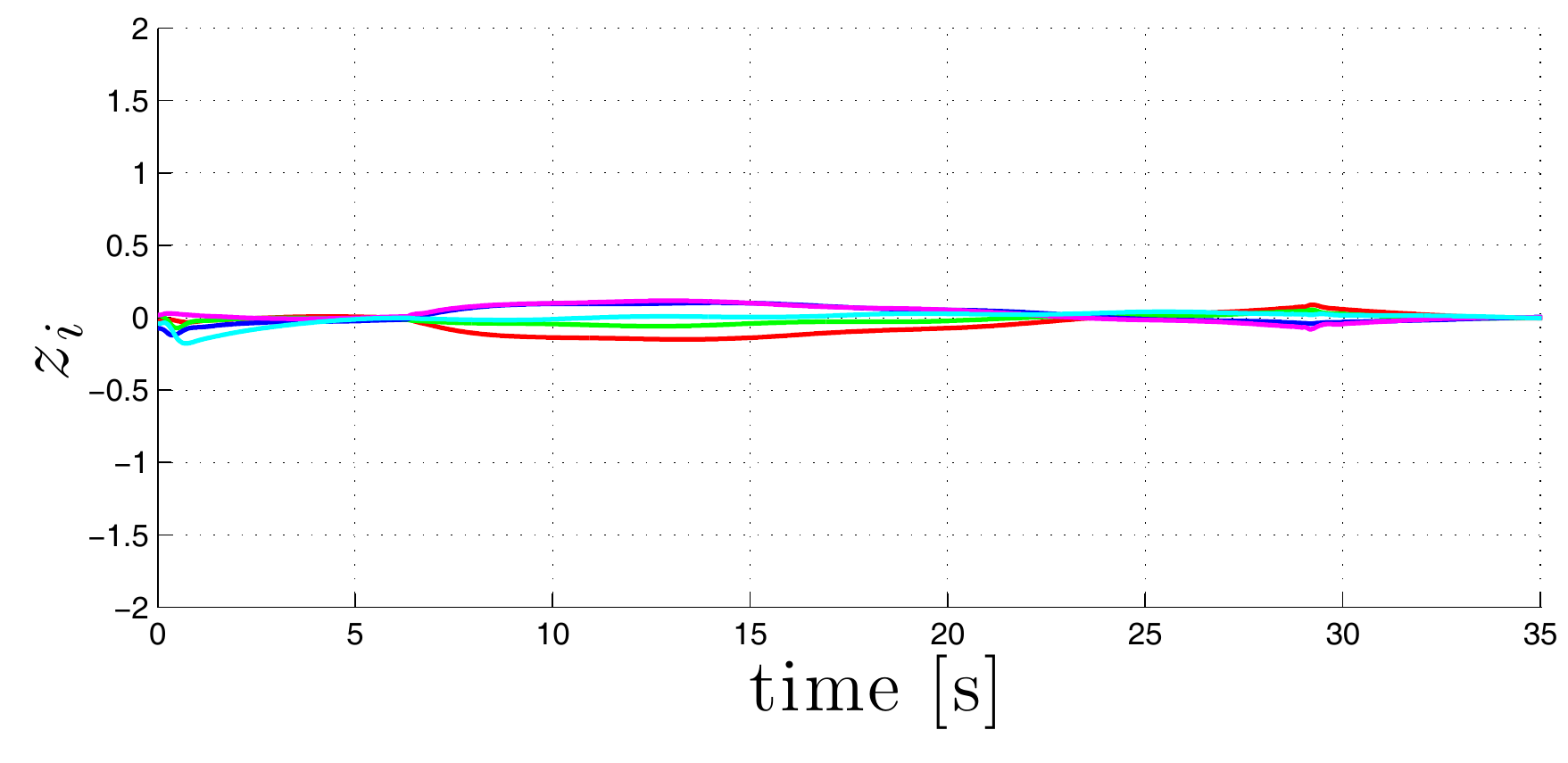}}}

\def\FigSimDDDotPT{{\includegraphics[width=0.59\columnwidth]{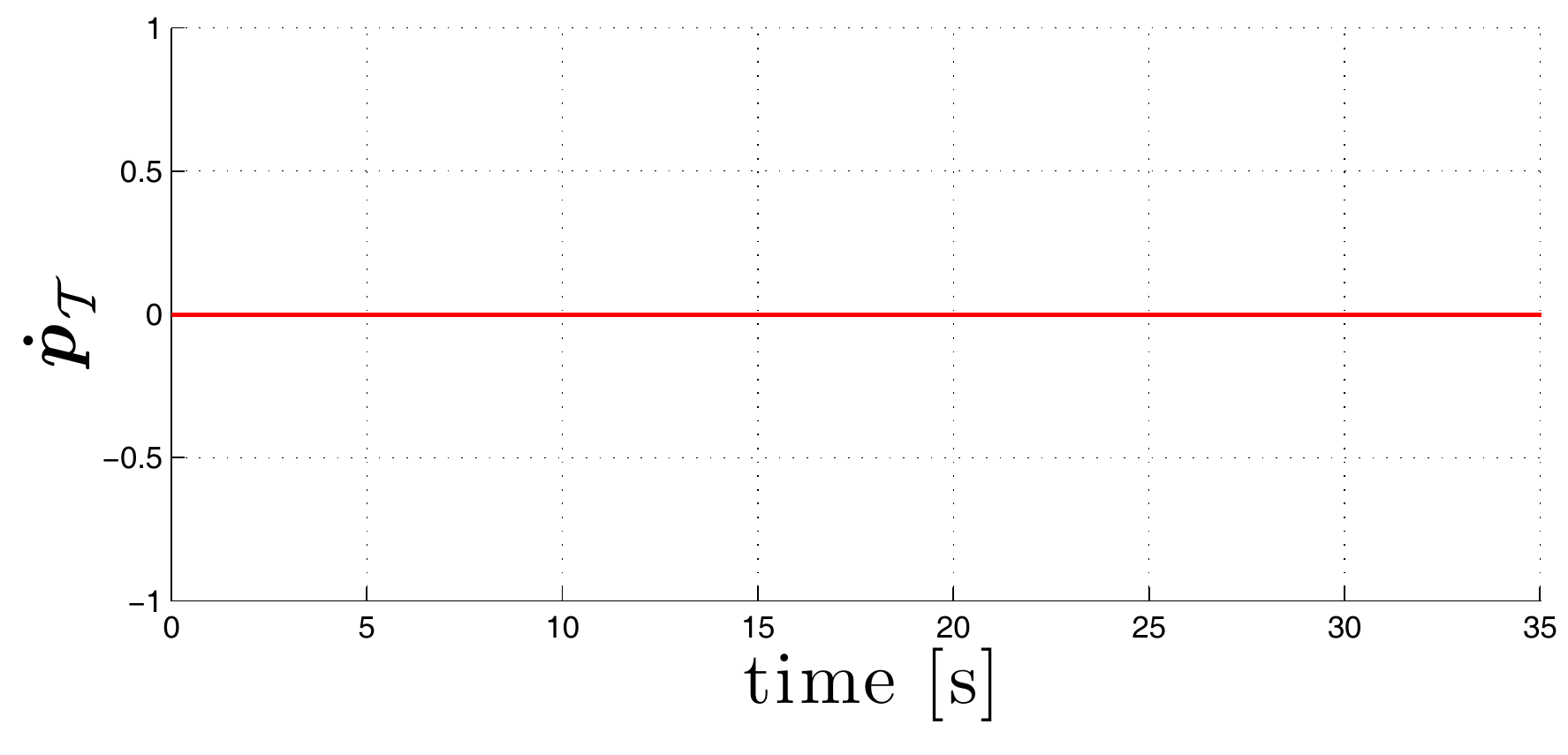}}}
\def\FigSimDDOmegaT{{\includegraphics[width=0.59\columnwidth]{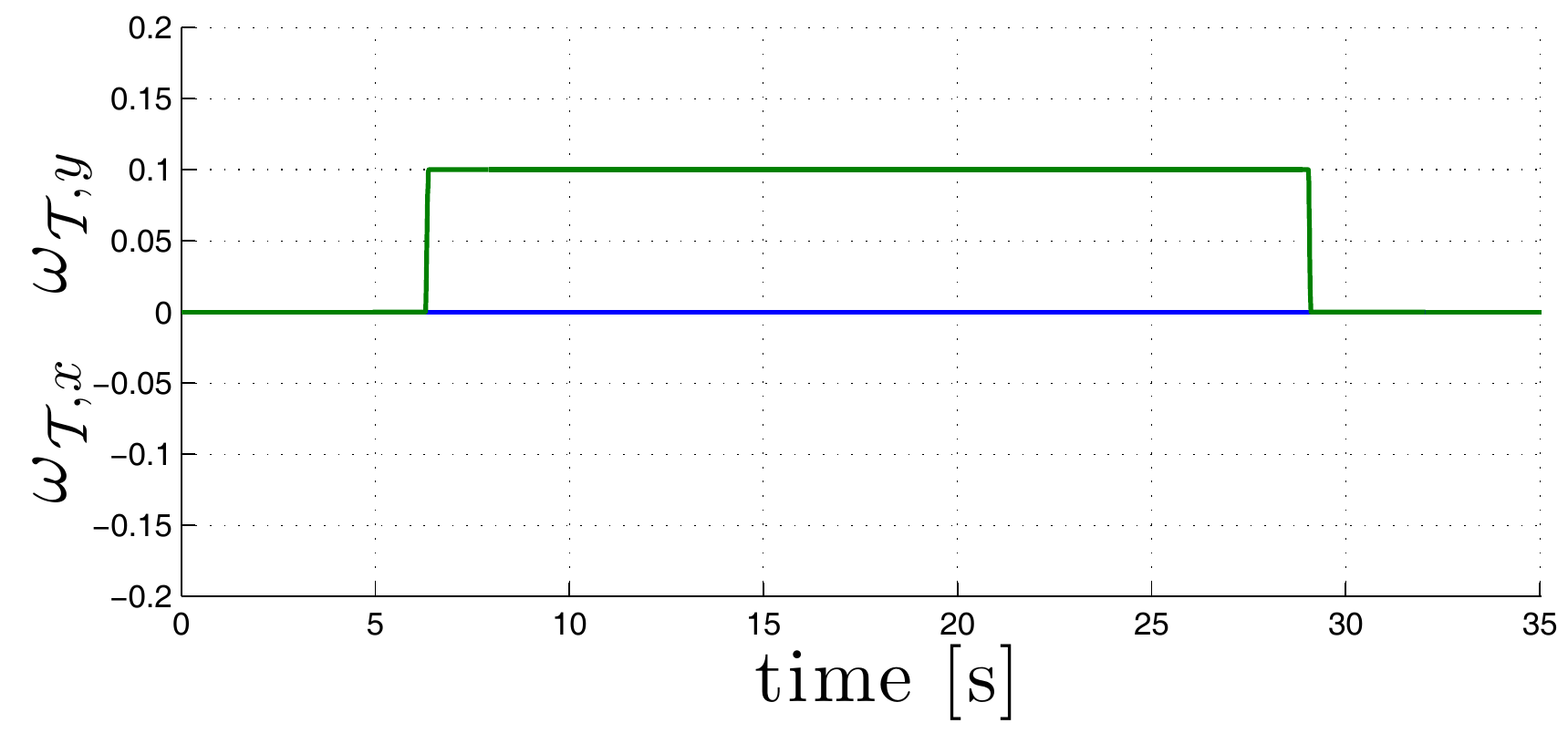}}}

\def\FigSimDScreenShots{{\includegraphics[width=1.99\columnwidth]{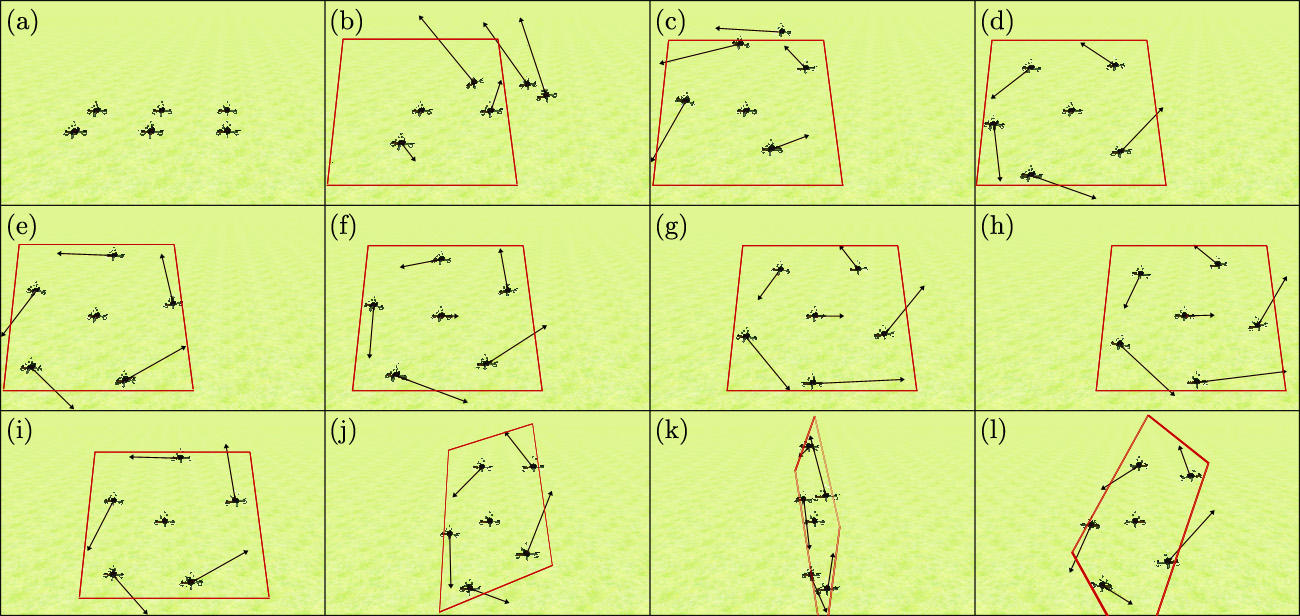}}}

\makeatletter
\def\ps@IEEEtitlepagestyle{
	\def\@oddfoot{\textcolor{red}{\sf\footnotesize Preprint - final, definitive version available at http://www.springer.com/ ~ ~ ~ ~ \thepage\hfill accepted for Autonomous Robots, June 2015}}
	\def\@evenfoot{}
	\def\@oddhead{\textcolor{red}{\sf\footnotesize Preprint version - final, definitive version available at http://www.springer.com/ \hfill accepted for Autonomous Robots, June 2015}}
	\def\@evenhead{\textcolor{red}{}}%
}%
\makeatother
\makeatletter
\def\ps@titlepagestyle{
	\def\@oddfoot{\textcolor{red}{\sf\footnotesize Preprint - final, definitive version available at http://www.springer.com/ ~ ~ ~ ~ \thepage\hfill accepted for Autonomous Robots, June 2015}}
	\def\@evenfoot{}
	\def\@oddhead{\textcolor{red}{}}
	\def\@evenhead{\textcolor{red}{}}%
}%
\makeatother
\makeatletter
\def\ps@headings{
	\def\@oddfoot{\textcolor{red}{\sf\footnotesize Preprint - final, definitive version available at http://www.springer.com/ ~ ~ ~ ~ \thepage\hfill accepted for Autonomous Robots, June 2015}}
	\def\@evenfoot{\textcolor{red}{\sf\footnotesize Preprint - final, definitive version available at http://www.springer.com/ ~ ~ ~ ~ \thepage\hfill accepted for Autonomous Robots, June 2015}}
	\def\@oddhead{}
	\def\@evenhead{}%
}%
\makeatother
\begin{document}\sloppy

\title{Decentralized Multi-Robot Encirclement of a 3D Target\\ with Guaranteed Collision Avoidance
}

\author{Antonio Franchi         \and
        Paolo Stegagno \and
        Giuseppe Oriolo 
}

\institute{A.~Franchi \at
              CNRS, LAAS,
7 Avenue du Colonel Roche, F-31400 Toulouse, France. \\
              Univ de Toulouse, LAAS, F-31400 Toulouse, France\\
              \email{antonio.franchi@laas.fr}           
           \and
           P.~Stegagno \at
           Max Planck Institute for Biological Cybernetics, Spemanstra{\ss}e 38, 72076 T\"ubingen, Germany
           \email{paolo.stegagno@tuebingen.mpg.de}
           \and
           G.~Oriolo \at
           Dipartimento di Ingegneria Informatica, Automatica e Gestionale, Sapienza Universit\`a di Roma, Via Ariosto 25, 00185 Roma, Italy
           \email{oriolo@diag.uniroma1.it}
}

\date{Received: date / Accepted: date}
\thispagestyle{headings}

\maketitle

\begin{abstract}
We present a control framework for achieving encirclement of a target moving in 3D using a multi-robot system. Three variations of a basic control strategy are proposed for different versions of the encirclement problem, and their effectiveness is formally established. An extension ensuring maintenance of a safe inter-robot distance is also discussed. The proposed framework is fully decentralized and only requires local communication among robots; in particular, each robot locally estimates all the relevant global quantities. We validate the proposed strategy through simulations on kinematic point robots and quadrotor UAVs, as well as  experiments on differential-drive wheeled mobile robots. 
\end{abstract}

\ifx\nointro\undefined
\section{Introduction}\label{sec:intro}

The general problem of steering a group of mobile ag\-ents/robots in a regular and cohesive formation is an important topic in robotics because of the large number of potential applications. As such, it 
has been often considered in the research literature.

Early works on this topic focus on basic tasks such as aggregation and obstacle avoidance.
\cite{2001-LeoFio} use virtual points to change the formation shape and to address collision avoidance. The performance of a swarm that approaches a goal maintaining cohesion is analyzed by~\cite{2002-GazPas} depending on the attractive and repulsive control profiles.
A study of the convergence depending on the topology of the communication graph is considered by \cite{2005-Mor}, while \cite{2005-LinFraMag} place the \emph{$\alpha-$stability} concept at the basis of a fixed-topology algorithm, and \cite{2007b-Ren} applies consensus results to formation control.

\cite{2011-GonPimPer}, \cite{2010-SabSecFanDem,2013-SabSecFan} and~\cite{2008-HsiKumCha} present methods based on artificial potentials or  fields to make a group of robots circulate along a static curve defined by two implicit functions. \cite{2012-TurMicKum} show that it is possible to solve an optimization problem  online (i.e., at each control step) in order to drive a multi-UAV system along a pre-planned reference trajectory.

Many authors have investigated formation controllers for specific tasks.
An important example in this category is the encirclement of a point or target.
Despite its relative simplicity, this task represents several interesting missions, such as {\em coverage} (retrieve and fuse data about an environmental point-of-interest from different viewpoints), {\em patrolling} (guard the perimeter of a given area centered at the encircled point), {\em escorting} (protect an important member of the group from unfriendly agents) and {\em entrapment} (contain the motion of a hostile object).

The problem of moving a group of unicycles in a regular formation around a common point is considered by \cite{2007-SepPalLeo}. A related approach is presented by \cite{2009-MosMicJadDan} where a centralized vision system is used for the experimentation. Encirclement with relative bearing sensors is investigated by~\cite{2008-CecDimGarGia}.
The convergence of a decentralized controller is proven using Lyapunov arguments in \cite{2012-SadKosVanHuiNij}, while other authors develop controllers based on cyclic pursuit-evasion schemes \citep{2007-PavFra,2008-HarKimHor} and consensus techniques \citep{2010-JonKao}.
\cite{2010-LanLin} propose a decentralized hybrid controller for encircling and tracking a moving target using multiple unicycles, under the assumption that the velocity of the target is known.

Some works rely on a centralized approach to the problem.
For example, in \cite{2008-AntArrChi} a global vision system provides the configuration of each robot to a centralized controller based on input-output linearization.
Similarly, \cite{2009-MasLi_AcaKit} propose a centralized system in which  measurements are expressed in a world frame and cluster space control is used. An additive robot-level obstacle avoidance term introduces some decentralization in the system.

In other works, additional challenges (e.g., higher dimensional problem, disturbances) have been introduced w.r.t. the basic encirclement problem. 
\cite{2009-KawTor} prove the stability of a decentralized controller for a multi-robot system moving in a 3D space. 
A different problem is considered by \cite{2010-MelPal}, who design a backstepping controller for stabilizing a circular formation in a uniform flow field.
\cite{2010-ShaFidAnd} present a control law that steers two-dimensional agents on a fixed regular-polygon formation.

In this paper, we introduce a new type of decentralized encirclement controller. In particular, three variations of a basic control strategy are proposed for different versions of the encirclement problem, and their effectiveness is formally established. The most significant contributions of the proposed approach with respect to the literature are the following:

\begin{itemize}

\item the applicability to both 3D (spatial) and 2D (planar) encirclement without modifications;

\item an integrated scheme for estimating in a decentralized way all the global quantities needed by the control law;

\item a provably effective strategy for inter-robot collision avoidance;

\item an extensive numerical validation showing applicability of the method to both holonomic point robots and underactuated UAVs;

\item an experimental implementation on nonholonomic ground vehicles using only onboard sensors (i.e., without any external localization system).

\end{itemize} 

In particular, the last point shows that our approach is viable and robust in a realistic setting, in which each robot must estimate all quantities needed by the encirclement controller using only  information gathered by its own sensory system.
In particular, there is no need for an external tracking system, such as a GPS or a motion capture system.

While some of these properties were individually achieved in previous works with different controllers, we present here a self-contained and comprehensive  work that covers all the aspects of the encirclement problem: theoretical analysis, control design,  discussion, simulations, and experimental validation with onboard-only sensors. We believe that this feature represents a contribution, per se, in addition to each single contribution.

The main novelties of this paper with respect to our previous related work~\citep{2010c-FraSteDirOri} are the following: {\em (i)} extension of the controller to the 3D case {\em (ii)} integration of a decentralized mechanism for maintaining a safe inter-robot distance {\em (iii)} decentralized estimation of the global quantities {\em (iv)} new simulation and experiments, including the case of 3D aerial vehicles. 

The paper is organized as follows. Section~\ref{sec:prob} discusses the encirclement problem and formulates its three versions considered in the paper. Section~\ref{sec:control} introduces the encirclement controllers, while Section~\ref{sec:distance} describes an extension that guarantees maintenance of a safe inter-robot distance. Section~\ref{sec:simexp} presents simulation results with 3D point robots and quadrotor UAVs, as well as experimental results with differential-drive ground robots. Section~\ref{sec:concl} concludes the manuscript and hints at some future work.

\fi
\ifx\noprob\undefined
\section{Problem Formulation}
\label{sec:prob}

\begin{table}[h!]
\begin{center}
\begin{tabular}{|c|p{5cm}|}
\hline
$n$ & number of robots\\[0.06cm]
$R_i$ & $i$-th point robot\\[0.06cm]
$\calW$ & inertial frame\\[0.06cm]
$\pv_i\in\mathbb{R}^3$ & cartesian position of $R_i$ in $\calW$\\[0.06cm]
$\uv_i\in\mathbb{R}^3$ & cartesian velocity input for $R_i$\\[0.06cm]
$T$ & representative point of the target\\[0.06cm]
$\calN_i$ & the \emph{neighbor set} of robot $i$\\[0.06cm]
$\calT$ & target frame\\[0.06cm]
$\pv_\calT\in\mathbb{R}^3$ & cartesian position of $T$ in $\calW$\\[0.06cm]
$\Rm_{\calT} \in SO(3)$ & rotation matrix from~$\calW$ to $\calT$\\[0.06cm]
$\qv_i=(\rho_i \> \phi_i \> z_i)^T$ & position of $R_i$ in $\calT$ in cylindrical coordinates \\[0.06cm]
$\rho_i \in \mathbb{R}_0^+$ & radius of $R_i$ in $\calT$\\[0.06cm]
$\phi_i \hspace{-0.7mm}\in\hspace{-0.7mm} [0, 2\pi)$ & phase of $R_i$ in $\calT$ \\[0.06cm]
$z_i\in\mathbb{R}$ & height of $R_i$ in $\calT$ \\[0.06cm]
$\vv_i=\dot \qv_i$ & cylindrical velocity input for $R_i$\\[0.06cm]
$\bar \phi_i$ & average between the phases \\[0.06cm]
	\; & \quad of the successor and the predecessor of $R_i$\\[0.06cm]
$\Delta_i$ & half-difference between the phases\\[0.06cm]
       \; & \quad of the successor and the predecessor of $R_i$\\[0.06cm]
$\delta_i$ & difference between the phases of $R_i$ \\[0.06cm]
       \; & \quad and its predecessor\\[0.06cm]
$\rho^*$ & desired encirclement radius\\[0.06cm]
$\omega$ & encirclement angular speed\\[0.06cm]
$\omega^*$ & desired encirclement angular speed\\[0.06cm]
$s$ & escape window\\[0.06cm]
$\Cm\hspace{-0.5mm}\in\hspace{-0.5mm}\mathbb{R}^{n\times n}$ & circulant matrix with first row\\[0.06cm]
       \; & \quad $\left(0\>\> \frac{1}{2} \>\> 0\> \cdots\> 0\>\> \frac{1}{2}\right)$\\[0.06cm]
$\Dm\hspace{-0.5mm}\in\hspace{-0.5mm}\mathbb{R}^{n\times n}$ & circulant matrix with first row\\[0.06cm] 
       \; & \quad $\left(0\>\> -\frac{1}{2}\>\> 0\> \cdots\> 0\>\> \frac{1}{2}\right)$\\[0.06cm]
$\mat{H}\hspace{-0.5mm}\in\hspace{-0.5mm}\mathbb{R}^{n\times n}$ & circulant matrix with first row\\[0.06cm] 
       \; & \quad $\left(0\>\> 0\>\> 0 \>\> \cdots\>\> 0\>\> -1\right)$\\[0.06cm]
$\vect{1}\in\mathbb{R}^{n}$  & $\left(1 \cdots 1\right)^T$\\[0.06cm]
$\vect{0}\in\mathbb{R}^{n}$  & $\left(0 \cdots 0\right)^T$\\[0.06cm]
$\vect{b}\in\mathbb{R}^{n}$  & $\left(-\pi \; 0 \cdots 0  \; \;\pi \right)^T$\\[0.06cm]
$\vect{g}\in\mathbb{R}^{n}$  & $\left(\pi \; 0 \cdots 0  \; \pi \right)^T$\\[0.06cm]
$\vect{h}\in\mathbb{R}^{n}$  & $\left(2\pi \; 0 \cdots 0  \; 0 \right)^T$\\[0.06cm]
${\vect{\phi}}\in\mathbb{R}^{n}$ & vector of robot phases\\[0.06cm]
$\bar{\vect{\phi}}\in\mathbb{R}^{n}$ & vector of phase averages\\[0.06cm]
${\vect{\Delta}}\in\mathbb{R}^{n}$ & vector of phase half-differences \\[0.06cm]
${\vect{\delta}}\in\mathbb{R}^{n}$ & vector of consecutive phase differences \\[0.06cm]
${\vect{e}}_{\phi}\in\mathbb{R}^{n}$ & phase error vector\\[0.06cm]
$\xi_i$ & constant forcing term for $R_i$\\[0.06cm]
$\bar \xi$ &  average of the forcing terms\\[0.06cm]
$\hat \eta_i$  & estimate of a generic global quantity $\eta$\\[0.06cm]
  \; & \quad computed by $R_i$\\[0.06cm]
$\delta_{\rm min}(t)$ & $\min_{i}\delta_i(t)$\\[0.06cm]
$r$ & safety radius of the robots\\[0.06cm]
$D_{ij}$ & inter-distance between $R_i$ and $R_j$\\[0.06cm]
\hline
\end{tabular}
\end{center}
\caption{Main symbols used in the paper}
\label{table:symbols}
\end{table}

\noindent
For the convenience of the reader, we have collected in Table~\ref{table:symbols} the main symbols to be used in the paper.

Consider a system of mobile robots and a target moving in a 3D space. The target can be an inanimate object, another robot, or even a living agent. The task assigned to the multi-robot system is to {\em encircle} the target, i.e., move around it in a regular circular formation, often referred to as {\em splay state} formation~\citep{2008-SepPalLeo}. The problem can be reformulated in 2D, if needed, by assuming that robots and target always move on the same plane and discarding the orthogonal coordinate to that plane.

The robots are modeled as $n$ kinematic 3D points $R_1,\ldots,R_n$. Denoting the position of $R_i$ in the inertial world frame $\calW$ by $\pv_i\in \mathbb{R}^3$, each robot is modeled as a first-order dynamical system of the form
\begin{align}
\dpv_i = \uv_i, \qquad i=1,\ldots,n,
\label{eq:KinRobot}
\end{align}
where $\uv_i$ is the velocity control input. Note that the number $n$ is \emph{not} known to the robots, and will not be directly used in any of the control laws to be designed, 

Using a fully actuated kinematic model for the robots allows to focus on the design of decentralized control laws for achieving the encirclement task rather than on the specific dynamics of the robot. On the other hand, our control method will still be applicable to a large class of robots. In fact, the cartesian trajectories generated by the ideal model~(\ref{eq:KinRobot}) can be effectively used as reference for any mobile robot provided that at least one point $P_i$ of the robot can asymptotically track any (smooth) trajectory. A sufficient condition for this is that the position of $P_i$ is (part of) a set of linearizing outputs for the robot; in this case, in fact, there exists a feedback transformation such that the position of $P_i$ is produced by a chain of input-output integrators~\citep{1995-Isi}. Relevant examples include:

\begin{itemize}

\item the majority of wheeled mobile robots, and in particular differential-drive and car-like vehicles, in which feedback linearization of the position of suitable 2D points of the robot body can be obtained via static or dynamic feedback~\citep{2002-OriDelVen};

\item helicopter and quadrotor UAVs, where dynamic feedback linearization of the 3D center of mass can be achieved~\citep{1998-VniMur,2001-MisBenMsi};

\item more in general, all differentially flat systems~\citep{1995-FliLevMarRou} in which the flat outputs include the cartesian position of a point.

\end{itemize}

The effectiveness of this approach will be illustrated in Sect.~\ref{sec:simexp}, where we will report simulations on quadrotor UAVs and experiments on differential-drive mobile robots.

In the following, we also assume that each robot~$i$ can communicate with a subset of robots, denoted by $\calN_i$ (the \emph{neighbor set}), which implicitly defines the  \emph{communication graph}.
The communication graph can change arbitrarily over time with the only constraint that connectivity is preserved. Since the communication graph is connected, each robot can in principle exchange information with any other robot, e.g., using a suitable routing strategy~\citep{2009-GuiDaiCim}. Nevertheless, for the sake of scalability and decentralization, all-to-all information exchange should be avoided as much as possible in multi-robot algorithms. In particular, all the controllers to be presented in the paper can be implemented provided that each robot can communicate with two other members of the group, regardless of its size. We will come back on this important aspect in Section III-E.

\begin{figure*}[!t]
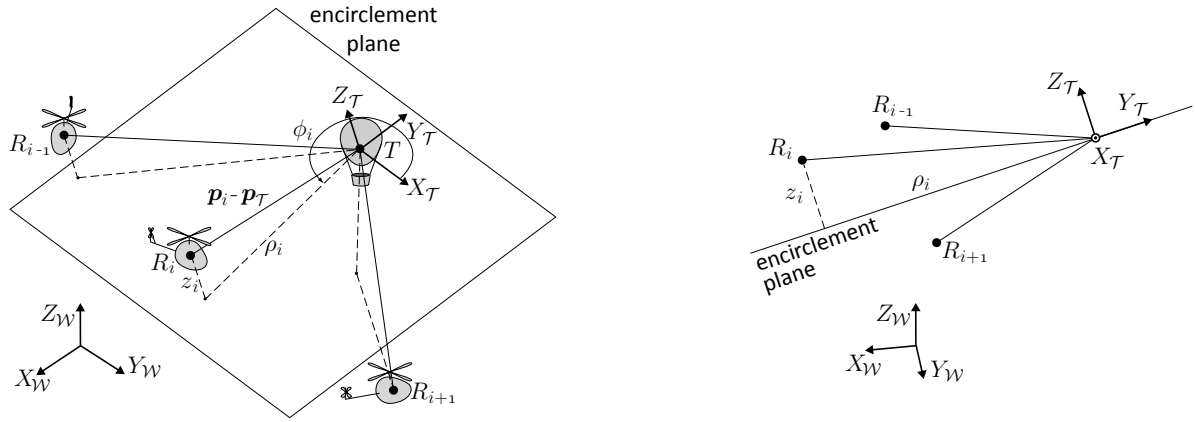

\figPolarThreeD
\centering
\caption{Geometrical setting for the encirclement problem: perspective view (left) and side view (right). The target to be encircled is represented as a balloon whereas the robots are helicopters. Note the cylindrical coordinates and the robot ordering defined by the phase angles.}
\label{fig:cylindrical}
\end{figure*}

Consider a representative point $T$ of the target. The encirclement task requires that $R_1,\dots,R_n$ to converge to a regular circular formation centered at $T$ and lying on a plane passing through $T$, called {\em encirclement plane}, whose orientation is assigned. We consider then a target frame $\calT$ centered at  $T$ and such that the plane $X_\calT$-$Y_\calT$ coincides with the encirclement plane. Since the target may move, and a time-varying orientation may be assigned to the encirclement plane, $\calT$ is in general a moving frame.

A natural formulation of the encirclement problem is obtained by using cylindrical (rather than cartesian) coordinates centered at $T$ as robot configuration variables. In particular, with reference to Fig.~\ref{fig:cylindrical}, let 
\begin{align}
\qv_i=(\rho_i \>\> \phi_i \>\> z_i)^T,\notag
\end{align}
where $\rho_i$ is the distance between $T$ and the orthogonal projection of $R_i$ on the encirclement plane $X_\calT$-$Y_\calT$, $\phi_i$ is the angle between $X_\calT$ and the line joining that projection with $T$, and $z_i$ is the distance between $R_i$ and $X_\calT$-$Y_\calT$. We will call the coordinates $\rho_i$, $\phi_i$, and $z_i$ respectively the {\em radius}, {\em phase}, and {\em height} of the $i$-th robot.

The cylindrical coordinates $\qv_i$ can be easily computed from the cartesian coordinates $\pv_i$. To this end, define the following scalar functions of a generic position vector $\pv = (p_x \>\> p_y \>\> p_z)^T$ 
\begin{align}
 \rho&:\mathbb{R}^3\to \mathbb{R}^+_0,  &\rho(\pv) &= \sqrt{p_x^2+p_y^2} \label{eq:cyl_rho}\\
 \phi&:\mathbb{R}^3\to[0,2\pi),  &\phi(\pv) &= {\rm atan2}(p_y,p_x) \qquad\qquad\label{eq:cyl_phi}\\
 z&:\mathbb{R}^3\to\mathbb{R},      &z(\pv) &=  p_z\label{eq:cyl_z}
\end{align}
and the vector function
\begin{equation}
\qv:\mathbb{R}^3\to \mathbb{R}^3, \qquad  \qv(\pv) = (\rho(\pv) \>\> \phi(\pv) \>\> z(\pv))^T.
\label{eq:qfun}
\end{equation}
We can then write
\begin{align}
 \rho_i &= \rho\left(\Rm_\calT^T(\pv_i - \pv_\calT)\right) \label{eq:cyl_rho_i}\\
 \phi_i &= \phi\left(\Rm_\calT^T(\pv_i - \pv_\calT)\right)\label{eq:cyl_phi_i}\\
  z_i   &= z\left(\Rm_\calT^T(\pv_i - \pv_\calT)\right),\label{eq:cyl_z_i}
\end{align} 
where $\Rm_\calT$ is the rotation matrix from~$\calW$ to $\calT$ and $\pv_\calT$ is the position of $T$ in~$\calW$.
In a compact form, we have
\begin{equation}
\qv_i=\qv\left(\Rm_\calT^T(\pv_i - \pv_\calT)\right).
\label{eq:cylcoords}
\end{equation}

In the following, it is assumed that the robot index $i$ refers to the cyclic counterclockwise ordering of the robots defined by their increasing phase angles at the initial time instant $t_0$ (see Fig.~\ref{fig:cylindrical}).
Note that the labeling defined by the phase ordering at time $t_0$ is never changed: that is, from the viewpoint of the generic robot $i$, the identity of robot $i-1$ ($i+1$) is the same throughout the motion, even if at time $t>t_0$ the actual predecessor (successor) in the phase ordering may be a different robot.

Define the average between the phases of the successor and the predecessor of the $i$-th robot as
\begin{align}
\bar\phi_1 &= \frac{\phi_2 + \phi_n - 2\pi}{2}\label{eq:PhAve1}\\[0.15cm]
\bar \phi_i &= \frac{\phi_{i+1}+\phi_{i-1}}{2} \qquad i=2,\dots,n-1\label{eq:PhAvei}\\[0.15cm]
\bar\phi_n &= \frac{\phi_1 + 2\pi + \phi_{n-1}}{2}.\label{eq:PhAven}
\end{align} 

We have now all the elements to state our basic problem. 

\begin{problem}[Encirclement]

The encirclement task is encoded by the following asymptotic conditions
\begin{align}
\lim_{t \to \infty} \rho_i(t) &= \rho^* \label{eq:EncTask_R}\\
\lim_{t \to \infty} \phi_i(t)  &= \bar\phi_i(t)\label{eq:EncTask_phi}\\
\lim_{t \to \infty} \dot\phi_i(t) &= \omega,\label{eq:EncTask_omega}\\
\lim_{t \to \infty} z_i(t) &= 0\label{eq:EncTask_z}
\end{align}
for all $i=1\ldots n$. Here, $\rho^*$ and $\omega$ are respectively the \emph{encirclement radius} and \emph{encirclement angular speed}, identical for all robots.
\end{problem}

We will consider three different versions of the basic encirclement problem entailed by~(\ref{eq:EncTask_R}--\ref{eq:EncTask_z}).  In all of them, the encirclement radius $\rho^*$ is assigned in advance. The three versions differ in the way the encirclement angular speed $\omega$ in~(\ref{eq:EncTask_omega}) is generated.

\smallskip
\noindent
{\em Encirclement Problem, Version 1:}  A desired angular speed $\omega=\omega^*$ is specified in advance.
\smallskip

In this version, the value of $\omega^*$ typically corresponds to a preferred cruise speed derived, e.g., from energy-related considerations.

\smallskip
\noindent
{\em Encirclement Problem, Version 2:} An {\em escape window} $s$ is assigned, defined as the time interval between two consecutive passings of robots through a generic point of the circle at steady state.
\smallskip

The practical motivation behind Version~2 of the encirclement problem could be to guarantee the effectiveness of the entrapment/escorting task by limiting the possibility that the entrapped target escapes or that the escort is penetrated by a hostile agent. In fact, $s$ may be interpreted as the time in which the circular formation may be violated.

\smallskip
\noindent
{\em Encirclement Problem, Version 3:} The robots must autonomously agree on a certain value of the encirclement angular speed $\omega$. 
\smallskip

Version 3 is interesting from both the theoretical and practical viewpoint since it gives the opportunity to the multi-robot system to autonomously regulate its cruise speed without the need for an external command.

\section{Encirclement Control}
\label{sec:control}

\noindent
We first establish some notation which will be useful for analyzing the proposed control laws. 

Throughout the paper, we denote with $\mat{I}$ the $n\times n$ identity matrix, and with $\mat{C},\mat{D},\mat{H}$ the $n\times n$ {\em circulant} matrices \citep{2006-Gra} with first rows
$\left(0\>\> 1/2\>\> 0\>\> \cdots\>\> 0\>\> 1/2\right)$, $\left(0\>\> -1/2\>\> 0\>\> \cdots\>\> 0\>\> 1/2\right)$, and 
$\left(0\>\> 0\>\> 0\>\> \cdots\>\> 0\>\> -1\right)$, respectively.
Moreover, $\vect{1}$, $\vect{0}$, $\vect{b}$, $\vect{g}$, and $\vect{h}$ are constant vectors whose definition is given in Table~\ref{table:symbols}. Finally, we aggregate the robot phases in ${\vect{\phi}} = \left( \phi_1 \; \cdots \; \phi_n \right)^T$.

We can now define in a compact way three useful vectors: $\bar{\vect{\phi}}$, ${\vect{\Delta}} $, and ${\vect{\delta}}$. The first collects the phase averages $\bar\phi_i$, $i=1,\dots,n$, already defined in~(\ref{eq:PhAve1}--\ref{eq:PhAven}):
\[
\bar{\vect{\phi}} = \left( \bar\phi_1 \>\> \cdots \>\> \bar\phi_n \right)^T = {\mat{C}}{\vect{\phi}} + \vect{b}.
\]
The $i$-th component of the second vector ${\vect{\Delta}}$ is the half-difference between the phases of the successor and the predecessor of robot $i$ (compare with~(\ref{eq:PhAve1}--\ref{eq:PhAven})):
\begin{align}
\Delta_1 &= \frac{\phi_2 - \phi_n + 2\pi}{2} \label{eq:PhHDif1}\\[0.15cm]
\Delta_i &= \frac{\phi_{i+1}-\phi_{i-1}}{2} \qquad i=2,\dots,n-1 \label{eq:PhHDifi}\\[0.15cm]
\Delta_n &= \frac{\phi_1 + 2\pi - \phi_{n-1}}{2} .\label{eq:PhHDifn}
\end{align} 
and the vector itself can be written as
\[
{\vect{\Delta}} = \left( \Delta_1 \>\> \cdots \>\> \Delta_n \right)^T = {\mat{D}} \,{\vect{\phi}} + \vect{g}.
\]
Finally, define the consecutive phase differences 
\begin{align}
\delta_1 &= \phi_1 - \phi_n + 2\pi \label{eq:PhCDif1}\\[0.15cm]
\delta_i &= \phi_{i}-\phi_{i-1} \qquad i=2,\dots,n \label{eq:PhCDifi}
\end{align} 
and the third vector is defined as
\[
{\vect{\delta}} = \left( \delta_1 \>\> \cdots \>\> \delta_n \right)^T = ({\mat{H} + \mat{I}}){\vect{\phi}} + \vect{h}.
\]

We are now ready to address the design of encirclement control laws. The dynamics of the generic robot in cylindrical coordinates is obtained from~(\ref{eq:cylcoords}) and (\ref{eq:KinRobot}) as:
\[
\dqv_i = \Jm_i \left(\dot{\Rm}_\calT^T(\pv_i - \pv_\calT) + \Rm_\calT^T(\uv_i - \dpv_\calT)\right),
\]
where $\Jm_i=\partial \qv/\partial \pv |_{\pv=\pv_i}$ is the Jacobian of the map defined by~(\ref{eq:qfun}), computed at $\pv_i$. Therefore, by letting\footnote{Note that matrix $J_i$ is invertible whenever $\phi_i$ is defined, i.e., unless the $i$-th robot is exactly above the target. This zero-measure case represents a purely theoretical problem, especially considering the presence of noise in the measurements.}
\begin{equation}
\uv_i = \dpv_\calT + 
{\Rm_\calT}
\left(
\Jm_i^{-1} \vv_i - \dot{\Rm}_\calT^T(\pv_i - \pv_\calT)
\right)
\label{eq:feed_lin}
\end{equation}
the dynamics of the robot in cylindrical coordinates become linear and decoupled
\begin{equation}
\dqv_i = \vv_i.
\label{eq:LinDynCylCoords}
\end{equation}
in the new control input $\vv_i=(\dot \rho_i \>\> \dot\phi_i \>\> \dot z_i)^T$.

We have the following straightforward result.

\begin{lemma}\label{res:rho_z}
Letting
\begin{align}
\dot\rho_i   &= k_\rho (\rho^* - \rho_i),  
\label{eq:rhoDyn}\\
\dot z_i &= -k_z z_i, 
\label{eq:zDyn}
\end{align}
with $k_\rho$, $k_z$  positive gains, $\rho_i$ and $z_i$ exponentially converge to $\rho^*$ and $0$, respectively, for any initial condition. \end{lemma}

In other words, all robots will converge to the desired circular trajectory centered at the target and lying on the encirclement plane. 
Note that the evolution of coordinates $\rho_i$ and $z_i$ is not influenced by the motion of the other robots.
In Section~\ref{sec:distance}, we shall modify eq.~\eqref{eq:rhoDyn} to guarantee that a safe inter-robot distance is maintained.

The choice of the second component of the control input $\vv_i$ (i.e., $\dot\phi_i$) depends on which version of the encirclement problem is considered. The three versions are analyzed in detail in the rest of this section. 

\subsection{Encirclement Control, Version 1}
\label{sec:ctrl_ver1}

\noindent
In Version 1 of the encirclement problem, a desired encirclement angular speed $\omega^\ast$ is assigned. Let the second component of the control input $\vv_i$ be defined as
\begin{equation}
\dot\phi_i = \omega^\ast + k_\phi (\bar\phi_i - \phi_i), 
\label{eq:deltaDyn}
\end{equation}
where $k_\phi$ is a positive gain. We have the following result.

\begin{prop}[Controller~1, Desired Angular Speed]
\label{prop:contrOne} 
The control law expressed by~(\ref{eq:feed_lin}) and~(\ref{eq:rhoDyn}), (\ref{eq:zDyn}), (\ref{eq:deltaDyn}) guarantees global exponential convergence of $\rho_i$ to $\rho^*$, $\phi_i$ to $\bar\phi_i$, $\dot\phi_i$ to $\omega^*$, and $z_i$ to $0$, for any choice of $\rho^*$ and $\omega^*$.
\end{prop}

\begin{proof}
In Lemma~\ref{res:rho_z} we have already established that $\rho_i$ and $z_i$ exponentially converge to $\rho^*$ and $0$, respectively. In order to prove the rest of the thesis, it is sufficient to show that the phase error vector
\begin{equation}
\vect{e}_{\phi}=\bar{\vect{\phi}}-{\vect{\phi}} = ({\mat{C}} - {\mat{I}}){\vect{\phi}} + \vect{b}
\label{eq:aux1}
\end{equation}
converges to $\vect{0}$. Rewrite (\ref{eq:deltaDyn}) for the multi-robot system as
\begin{equation}
\dot{\vect{\phi}} = \omega^\ast \vect{1} + k_\phi \,\vect{e}_{\phi},
\label{eq:aux2}
\end{equation}
so that the error dynamics is obtained as
\[
\dot{\vect{e}}_{\phi} = ({\mat{C}}- {\mat{I}}) \dot{\vect{\phi}} = 
\omega^*({\mat{C}}- {\mat{I}})\vect{1} + k_\phi({\mat{C}}- {\mat{I}}){\vect{e}}_{\phi}.
\]
Since $-(\mat{C}- \mat{I})$ can be interpreted as the Laplacian matrix of the undirected ring with weights $1/2$, we conclude the proof by following closely the theory of consensus protocol (see, e.g.~\cite{2004-OlfMur}). In particular, note the following facts:

\begin{itemize}

\item $\vect{1}^T ({\mat{C}}- {\mat{I}})= ({\mat{C}}- {\mat{I}})\vect{1} = {\vect{0}}$;

\item ${\rm ker}({\mat{C}}- {\mat{I}}) = {\rm span}\{\vect{1}\}$; 

\item  ${\mat{C}}- {\mat{I}}$ has a single zero eigenvalue and $n-1$ negative real eigenvalues (hence, it is is negative semidefinite).

\end{itemize}

\noindent
The error dynamics becomes then
\[
\dot{\vect{e}}_{\phi} = k_\phi({\mat{C}}- {\mat{I}}){\vect{e}}_{\phi} .
\]
Writing the free evolution of this linear system in spectral form and using the aforementioned properties of  ${\mat{C}}- {\mat{I}}$ it is easy to conclude that 
\[
\lim_{t \to \infty} \vect{e}_{\phi} = (\vect{1}^T{\vect{e}}_{\phi}(0))\vect{1} = 
\left(\vect{1}^T({\mat{C}}- {\mat{I}}){\vect{\phi}}(0) + \vect{1}^T\vect{b}\right)\vect{1} ={\vect{0}},
\]
and that convergence is exponential.
In view of~(\ref{eq:deltaDyn}), this also implies that $\dot \phi_i$ tends exponentially to $\omega^\ast$, and the proof is complete.
\end{proof}

\smallskip
Note that, independently on the value of $n$,  robot $i$ only needs to communicate with robot $i-1$ and $i+1$ (whose identity has been defined at $t_0$) to implement Controller~1.

\subsection{Encirclement Control, Version 2}

\noindent
In Version 2 of the encirclement problem the robots are assigned a steady-state escape window $s$, which would require an asymptotic angular speed $\omega= 2\pi/ n\,s$. 
However, since $n$ is unknown, the robots must compute a decentralized estimate of this number, denoted by $\hat n$.

In particular, each robot computes its own current estimate as $\hat n_i = 2\pi/\Delta_i$, and correspondingly a desired angular speed $\omega_i = 2\pi/ \hat n_i \, s =\Delta_i/s$, with $\Delta_i$ given by~(\ref{eq:PhHDif1}--\ref{eq:PhHDifn}). This is used as a feedforward term in~(\ref{eq:deltaDyn}) in place of $\omega^\ast$, leading to the following control law for the robot phase:
\begin{equation}
\dot\phi_i = \Delta_i/s  + k_\phi (\bar\phi_i - \phi_i). 
\label{eq:deltaDynTwo}
\end{equation}

\begin{prop}[Controller~2, Desired Escape Window] 
The control law expressed by~(\ref{eq:feed_lin}) and~(\ref{eq:rhoDyn}), (\ref{eq:zDyn}), (\ref{eq:deltaDynTwo}) guarantees global exponential convergence of $\rho_i$ to $\rho^*$, $\phi_i$ to $\bar\phi_i$, $\dot\phi_i$ to $2\pi/n\,s$, and $z_i$ to $0$, for any choice of $\rho^*$ and $s$.
\label{prop:contrTwo}
\end{prop}

\begin{proof}
 Let $f=1/s$ and write~(\ref{eq:deltaDynTwo}) for the multi-robot system as
\[
\dot{\vect{\phi}} = f {\vect{\Delta}} + k_\phi (\bar{\vect{\phi}}-{\vect{\phi}}).
\]
The error dynamics is
\[
\dot{\vect{e}}_{\phi} =  ({\mat{C}}- {\mat{I}}) \dot{\vect{\phi}} = 
f({\mat{C}}- {\mat{I}})({\mat{D}}{\vect{\phi}} + \vect{g}) + k_\phi({\mat{C}}- {\mat{I}}){\vect{e}}_{\phi}.
\]
Using the fact that ${\mat{C}}- {\mat{I}}$ and ${\mat{D}}$ commute, and rearranging terms, we get
\[
\dot{\vect{e}}_{\phi} = (k_\phi({\mat{C}}- {\mat{I}}) + f{\mat{D}}){\vect{e}}_{\phi} + f(({\mat{C}}- {\mat{I}})\vect{g} - {\mat{D}}\vect{b}),
\]
and since $({\mat{C}}- {\mat{I}})\vect{g} - {\mat{D}}\vect{b} = {\vect{0}}$ we conclude that
\[
\dot{\vect{e}}_{\phi} = (k_\phi({\mat{C}}- {\mat{I}}) + f{\mat{D}}){\vect{e}}_{\phi}.
\]
It is easy to verify that matrix $k_\phi({\mat{C}}- {\mat{I}}) + f{\mat{D}}$ has exactly the same properties\footnote{It is a differently weighted Laplacian of the undirected ring.} of ${\mat{C}}- {\mat{I}}$ which were used in the proof of Proposition~\ref{prop:contrOne}. Therefore, we can once again conclude that ${\vect{e}}_{\phi}$ converges to ${\vect{0}}$, and this automatically implies that $\hat{\phi}_i$ converges to $2\pi/n$ and $\dot{\phi}_i$ to $2\pi/ns$. Note that all variables converge exponentially.
\end{proof}

\smallskip
The communication requirements of Controller~1 are the same as Controller~2.

\subsection{Encirclement Control, Version 3}

\noindent
In Version 3 of the encirclement problem the robots must autonomously agree on a common value of the angular speed $\omega$. 
To this end, we propose the following {\em dynamic} control law for controlling the phase of the $i$-th robot:
\begin{align}
\dot\omega_i &= k_\omega (\bar\phi_i-\phi_i), \qquad \omega_i(t_0) = 0
\label{eq:omegaDynNoRefSing}\\
\dot\phi_i &= \omega_i + k_\phi (\bar\phi_i-\phi_i) + \xi_i, 
\label{eq:deltaDynNoRefSing}
\end{align}
where $k_\omega,k_\phi$ are positive gains and $\xi_i$ is a nonnegative constant forcing term. Denote by $\bar \xi=\sum_{i=1}^n \xi/n$ the average of the forcing terms  over the multi-robot system.

To prove that~\eqref{eq:omegaDynNoRefSing}--\eqref{eq:deltaDynNoRefSing} achieve the desired control objective we need a preliminary result.

\begin{lemma}\label{lem:eigen}
Consider a $2n\times 2n$ matrix of the form
\[
{\mat{A}} = \left(\begin{array}{cc}
{\vect{0}} &  k_1{\mat{I}} \\
{\mat{B}} &  k_2{\mat{B}}
\end{array}\right)
\]
where ${\vect{0}}$  is the $n\times n$ null matrix, ${\mat{I}}$ is the  $n\times n$ identity matrix, ${\mat{B}}$ is a  $n\times n$ matrix, and $k_1,k_2$ are nonzero real numbers. For any  eigenvalue $\mu$ of ${\mat{B}}$ with eigenvector ${\vect{u}}$, the roots  $\lambda_{1,2}$ of $\lambda^2 - k_2\mu \lambda - k_1\mu$, are eigenvalues of ${\mat{A}}$ with eigenvectors 
$\left(k_1{\vect{u}}^T  \;\;
\lambda_{1,2}{\vect{u}}^T\right)^T$.
\end{lemma}

\begin{proof} 
In view of the structure of ${\mat{A}}$, vector $\left( {\vect{v_1}}^T {\vect{v_2}}^T\right)^T$ is an eigenvector of ${\mat{A}}$ associated to $\lambda$ if
\begin{align}
k_1{\vect{v_2}} &= \lambda {\vect{v_1}}\label{eq:eigenOne}\\
{\mat{B}}{\vect{v_1}} + k_2{\mat{B}}{\vect{v_2}} &= \lambda {\vect{v_2}}.\label{eq:eigenTwo}
\end{align}
Eq.~(\ref{eq:eigenOne}) means that eigenvectors associated to $\lambda$ must have the structure $\left(k_1{\vect{v}}^T  \;\;
\lambda {\vect{v}}^T\right)^T$. Setting ${\vect{v}}={\vect{u}}$ in this structure, and substituting into (\ref{eq:eigenTwo}) we obtain
$k_1\mu{\vect{u}} + k_2\mu\lambda{\vect{u}} = \lambda^2{\vect{u}}$. The thesis follows. 
\end{proof}

\medskip
The convergence result can now be established.

\begin{prop}[Controller~3, Angular Speed Consensus]
\label{prop:omegaCons} 
The control law expressed by~(\ref{eq:feed_lin}) and~(\ref{eq:rhoDyn}), (\ref{eq:zDyn}), (\ref{eq:omegaDynNoRefSing}--\ref{eq:deltaDynNoRefSing}) guarantees global exponential convergence of $\rho_i$ to $\rho^*$, $\phi_i$ to $\bar\phi_i$, $\dot\phi_i$ to $\bar \xi$, and $z_i$ to $0$, for any choice of $\rho^*$.
\end{prop}

\begin{proof}
Let  ${\vect{\omega}}=\left( \omega_1 \; \cdots \; \omega_n \right)^T$, ${\vect{\xi}} = \left( \xi_1 \; \cdots \; \xi_n \right)^T$ and define the angular frequency error (the reason for the name will be clear at the end of the proof) as
\[
{\vect{e}}_{\omega} = {\vect{\omega}} + {\vect{\xi}} - \bar\xi \vect{1}.
\] 

Writing (\ref{eq:omegaDynNoRefSing}), (\ref{eq:deltaDynNoRefSing}) for the multi-robot system we obtain
\begin{align}
\dot{\vect{\omega}} &= k_\omega (\bar{\vect{\phi}}-{\vect{\phi}}), \qquad {\vect{\omega}}(t_0) = {\vect{0}}\notag\\
\dot{\vect{\phi}} &= {\vect{\omega}} + k_\phi (\bar{\vect{\phi}}-{\vect{\phi}}) +  {\vect{\xi}}.\notag
\end{align}
Now compute the dynamics of the error ${\vect{e}} = ({\vect{e}}_{\phi}^T \; {\vect{e}}_\omega^T)^T$
\begin{align}
\dot{\vect{e}} &=  \left(\begin{array}{c}
k_\omega{\vect{e}}_{\phi} \\
({\mat{C}}-{\mat{I}})({\vect{\omega}}+{\vect{u}}) + k_\phi({\mat{C}}-{\mat{I}}){\vect{e}}_{\phi}
\end{array}\right)=\notag \\
&= \left(\begin{array}{cc}
       {\vect{0}}      & k_\omega{\mat{I}} \\
{\mat{C}}-{\mat{I}}\;\; & k_\phi({\mat{C}}-{\mat{I}})
\end{array}\right)\left(\begin{array}{c}
{\vect{e}}_{\omega} \\
{\vect{e}}_{\phi} 
\end{array}\right)={\mat{\tilde{A}}}{\vect{e}},\notag
\end{align}
where we have used $\bar \xi = \vect{1}^T{\vect{\xi}}/n$ and $({\mat{C}}-{\mat{I}})\vect{1} = {\vect{0}}$. In view of Lemma~\ref{lem:eigen}, the eigenvalues of ${\mat{\tilde{A}}}$ are computed by solving $\lambda^2 - k_\phi\mu \lambda - k_\omega\mu=0$, with $\mu$ eigenvalue of $\mat{C} - \mat{I}$.
We obtain thus
\[
\lambda_{1,2} (\mu) = \frac{1}{2}\left(k_\phi\mu\pm\sqrt{k_\phi^2\mu^2+4k_\omega\mu}\right).
\]
We recall (see the proof of Proposition~\ref{prop:contrOne}) that $\mat{C} - \mat{I}$ has a single zero eigenvalue and $n-1$ negative real eigenvalues. In correspondence to $\mu=0$ we immediately get $\lambda_{1,2}(0)=0$, whereas a simple reasoning shows that for any other eigenvalue $\mu<0$ we get $\lambda_{1,2}(\mu)<0$.

To conclude the proof, we show that the error $\vect{e}$ is always orthogonal to the eigenspace of $\mat{\tilde{A}}$ associated to the double eigenvalue in $0$. This is a consequence of three facts. First, it may be readily verified that such eigenspace is $A_0={\rm span}\{(\vect{1}^T \; {\vect{0}}^T)^T,({\vect{0}}^T \; \vect{1}^T)^T\}$. Second, the orthogonal complement $A_0^\perp$ of $A_0$ is an invariant set for the error dynamics, because for any $\vect{w}_\perp\in A_0^\perp$ we have
\[
\begin{pmatrix}
\vect{1}^T & \vect{0}^T\\
\vect{0}^T & \vect{1}^T
\end{pmatrix}
\mat{\tilde{A}}\vect{w}_\perp = 
\begin{pmatrix}
\vect{0}^T & k_\omega\vect{1}^T\\
\vect{0}^T & \vect{0}^T
\end{pmatrix}\vect{w}_\perp =
\begin{pmatrix}
0\\
0
\end{pmatrix}, 
\]
where we exploited twice the fact that $\vect{1}^T(\mat{C}-\mat{I})=\vect{0}$.
Finally, $\vect{e}(t_0)$ belongs to $A_0^\perp$ by construction, being both $\vect{1}^T{\vect{e}}_{\omega}(t_0)=0$ and $\vect{1}^T{\vect{e}}_{\phi}(t_0) = 0$.

Wrapping up, the error ${\vect{e}} = ({\vect{e}}_{\phi}^T \; {\vect{e}}_\omega^T)^T$ converges to zero. The convergence of ${\vect{e}}_{\phi}$ to zero implies the convergence of $\vect{\phi}$ to $\bar{\vect{\phi}}$, whereas the convergence of ${\vect{e}}_{\omega}$ to zero implies that ${\vect{\omega}} + {\vect{\xi}}$ converges to $\bar\xi \vect{1}$, i.e., that $\dot\phi_i$ converges to $\bar \xi$ (see~\eqref{eq:deltaDynNoRefSing}). Once again, all variables converge exponentially.
\end{proof}

\smallskip
As the previous control laws, also Controller~3 can be implemented on robot $i$ provided that the phases of robot $i-1$ and $i+1$ are available via communication.

An interesting scenario for Controller~3 is considered in the following
\begin{corollary} 
Assume that all forcing terms $\xi_i$ in~(\ref{eq:deltaDynNoRefSing}) are zero, with a single exception $\xi_k \neq 0$. Then, the $k$-th robot acts as a leader by imposing $\xi_k/n$ as encirclement angular speed to the whole multi-robot system.
\end{corollary}

\subsection{Decentralized Estimation of the Global Quantities} 
\label{Sect:DecEst}

\noindent
All the three proposed controllers hinge upon the feedback transformation~(\ref{eq:feed_lin}) to linearize and decouple the robot dynamics in cylindrical coordinates. To compute such transformation, each robot should know the quantities $(\pv_\calT,\dpv_\calT)$  and $(\Rm_\calT$, $\dot{\Rm}_\calT)$. While the first two (respectively, position and velocity of the target point) can in principle be measured or reconstructed by on-board sensors, the last two are related to the desired orientation of the encirclement plane and, as such, are specified by the task.
Note that the robots are not required to express the estimated global quantities in an absolute frame of reference, being sufficient the agreement of all robots on those quantities in a relative frame. Once again, absolute measurements are not required, as also proven in the experimental section.

We consider here the most challenging case, in which only one of the robots is {\em informed} (i.e., knows the above quantities), either by direct measurement or as part of the task description. In order to propagate the necessary information to  the remaining $n-1$ robots of the group, we adopt a decentralized estimator based on the consensus tracking algorithm proposed in~\cite{2007c-Ren}.

Denote with $l$ the index of the informed robot that knows $\pv_\calT$, $\dpv_\calT$, $\Rm_\calT$, $\dot{\Rm}_\calT$, and with $\eta$ the generic scalar component of these vector/matrix quantities.
The $i$-th robot computes an estimate $\hat{\eta}_i$ of $\eta$ by using the following algorithm:
\begin{align}
\dot{\hat{\eta}}_i =
\begin{cases}
\dot{\eta} + k_\eta\left(\eta - \hat{\eta}_i\right) & i = l\\
{\rm ave}_{j\in\calN_i} (\dot{\hat{\eta}}_j)
+ k_\eta{\rm ave}_{j\in\calN_i} (\hat{\eta}_j), & i \neq l\
\end{cases}
\label{eq:cons_track}
\end{align}
where $k_\eta$ is a positive constant, and ${\rm ave}_{j_\in\calN_i}(\cdot)$ returns the average of the estimates of the neighbors of robot~$i$.
The following result holds.

\begin{prop}
If the communication graph remains connected, the multi-robot system can achieve decentralized estimation of any time-varying quantity $\eta$ known by one robot using the algorithm~\eqref{eq:cons_track}; i.e., $\hat{\eta}_i$ globally converges to $ \eta$, for $i=1\ldots n$.  
\end{prop}
\begin{proof}
The adjacency graph underlying the problem topology contains always a directed spanning tree with robot~$1$ as unique root. Then, the convergence directly follows from the proof of the consensus tracking algorithm presented in~\cite{2007c-Ren}.
\end{proof}

\medskip

To apply~\eqref{eq:cons_track}, each robot should in principle know the time derivatives of the estimates of its neighbors. These quantities are  needed to compute the feedforward term for tracking the time-varying signal $\eta$. In a practical (necessarily discrete-time) implementation, the derivatives can be numerically computed using the previous values of the estimates.

We emphasize that, as for any decentralized control strategy that relies on recursive estimation of global quantities, successful performance relies on the existence of a sufficient time-scale separation between the dynamics of the consensus (i.e., of the estimation error) and that of the controller. This assumption is generally satisfied for ground vehicles, but may be more critical for other kinds of robots (e.g., UAVs) that need fast control response.

\subsection{Communication and Scalability} 

\noindent
It has been shown above that the proposed estimation-control scheme works provided that the communication graph remains connected over time.
Indeed, the communication graph may even reduce to a \emph{tree}, i.e., it may contain as little as $n-1$ edges.

In particular, the three proposed encirclement controllers require that the $i$-th robot knows the phases of the $(i+1)$-th and $(i-1)$-th robots. Since the communication graph is connected, this can always be guaranteed using multi-hop communication~\citep{2009-GuiDaiCim}. Therefore:

\begin{itemize}

\item the number of messages sent/received by each robot per unit of time is constant, regardless of the number $n$ of robots;

\item the total number of messages exchanged by the whole multi-robot system per unit of time is linear in the number $n$ of 
robots. 

\end{itemize}

As for the estimation part, the number of robots that need to know the global quantities $(\pv_\calT,\dpv_\calT)$  and $(\Rm_\calT$, $\dot{\Rm}_\calT)$, either by direct measurement or as part of the task description, is also $O(1)$. In particular, in the proposed approach it is sufficient that a single robot is informed.

Altogether, the above remarks indicate that the proposed approach for encirclement scales well with the cardinality $n$ of the multi-robot system, which --- we emphasize --- is unknown to the robots.

\section{Maintaining a Safe Distance}
\label{sec:distance}

\noindent
The objective of this section is to show how the previously described control approach can be extended to guarantee that the moving robots  never get closer to each other than a specified distance. This can be used for avoiding collisions among robots during the encirclement task.

In particular, we shall refer in the following
to Controller~1. We preliminary prove a {\em phase preservation} property which will be instrumental in deriving the main result.  

\subsection{Phase Preservation Property}
\label{sec:phase_preservation}

Proposition \ref{prop:contrOne} implies that under Controller~1 the robot phases at steady state are in the same order as the initial phases (actually, the same is true under Controllers 2 and 3, as implied by Propositions~\ref{prop:contrTwo} and~\ref{prop:omegaCons}, respectively). The next result states that along the trajectories of (\ref{eq:deltaDyn}) the initial ordering is actually maintained at {\em all} instants of time.

\begin{prop}\label{prop:phaseOrderOne}
Consider the phase dynamics (\ref{eq:deltaDyn}) and initial conditions $\delta_i(t_0)>0$, $i=1,\ldots,n$. Then: 
\begin{enumerate}
\item $\delta_i(t)>0$, $i=1,\ldots,n$, for all $t\geq t_0$;
\item the lower-bounding signal 
\begin{equation}
\delta_{\rm min}(t)=\min_{i}\delta_i(t)
\label{eq:delta_min}
\end{equation}
has the following properties:
\begin{enumerate}
	\item $\delta_{\rm min}(t)\le 2\pi/n$, for all $t\geq t_0$;
	\item $\delta_{\rm min}(t)$ is non-decreasing;
	\item $\delta_{\rm min}(t) \to 2\pi/n$ as $t \to \infty$.
\end{enumerate}
\end{enumerate}
\end{prop}

\begin{proof}
Using~\eqref{eq:aux1} in~\eqref{eq:aux2}, the phase dynamics for the multi-robot system becomes
\begin{align}
\label{eq:refPhaseTot}
\dot{\vect{\phi}} =  k_\phi (\mat{C} - \mat{I}){\vect{\phi}} + \omega \vect{1} + k_\phi\vect{b}.
\end{align}
In terms of consecutive phase differences, we have
\begin{align}
\dot{\vect{\delta}} &= (\mat{H}+\mat{I})\dot{\vect{\phi}} \notag\\ 
&= k_\phi (\mat{H}+ \mat{I})(\mat{C} - \mat{I}){\vect{\phi}} + \omega (\mat{H}+\mat{I})\vect{1} + k_\phi(\mat{H}+\mat{I})\vect{b} \notag\\
&= k_\phi(\mat{C} \!-\!\mat{I})((\mat{H}\!+\!\mat{I})\vect{\phi} \!+\! \vect{h})) \!+\!  k_\phi((\mat{H}\!+\!\mat{I})\vect{b} \!-\! (\mat{C}\!-\!\mat{I})\vect{h}))=\notag\\
&= k_\phi(\mat{C} - \mat{I}){\vect{\delta}},\label{eq:deltaRefPhaseDynTot}
\end{align}
where we exploited the fact that $\mat{H}+\mat{I}$ and $\mat{C}-\mat{I}$ commute, that $(\mat{H}+\mat{I}) \vect{1}= \vect{0}$, and that $(\mat{H} + \mat{I})\vect{b} - (\mat{C} - \mat{I})\vect{h} = \vect{0}$. Since $k_\phi(\mat{C} - \mat{I})$ is a \emph{Metzler} matrix (i.e., all its off-diagonal terms are positive), eq.~\eqref{eq:deltaRefPhaseDynTot} represents a \emph{positive} system, and therefore the elements of $\vect{\delta}$ remain positive during its evolution (see, e.g.,~\cite{2007-GurShoMas}). 

Concerning the properties of $\delta_{\rm min}$, note first that {\em a)} is a consequence of $\sum_{i=1}^n\delta_i (t)= 2\pi$ (by definition) together with $\delta_i(t)>0$, $i=1,\ldots,n$. Now define $\kappa(t)=\argmin_{i}\delta_i(t)$, i.e., the index such that $\delta_{\kappa(t)}(t)=\delta_{\rm min}(t)$. By definition $\delta_{\kappa(t)\pm 1}(t) \ge \delta_{\kappa(t)}(t)$ and thus~\eqref{eq:deltaRefPhaseDynTot} implies
\[
\dot\delta_{\rm min} =  k_\phi \left(\frac{1}{2}(\delta_{\kappa(t)-1} + \delta_{\kappa(t)+1}) - \delta_{\kappa(t)}\right) \ge 0,
\] 
which proves property {\em b)}. Finally, the convergence to $2\pi/n$, property {\em c)}, descends directly from Proposition~\ref{prop:contrOne}.
\end{proof}

\begin{figure*}[t]
\figSafeDistances
\centering
\caption{Illustration of the geometry for collision avoidance. Left: definition of relevant distances. Right: radial separation ($R_i$ and $R_j$) vs.\ phase separation ($R_i$ and $R_k$).}
\label{fig:safeDistances}
\end{figure*}

\subsection{Sufficient Conditions for Safety} 

\noindent
Denote by $r>0$ the {\em safety radius} of the robots, which represents the minimum acceptable clearance around the robot representative point $R$. The safety radius may be the actual radius of the robot (defined as the maximum distance between $R$ and any other point of the robot) or, typically, it may be further increased to provide a margin, e.g., for accepting trajectory tracking errors (either during the transient or at steady-state due to bounded disturbances). For simplicity, in the following we call {\em collision} the situation in which the inter-distance between the representative points of two robots is less or equal to $2\,r$.
 
Below, we give conditions for {\em statical} safety, i.e., avoidance of collisions between stationary robots.  These will be used for designing a {\em dynamically} safe encirclement controller in Set.~\ref{Sect:SEC}.
Throughout the rest of this section, refer to Fig.~\ref{fig:safeDistances} for illustration. 

The necessary and sufficient condition for avoiding a collision (including simple contact) between robots $i$ and $j$ is that their inter-distance is larger than the above threshold, i.e.,
\[
D_{ij}=\|\pv_i-\pv_j\| > 2r,  
\]
which may be rewritten in cylindrical coordinates as follows:
\begin{align}
\label{eq:distance}
D_{ij}=\sqrt{{\rho_i}^2 + {\rho_j}^2 - 2\rho_i\rho_j\cos(\phi_j - \phi_i) + 
(z_j-z_i)^2
} \ge 2r.
\end{align}
Now denote by $\tilde \pv_i$ and $\tilde \pv_j$ the projections of $\pv_i$ and $\pv_j$,  respectively, on the encirclement plane and let
\[
d_{ij} = \|\tilde \pv_i-\tilde \pv_j\| =\sqrt{{\rho_i}^2 + {\rho_j}^2 - 2\rho_i\rho_j\cos(\phi_j - \phi_i)}.
\]
We have
\begin{equation}
D_{ij} \geq d_{ij} \geq  \sqrt{{\rho_i}^2 + {\rho_j}^2 - 2\rho_i\rho_j} = |\rho_i - \rho_j|.
\label{eq:ineq1}
\end{equation}
On the other hand, letting $\rho^m_{ij}=\min({\rho_i},{\rho_j})$ we may also write
\begin{align}
D_{ij} \ge d_{ij} &\ge \sqrt{ (\rho^m_{ij})^2 +  (\rho^m_{ij})^2 - 2 \rho^m_{ij}  \rho^m_{ij} \cos(\phi_j - \phi_i)}=\notag\\
&=  \rho^m_{ij}  \sqrt{2(1 - \cos(\phi_j - \phi_i))} = \notag\\
&= 2 \rho^m_{ij}  \left| \sin\left(\tfrac{\phi_j - \phi_i}{2}\right) \right| = \tilde d_{ij}.
\label{eq:ineq2}
\end{align}

We shall say that robots $i$ and $j$ are (see Fig.~\ref{fig:safeDistances}):

\begin{itemize}

\item \emph{radially separated} if $|\rho_i - \rho_j| > 2r$;

\item \emph{phase separated} if $\tilde d_{ij} > 2r$, or equivalently 
\[
\rho^m_{ij}  > \frac{r}{| \sin\left((\phi_j - \phi_i)/2\right)|}.
\]

\end{itemize}

\medskip
\begin{prop} 
If two robots are either radially or phase separated, they are not in collision.
\end{prop}

\begin{proof}
It is a direct consequence of using the two separation definitions in~(\ref{eq:ineq1}) and~(\ref{eq:ineq2}), respectively.
\end{proof}

\medskip
Note that radial or phase separation are only sufficient conditions for avoiding collision between two robots; i.e., two robots that are neither radially nor phase separated are not necessarily in collision. The following proposition provides a sufficient condition for safety of robot $i$, i.e., avoidance of collision with any other robot.

\begin{prop}
\label{prop:suff}
Define $\sigma(t) = r/|\sin(\delta_{\rm min}(t)/2)|$, with $\delta_{\rm min}$ given by~(\ref{eq:delta_min}). 
If the following condition holds
\begin{equation}
\rho_i (t) \ge \sigma (t) + 2r
\label{eq:condition_2r}\\
\end{equation} 
at a time instant $t$, then robot $i$ is not in collision with any other robot at $t$.
\end{prop}
\begin{proof}
Assume that~\eqref{eq:condition_2r} holds (drop time dependence for compactness), and consider any other robot $j$, $j\neq i$. If $\rho_j >  \rho_i - 2r$ then $\rho_j > \sigma$, which implies $\rho_{ij}^m  > \sigma$.
Since 
\[
\sigma = \frac{r}{|\sin(\delta_{\rm min}/2)|} \ge  \frac{r}{|\sin((\phi_j-\phi_i)/2)|}
\]
we may conclude that $\rho_{ij}^m  > r/|\sin((\phi_j-\phi_i)/2)|$; i.e., robot $i$ and robot $j$ are phase separated. On the other hand, if $\rho_j \le  \rho_i - 2r$ then $\rho_i -\rho_j \ge  2r$; i.e., robot $i$ and robot $j$ are radially separated.
\end{proof}

\subsection{Safe Encirclement Control} 
\label{Sect:SEC}

\noindent
Define the function $\lambda(\rho_i,\sigma): \mathbb{R}^2 \rightarrow \mathbb{R}$ as follows:
\begin{align}
\lambda(\rho_i,\sigma) 
= 
\left\{
\begin{array}{l}
0 \quad \text{if} \;\; \rho_i < \sigma + 2r \\
1 \quad \text{if} \;\; \rho_i > \sigma + 2r + \varepsilon_r\\
(\rho_i - \sigma - 2r)/\varepsilon_r \quad  \text{otherwise},
\end{array}
\right.\label{eq:lambda_obs}
\end{align}
where $\varepsilon_r$ is any (small) positive constant. The profile of $\lambda$ is shown in Fig.~\ref{fig:safeDistances}, right.
The following proposition presents a collision-free extension of the controller presented in Sect.~\ref{sec:ctrl_ver1}.

\begin{prop}[Controller~1$^*$, Desired Angular Speed with Collision Avoidance] 
\label{prop:contrOneColl}
Replace~\eqref{eq:rhoDyn} in Controller~1 with
\begin{equation}
\dot\rho_i = \lambda(\rho_i,\sigma)\,k_\rho (\rho^*- \rho_i),
\label{eq:rhoDynColl}
\end{equation}
with $\sigma$ defined in Proposition~\ref{prop:suff} and $\lambda$ in \eqref{eq:lambda_obs}.
Then, the control law expressed by~(\ref{eq:feed_lin}) and~(\ref{eq:rhoDynColl}), (\ref{eq:zDyn}), (\ref{eq:deltaDyn}):

\begin{enumerate}

\item \label{item:no_collisions} 
ensures that no collision occurs among robots;

\item \label{item:convergence} 
guarantees global exponential convergence of  $\phi_i$ to $\bar\phi_i$, $\dot\phi_i$ to $\omega^*$, and $z_i$ to $0$, for any choice of $\omega^*$; and exponential convergence of $\rho_i$ to $\rho^*$, provided that

\begin{description}

\item{a)} $\rho^*> \frac{r}{|\sin(\pi/n)|} + 2r$;
\vskip 0.1cm

\item{b)}  $\rho_i (t_0) > \frac{r}{|\sin(\pi/n)|}+2r$;
\vskip 0.1cm

\item{c)} $|\rho_i (t_0) - \rho_j (t_0)| \ge 2r$, $\forall j=1,\ldots,n$, $j\neq i$. 

\end{description}

\end{enumerate}

\end{prop}

\begin{proof}
We shall prove the thesis in two parts.

\emph{Collision Avoidance:} We first prove that the generic $i$-th robot  cannot collide with the $j$-th robot, $\forall j\!\neq\! i$. From Proposition~\ref{prop:phaseOrderOne} we know that $\delta_{\rm min}$ is a non-decreasing signal that converges to $2\pi/n$, which implies that $\sigma+2r$ is a non-increasing signal that converges to $r/|\sin(\pi/n)|+2r$ from above.
Since $\rho_i (t_0) > r/|\sin(\pi/n)| + 2r$ by hypothesis, and being $\dot \rho_i =0$ as long as $\rho_i < \sigma + 2r$, the first   time instant $t_i$ such that $\rho_i (t_i) = \sigma(t_i) + 2r$, is certainly finite, $\forall i=1,\ldots,n$. At $t_i$, $\rho_i$ is `reached' from above by the signal $\sigma + 2r$, and for any $t>t_i$ it will be $\rho_i(t) \geq \sigma(t) + 2r$. For $t\geq t_i$, therefore, condition~\eqref{eq:condition_2r} of Proposition~\ref{prop:suff} holds, and the $i$-th robot cannot collide with any other robot. 

Now consider a generic $t<t_i$, and note that we have $\rho_i (t) = \rho_i (t_0)$. Partition the other robot indices in two sets $A(t)=\{j\,|\, t \ge t_j\}$ and $B(t)=\{j\,|\, t  < t_j \}$. For $j\in A(t)$ it is $t\ge t_j$ and thus condition~\eqref{eq:condition_2r} of Proposition~\ref{prop:suff} holds with $j$ in place of $i$; hence, the $j$-th robot cannot collide with any other robot (in particular, with the $i$-th robot). For $j\in B(t)$ it is $t < t_j$ and then $\rho_j (t) = \rho_j (t_0)$; therefore, the $j$-th robot and the $i$-th robot are radially separated by hypothesis and collisions are prevented also in this case.

\emph{Convergence:} We now prove that the regulation errors converge to zero. Convergence of $\phi_i$, $\dot\phi_i$ and $z_i$ is shown exactly as in Proposition~\ref{prop:contrOne}. To prove convergence of $\rho_i$ to $\rho^*$, we essentially exploit the fact that signal $\tilde\rho(t) + 2r$ (which determines the time-varying gain $\lambda$ in~\eqref{eq:lambda_obs}) converges to $r/|\sin(\pi/n)|+2r$, and thus $\rho_i$ certainly converges to $\rho^*$, since both $\rho^*$ and $\rho_i(t_0)$ are, by hypothesis, larger than $r/|\sin(\pi/n)|+2r$.     

More in detail, the assumption $\rho^* > r/|\sin(\pi/n)|+2r$ implies that $\rho^* = r/|\sin(\pi/n)|+2r+\varepsilon^*$ for a certain $\varepsilon^*>0$. 
Now define the following quantity
\[
\rho_{m,i}=\min\{\rho^*-\tfrac{\varepsilon^*}{2}\;,\;\rho_i(t_0)\},
\]
which is larger than $r/|\sin(\pi/n)|+2r$ by definition. 
Since $\tilde\rho(t) + 2r$ converges monotonically to $r/|\sin(\pi/n)|+2r$ from above, there certainly exists a time instant $t^*_i>t_0$ such that 
\[
\tilde\rho(t^*_i)+ 2r  = \rho_{m,i}
\quad \mbox{and} \quad
{\dot{\tilde\rho}}(t^*_i) < 0
\]
so that $\tilde\rho(t) + 2r < \rho_{m,i}$ $\forall t>t^*_i$.

For any $t\ge t^*_i$ it is clearly $\rho_i(t)\ge \tilde\rho(t) + 2r$; i.e., $\rho_i(t)$ lies in the right half-line with origin at $\tilde\rho(t) + 2r$, which contains also $\rho^*$ by construction. The only two possible equilibria of $\rho_i$ after $t^*_i$ (obtained imposing $\dot\rho_i=0$) are therefore (1) $\rho_i=\tilde\rho(t) + 2r$ (which implies $\lambda=0$) and (2) $\rho_i=\rho^*$.  The first equilibrium is  unstable, since for any $\rho_i\in(\tilde\rho(t) + 2r,\rho^*]$ it is $\frac{d}{dt}(\rho-\tilde\rho + 2r)=\dot\rho-\dot{\tilde\rho}>0$. On the other hand, $\rho^*$ is asymptotically stable and its region of attraction is the whole open interval $(\tilde\rho(t) + 2r,+\infty)$.

To conclude the proof, let us look at the value of $\rho_i$ at $t^*_i$. If  $\rho(t^*_i)>\tilde\rho(t^*_i) + 2r$, then $\rho_i$ is already in the region of attraction of $\rho^*$ and will then converge to it. If instead $\rho(t^*_i)=\tilde\rho(t^*_i) + 2r$, then $\left.\frac{d}{dt}\right|_{t^*_i}(\rho-\tilde\rho + 2r)= 0 - \dot{\tilde\rho}(t^*_i)>0$, which implies that $\tilde\rho$ after $t^*_i$ will  leave the unstable equilibrium $\rho-\tilde\rho + 2r$ to enter the region of attraction of $\rho^*$.
\end{proof}

\medskip

A comparison of Proposition~\ref{prop:contrOneColl} with Proposition~\ref{prop:contrOne} shows that the price to pay for adding guaranteed collision avoidance to Controller~1 is threefold. First, the encirclement radius $\rho^*$ cannot be too small (condition a): its minimum admissible value depends on the number of robots and their radius $r$, and in particular the higher $n$ (or $r$), the higher $\rho^*$. Second,
the initial distance of each robot from the target cannot be too small (condition b, and note that the threshold is the same of condition a).
Finally, all robots must be radially separated at the start (condition c). Taken together, these three requirements represent only a sufficient condition; collision-free encirclement may be obtained even if one (or more) of them is violated.

Note also that function~\eqref{eq:lambda_obs} is only the simplest choice for producing a gain $\lambda$ that varies continuously between $0$ and $1$. Different choices (see Section~\ref{sec:sim}) can be considered if a smoother control law is desired; Proposition~\ref{prop:contrOneColl} will still hold.

\subsection{Decentralized Estimation of $\sigma(t)$} 
\label{Sect:DecEstSigma}
     
\noindent
The safe encirclement control law~\eqref{eq:rhoDynColl} requires the knowledge of the globally defined quantity $\sigma(t)$. In order to preserve decentralization and scalability of the proposed approach, we show below how the generic $i$-th robot can compute a decentralized estimate $\hat{\sigma}_i$ that can be used in place of $\sigma$ in the control law~\eqref{eq:rhoDynColl} while preserving the validity of Proposition~\ref{prop:contrOneColl}.  

From the proof of Proposition~\ref{prop:contrOneColl}, it is clear that if
\begin{align}
\lim_{t \to \infty} \hat\sigma_i (t) = \sigma(t) \qquad \forall i,
\label{eq:asympt_rho}
\end{align} 
then the associated convergence properties still hold.
The proof additionally shows that if the estimates satisfy
\begin{align}
\lambda(\rho,\hat\sigma_i(t)) \le \lambda(\rho,\sigma(t))\quad  \forall i,\forall t>t_0,
\label{eq:lambda_le}
\end{align}
then the collision avoidance property is also preserved. In view of the definition of $\lambda$ in~\eqref{eq:lambda_obs}, condition~\eqref{eq:lambda_le} can be rewritten as
\begin{align}
\hat\sigma_i \ge \sigma.\label{eq:prop_est}
\end{align}
Therefore, we shall synthesize $\hat\sigma_i$ so as to satisfy both~\eqref{eq:asympt_rho} and~\eqref{eq:prop_est}.

The proposed decentralized estimator for $\sigma$ has a discrete-time structure. In particular, denoting by $T_c$ the control sampling time, consider the following basic iteration: 
\begin{eqnarray*}
\gamma_i [0] &\!\!\!\!=\!\!\!\!& \rho_i(0)\\
\gamma_i [k+1] &\!\!\!\!=\!\!\!\!&
\begin{cases}
 \rho_i(k T_c) \qquad\qquad\>\>\>\>\>\> \text{if } k \text{ is a multiple of } m &\\
 \max_{j\in\calN_i}\left\{ \gamma_i[k],  \gamma_{j}[k]  \right\} \qquad \text{otherwise,}  &
\end{cases}\\
\end{eqnarray*}
where $k$ is incremented every $T_c$ seconds and $m$ is any integer larger than $n-1$ (an upper bound on $n$ for the considered scenario is needed here). 
This scheme achieves a finite-time agreement every $m T_c$ seconds, i.e.,  
\[
\gamma_i[m(k\div m)]=\max_{j=1\ldots n}\rho_j((k-m)T_c),
\]
where $k\div m$ is the quotient of the division of $k$ and $m$.
Each robot then updates its estimate $\hat\sigma_i$ of $\sigma$ as follows:
\begin{equation}
\hat\sigma_i(k T_c) =
\begin{cases}
 \infty & \text{if } k<m\\
 \gamma_i[m(k\div m)] & \text{otherwise.}
\end{cases}
\label{eq:good_rho_min_est}
\end{equation}

Note that in this case the estimation algorithm is proposed in discrete time to account for multiple communications steps during a single control step.
As before, this decentralized estimation method can be implemented under any connected communication topology.

We have the following result.

\begin{prop}
Assume that the estimates $\hat\sigma_i$ produced by algorithm~\eqref{eq:good_rho_min_est} are used in place of $\sigma$ to implement a decentralized version of control law~\eqref{eq:rhoDynColl}. Then the thesis of Proposition~\ref{prop:contrOneColl} is still valid; in particular, collision-free encirclement with a desired speed is achieved.
\end{prop}

\begin{proof}
We know from Proposition~\ref{prop:phaseOrderOne} that $\delta_{\rm min}$ is non-decreasing and converges to $2\pi/n$. Thus, $\sigma$ is non increasing and converges to the constant value $r/|\sin(\pi/n)|$. Exploiting this fact, it is straightforward to prove that the estimates produced by the proposed protocol satisfy both~\eqref{eq:asympt_rho} (i.e., decentralized estimation of $\sigma$) and~\eqref{eq:prop_est} .
\end{proof}

\fi
\ifx\noexp\undefined

\begin{figure*}[t]
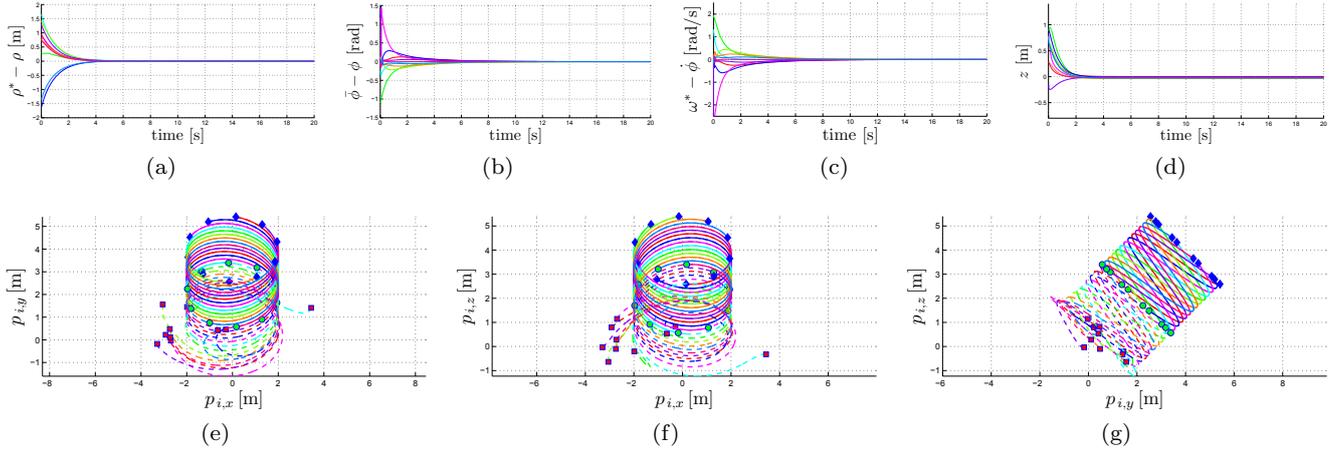

\centering
\subfloat[\label{fig:SimARho}]{\FigSimAErrRho} 
\hfill 
\subfloat[\label{fig:SimAPhi}]{\FigSimAErrPhi}
\hfill
\subfloat[\label{fig:SimADotPhi}]{\FigSimAErrDotPhi}
\hfill
\subfloat[\label{fig:SimAZ}]{\FigSimAErrZ}
\\
\subfloat[\label{fig:SimAXY}]{\FigSimATrajXY}
\hfill 
\subfloat[\label{fig:SimAXZ}]{\FigSimATrajXZ}
\hfill
\subfloat[\label{fig:SimAYZ}]{\FigSimATrajYZ}
\caption{3D point robots, first simulation: Controller~1 (Desired Angular Speed) with 10 robots. \protect\subref{fig:SimARho},\protect\subref{fig:SimAPhi},\protect\subref{fig:SimADotPhi},\protect\subref{fig:SimAZ}: Encirclement error signals. \protect\subref{fig:SimAXY},\protect\subref{fig:SimAXZ},\protect\subref{fig:SimAYZ}:  Projection of the robot trajectories on the coordinate planes (dashed: for $t\in[0,10]$~s; solid: for $t\in[10,20]$~s). Robot positions at $t=0$~s, $t=10$~s, and $t=20$~s are shown as red squares, green circles, and blue diamonds, respectively.} 
\label{fig:SimA}
\end{figure*}

\begin{figure*}[t]
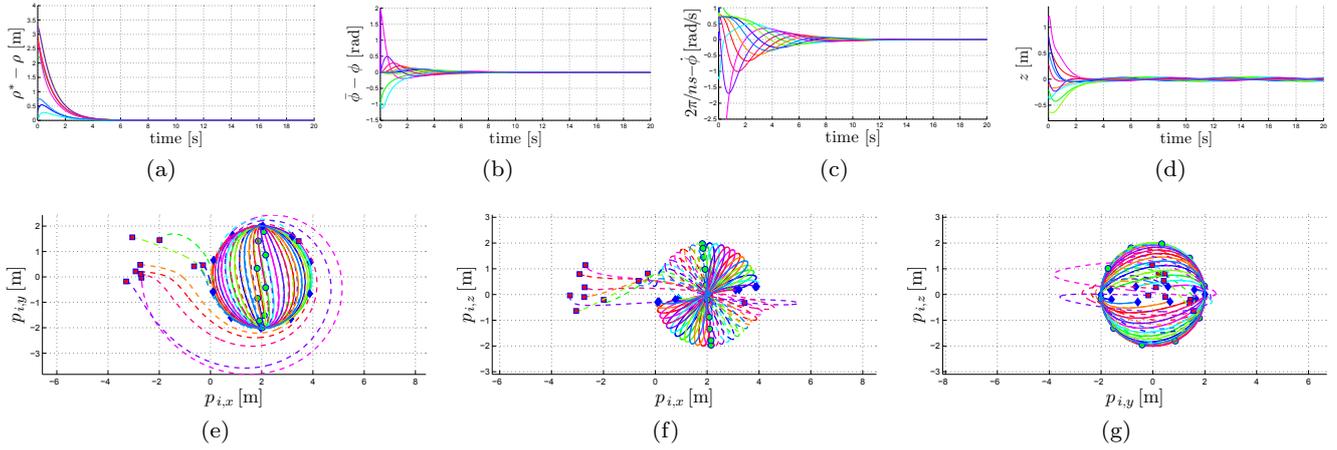

\centering
\subfloat[\label{fig:SimBRho}]{\FigSimBErrRho} 
\hfill 
\subfloat[\label{fig:SimBPhi}]{\FigSimBErrPhi}
\hfill
\subfloat[\label{fig:SimBDotPhi}]{\FigSimBErrDotPhi}
\hfill
\subfloat[\label{fig:SimBZ}]{\FigSimBErrZ}
\\
\subfloat[\label{fig:SimBXY}]{\FigSimBTrajXY}
\hfill 
\subfloat[\label{fig:SimBXZ}]{\FigSimBTrajXZ}
\hfill
\subfloat[\label{fig:SimBYZ}]{\FigSimBTrajYZ}
\caption{3D point robots, second simulation: Controller~2 (Desired Escape Window) with 10 robots. \protect\subref{fig:SimARho},\protect\subref{fig:SimAPhi},\protect\subref{fig:SimADotPhi},\protect\subref{fig:SimAZ}: Encirclement error signals. \protect\subref{fig:SimAXY},\protect\subref{fig:SimAXZ},\protect\subref{fig:SimAYZ}:  Projection of the robot trajectories on the coordinate planes (dashed: for $t\in[0,10]$~s; solid: for $t\in[10,20]$~s). Robot positions at $t=0$~s, $t=10$~s, and $t=20$~s are shown as red squares, green circles, and blue diamonds, respectively.} 
\label{fig:SimB}
\end{figure*}

\section{Simulations and Experiments}
\label{sec:simexp}

\noindent
This section describes the simulations and experiments that have been performed in order to validate the proposed encirclement controllers. See the multimedia material attached to the paper for illustrative video clips.

\subsection{Simulations with Kinematic 3D Point Robots}
\label{sec:sim}

\begin{figure*}[t]
\centering
\subfloat[\label{fig:SimCRho}]{\FigSimCErrRho} 
\hfill 
\subfloat[\label{fig:SimCPhi}]{\FigSimCErrPhi}
\hfill
\subfloat[\label{fig:SimCDotPhi}]{\FigSimCErrDotPhi}
\hfill
\subfloat[\label{fig:SimCZ}]{\FigSimCErrZ}
\\
\subfloat[\label{fig:SimCXY}]{\FigSimCTrajXY}
\hfill 
\subfloat[\label{fig:SimCXZ}]{\FigSimCTrajXZ}
\hfill
\subfloat[\label{fig:SimCYZ}]{\FigSimCTrajYZ}
\caption{3D point robots, second simulation: Controller~3 (Angular Speed Consensus) with 10 robots. \protect\subref{fig:SimARho},\protect\subref{fig:SimAPhi},\protect\subref{fig:SimADotPhi},\protect\subref{fig:SimAZ}: Encirclement error signals. \protect\subref{fig:SimAXY},\protect\subref{fig:SimAXZ},\protect\subref{fig:SimAYZ}:  Projection of the robot trajectories on the coordinate planes (dashed: for $t\in[0,10]$~s; solid: for $t\in[10,20]$~s). Robot positions at $t=0$~s, $t=10$~s, and $t=20$~s are shown as red squares, green circles, and blue diamonds, respectively.} 
\label{fig:SimC}
\end{figure*}

\begin{figure*}[t]
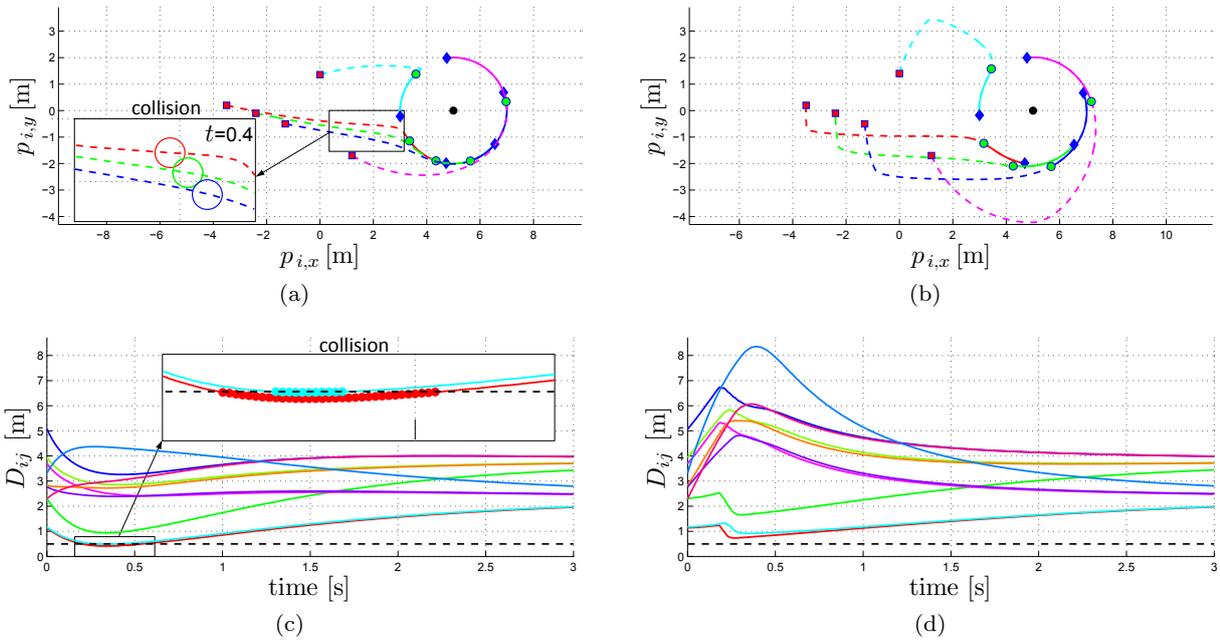

\centering
\hfill
\subfloat[\label{fig:SimNOCOxy}]{\FigSimNOCOTrajXY}
\hfill 
\subfloat[\label{fig:SimCOxy}]{\FigSimCOTrajXY} 
\hfill\,
\\
\hfill
\subfloat[\label{fig:SimNOCOdist}]{\FigSimNOCODistances}
\hfill\;
\subfloat[\label{fig:SimCOdist}]{\FigSimCODistances}
\hfill\,
\caption{3D point robots, fourth simulation: Encirclement control with 5 robots, with and without collision avoidance.
\protect\subref{fig:SimNOCOxy},\protect\subref{fig:SimNOCOdist}: Robot trajectories and inter-robot distances with  Controller~1; note the double collision. \protect\subref{fig:SimCOxy},\protect\subref{fig:SimCOdist}: Robot trajectories and inter-robot distances with Controller~1$^*$ (Desired Angular Speed with Collision Avoidance). 
} 
\label{fig:SimCO}
\end{figure*}

\noindent
The first set of simulations involves systems of point robots moving in 3D space. The objective is to test the proposed encirclement controllers for different motions of the target and of the encirclement plane. The global quantities $(\pv_\calT,\dot\pv_\calT)$ and $(\Rm_\calT,\dot\Rm_\calT)$ are always estimated via the algorithm~\eqref{eq:cons_track}, assuming that only one robot in the whole group is informed about the global quantities and that $\calN_i = \{i+1, i-1\}$. This is clearly the worst-case scenario, since the presence of additional communication links in the robot network would have the effect of accelerating the convergence of the estimates to the correct values, and hence of the multi-robot system to the encirclement steady state.
f
Figure~\ref{fig:SimA} shows the result of a simulation where Controller~1 (Desired Angular Speed) is used with 10 point robots. The desired encirclement values are set to $\rho^*=2$~m and $\omega^*=0.8$~rad/s. The target moves at constant velocity $\dot\pv_\calT = (0,0.2,0.2)$~m/s. The encirclement plane $X_\calT$-$Y_\calT$ is oriented orthogonally to $\dot\pv_\calT$; it translates because the target moves but it does not rotate. The control gains are $k_\rho =1$, $k_z =1.5$ and $k_\phi =2$. As expected, the four variables that encode the encirclement task according to~(\ref{eq:EncTask_R}--\ref{eq:EncTask_z}) converge exponentially to their desired value.

Figure~\ref{fig:SimB} considers the same system of robots under the action of Controller~2 (Desired Escape Window). The desired encirclement values are set to $\rho^*=2$~m and $s^*=0.78$~s. The target is fixed but the encirclement plane, which is initially horizontal, now rotates with angular velocity $\omegav_\calT = (0,0.15,0)$~rad/s. The control gains are the same of the first simulation. Again, the encirclement task is achieved with exponential speed; note in particular the convergence of the escape window $s$ to its desired value. At steady state, the robots move in a regular formation along a \emph{great circle} of the sphere of radius $\rho^*$ centered in the target;  this great circle rotates on the sphere over time due to the rotational motion of the encirclement plane.

The third simulation (Fig.~\ref{fig:SimC}) refers to the same robot system now subject to Controller~3 (Angular Speed Consensus). The value of the encirclement radius is again $\rho^*=2$~m, while vector $\vect{\xi}$ of the forcing terms is chosen randomly, resulting in $\bar\xi=0.8$~rad/s. The target moves at constant velocity $\dot\pv_\calT = (0.5,0,0)$~m/s; at the same time, the encirclement plane, which is initially horizontal, rotates with angular velocity $\omegav_\calT = (0,0.3,0)$~rad/s. The control gains $k_\rho$, $k_z$, $k_\phi$ are the same as before, whereas $k_\omega=3$. As before, the encirclement signal errors decay exponentially to zero; in particular, the encirclement angular speed converges to $\bar\xi$. Since the motion of the encirclement plane is now a full roto-translation, the robot trajectories tend to become composite helical-spherical curves.

The final simulation is aimed at validating Controller~1$^*$ (Desired Angular Speed with Collision Avoidance). To this end, we have considered a system of 5 circular robots of radius $r=0.25$~m. Both the target and the encirclement plane are now fixed. For simplicity, it is assumed that all the robots start already on the encirclement plane ($z_i=0$, $\forall i$), so that their motion is actually planar. As shown in Fig.~\ref{fig:SimCO}, when the basic Controller~1 is applied (in particular, when $\rho$ is controlled using~(\ref{eq:rhoDyn})), two pairwise collisions occur during the robots' approach to the steady-state circular trajectory: this is confirmed by the plot of the inter-distances $D_{ij}$, two of which go below the required threshold of $2\,r=0.5$~m. The application of Controller~1$^*$, in which~(\ref{eq:rhoDynColl}) is used in place of~(\ref{eq:rhoDyn}), is instead successful; in particular, the figure clearly shows how the controller prevents radial motion towards the target until a sufficient phase separation is achieved. The global quantity $\sigma$ is estimated as explained in Section~\ref{Sect:DecEstSigma}.

To obtain a smoother behavior, a sinusoidal transition from 0 to 1 was used for $\lambda$ in place of the linear transition entailed by~(\ref{eq:lambda_obs}).

\begin{figure*}[tb]
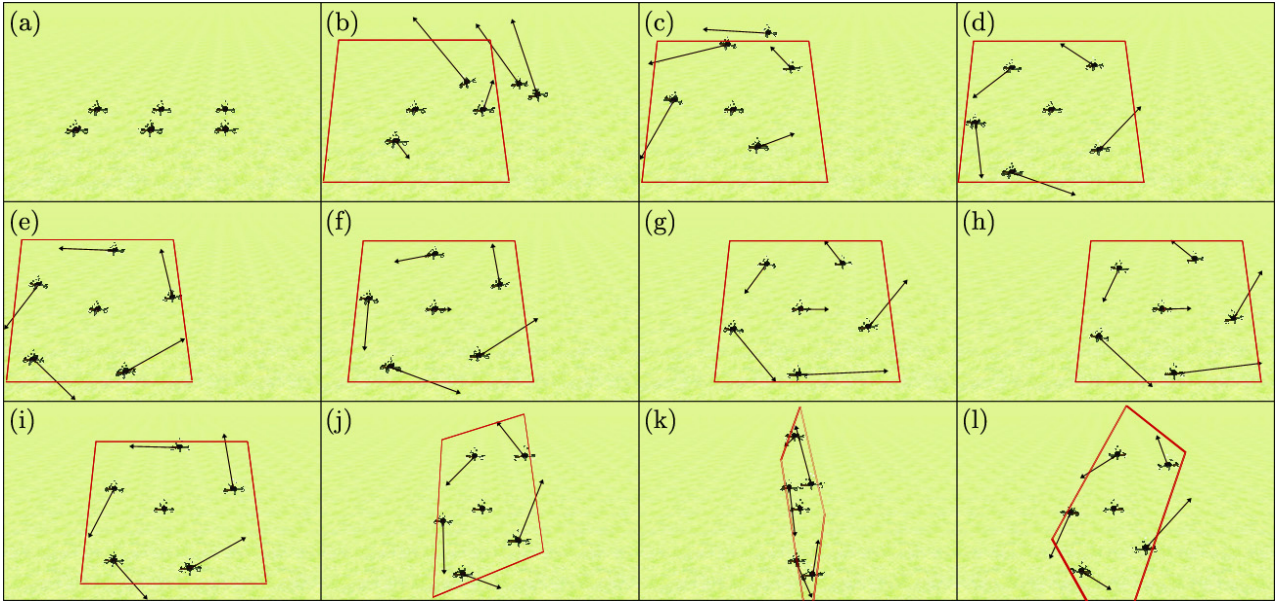

\centering
\FigSimDScreenShots
\caption{Quadrotor UAVs, first simulation: Some representative snapshots. a) The starting formation with the six quadrotors hovering above the ground.
b-d) Five quadrotors encircle the stationary quadrotor, which acts as target. e--h) The encirclement continues with the target now moving on a line
left to right. i--l) Final encirclement with the target stationary again but the encirclement plane rotating. The arrows represent the reference velocity vector $\pv_i$. The target plane is shown in red.}  
\label{fig:ScreenShots}
\end{figure*}

\begin{figure*}[t]
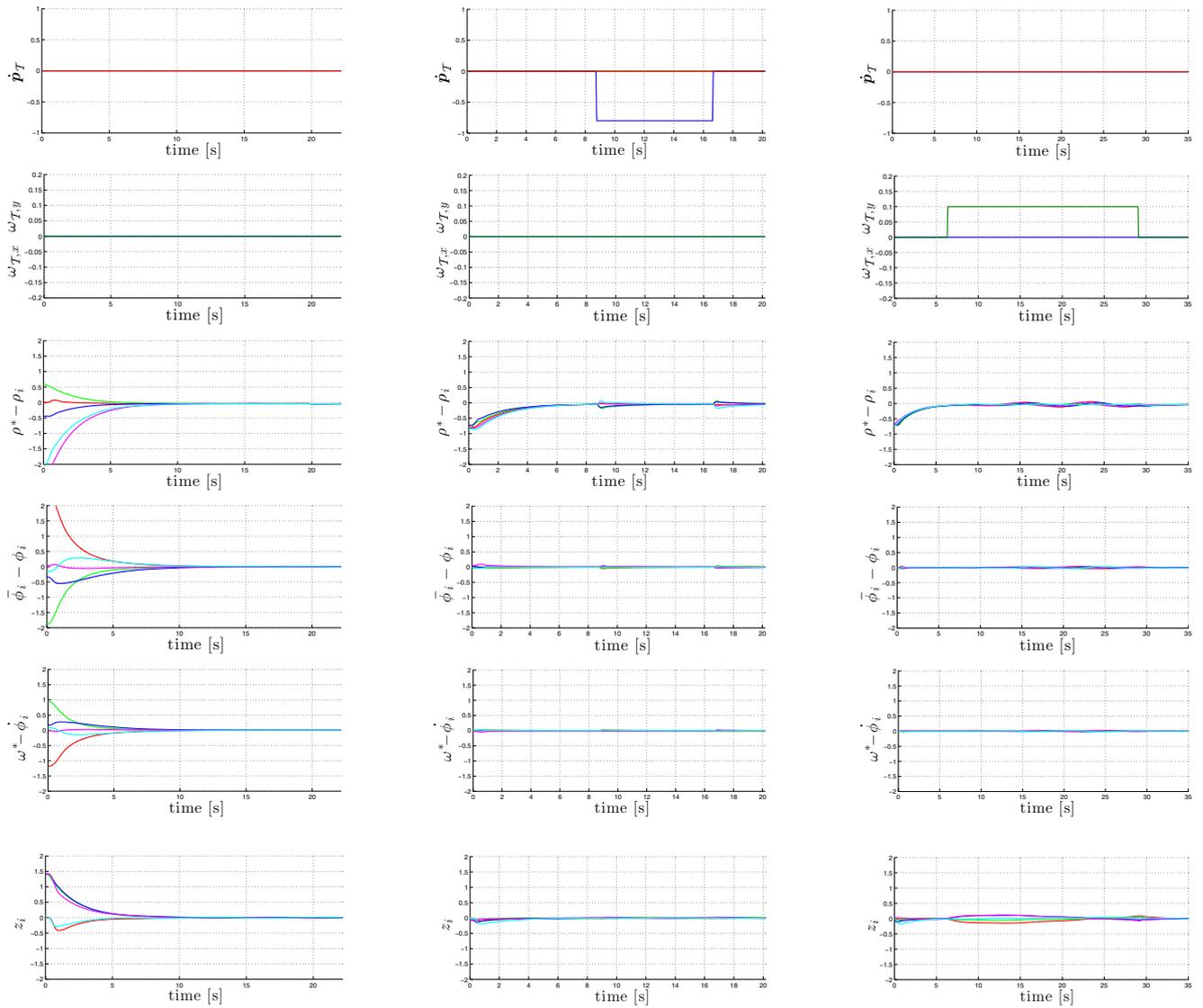

\centering
\FigSimDADotPT     \hfill \FigSimDEDotPT  \hfill   \FigSimDDDotPT  \\
\FigSimDAOmegaT    \hfill \FigSimDEOmegaT \hfill   \FigSimDDOmegaT     \\
\FigSimDAErrRho    \hfill \FigSimDEErrRho \hfill   \FigSimDDErrRho     \\
\FigSimDAErrPhi    \hfill \FigSimDEErrPhi \hfill   \FigSimDDErrPhi    \\
\FigSimDAErrDotPhi \hfill \FigSimDEErrDotPhi \hfill  \FigSimDDErrDotPhi\\
~\subfloat{\FigSimDAErrZ}
      \hfill
\subfloat{ \FigSimDEErrZ }
   \hfill
\subfloat{    \FigSimDDErrZ}
\caption{Quadrotor UAVs, second set of simulations under unmodeled perturbations. Each column (left, center, right) refers to a different simulation. For each simulation, the two top plots above show the velocity of the target and the angular velocity of the encirclement plane, whereas the four bottom plots show the evolution of the encirclement errors.}  
\label{fig:SimD-AE}
\end{figure*}

\subsection{Simulations with Quadrotor UAVs}
\label{sec:quadsim}

\noindent
To further validate our approach in a more realistic scenario, a second simulation study was performed on quadrotor UAVs. In particular, quadrotors are simulated as rigid bodies with a mass of approximately $0.75$\,kg subject to four generalized forces (one thrust and three torques) which are related to the rotational speeds of the four rotors. To this end, SwarmSimX~\citep{2012m-LaeFraBueRob} was used together with the TeleKyb~\citep{2013j-GraRieBueRobFra} framework.

Clearly, a quadrotor cannot be modeled as a simple integrator. However, its center of mass can track any smooth trajectory because its position is (part of) a differentially flat output. Therefore, we use the proposed encirclement schemes to produce a reference trajectory $\pv_i(t)$ for the center of mass of the $i$-th quadrotor, and rely on the built-in tracking controller for generating actual motion commands. In particular, each quadrotor has a built-in trajectory tracking controller with a standard two-stage structure (see, e.g.,~\cite{2013b-LeeFraSonBueRob} for details). The first stage (Cartesian controller) takes as reference the trajectory $\pv_i(t)$ with its time derivatives\footnote{Note that the first-order derivative $\dot \pv_i=\uv_i$ is directly given by the general expression~(\ref{eq:feed_lin}), whereas the second-order derivative $\ddot{\pv}_i(t)$ is numerically computed.} $\dot\pv_i(t)$ and $\ddot\pv_i(t)$, and generates the desired acceleration of the center of mass via a simple PD + feedforward controller:
\begin{align}
\av_{com,i} = \ddot{\pv}_i + k_p(\pv_i-\pv_{com,i}) + k_d(\dot\pv_i-\dot\pv_{com,i}), 
\label{eq:pd_ctrl} 
\end{align}
where $\pv_{com,i}$ is the position of the center of mass of the $i$-th quadrotor. In the second stage, $\av_{com,i}$ is first converted via the quadrotor model to the desired values of roll, pitch and thrust that would generate such acceleration given the current yaw; then, the desired values for the roll and pitch angles are used as reference signals for a PID attitude controller, which generates the torque to be applied to the quadrotor through the propeller rotational speeds.  
This simple cascaded approach for trajectory tracking relies on the fact that the attitude controller is much faster than the Cartesian controller. However, since the former relies on approximate linearization around zero roll and pitch angles, it is only accurate for near-hovering trajectories. In such conditions, this approach has been successfully employed in practice (see, e.g.,~\cite{2012f-FraSecRylBueRob}).

In the simulations, five quadrotors are in charge of the encirclement task while a sixth quadrotor (actually, its center of mass) acts as target. Figure~\ref{fig:ScreenShots} summarizes the results of a typical simulation, in which both the target and the encirclement plane are first stationary; then, the target moves at constant velocity; and finally the encirclement plane rotates. Controller~1 (Desired Angular Speed) is used for controlling the phase of the quadrotors, with $k_\rho =0.5$, $k_z =0.5$ and $k_\phi =0.5$. The desired encirclement values are set to $\rho^*=2$~m and $\omega^*=0.8$~rad/s. Finally, the control gains in~\eqref{eq:pd_ctrl} are set to $k_p = 9$ and
$k_d = 7.5$.
The quadrotors are able to track the reference trajectory very closely, and therefore the encirclement task is successfully executed.
 
To further show the robustness of the proposed encirclement controllers to unmodeled dynamics, we present a set of simulations in which the proportional term of the quadrotor Cartesian controller is suppressed by setting $k_p=0$ in~(\ref{eq:pd_ctrl}). In addition, the encirclement control law $\uv_i$ is computed at the actual position $\pv_{com,i}$ of the quadrotor rather than at the nominal position $\pv_i$. 
These modifications, aimed at emphasizing the non-ideal behavior of the quadrotor with respect to the integrator dynamics, lead to the following Cartesian controller:
\[
\av_{com,i}  = \dot\uv_i + k_d(\uv_i-\dot\pv_{com,i}), 
\]
with $\dot\uv$ computed numerically from $\uv_i$. The values of all the other controller gains are the same as in the previous case, as well as the values of $\rho^*$ and $\omega^*$.

The results shown in Fig.~\ref{fig:SimD-AE} refer to three specific cases: 
\begin{enumerate}
\item $\pv_\calT$ and $\omega_\calT$ are both identically zero (first column);
\item $\pv_\calT$ is a rectangular impulse along the $X_\calW$ direction and  $\omega_\calT$ is identically zero (second column);
\item $\pv_\calT$ is identically zero and $\omega_\calT$ is a rectangular impulse around the $Y_\calW$ axis (third column).
\end{enumerate}
In particular, note that the rotation of the encirclement plane violates the near-hovering assumption implicit in the design of the built-in trajectory controller.
Altogether, the plots of the encirclement errors confirm that the proposed scheme is rather robust in practice, as transient converge quickly and steady-state errors, when present, are very small.

\subsection{Experiments with Differential-Drive Robots}
\label{sec:exp}

\noindent
The proposed control framework for encirclement in 3D space can be directly applied to the 2D case by assuming that the encirclement plane coincides with the motion plane (this simply leads to zeroing the $z$ coordinate in all formulas). Accordingly, an experimental validation of the proposed approach has been carried out using a team of Khepera~III wheeled mobile robots 
 Each of these small-size differential-drive vehicles has been equipped with a Hukuyo 
 URG-04LX laser range finder, that has an angular field-of-view of $240^\circ$ and thus leaves a blind zone of $120^\circ$ behind the robot. 
Simultaneous calibration of odometric and sensor parameters was performed using the algorithm in~\cite{2013d-CenFraMarOri}. The built-in wi-fi card allows each robot to communicate with the others.

\begin{figure}[t]
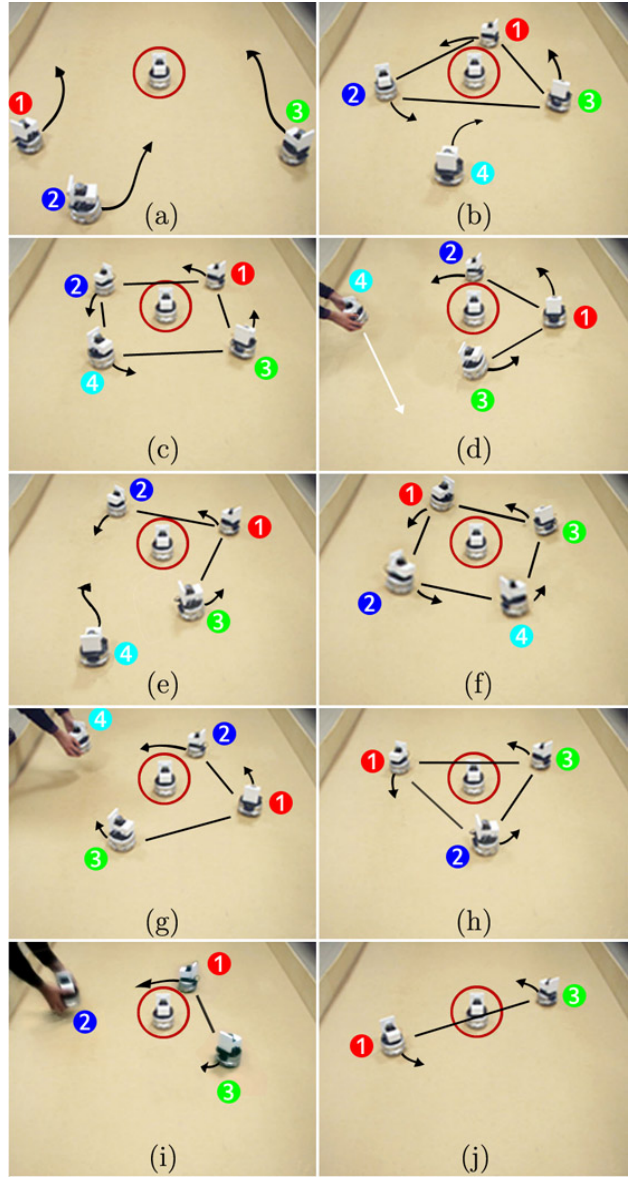

\figExpStatTarget
 \caption{Differential-drive robots, first experiment: Some representative snapshots. The target robot (shown enclosed in a red circle) is stationary. a) The initial configuration of the multi-robot system. b) When the three robots have achieved the encirclement task, a fourth robot (4) is released. c) The robots rearrange themselves in a rotating square formation. d-e) Robot 4 is kidnapped and released at a different location. f) The rotating square formation is recovered. g) Robot 4 is removed. h) A triangular encirclement formation is achieved again. i) Robot 2 is also removed. j) The two remaining robots assume a dipolar encirclement formation.} 
\label{fig:expSnap}
\end{figure}

Experiments involve a total of five robots, one of which acts as target (either stationary or moving) while the others must achieve encirclement.
Each robot inspects its own laser scan with a feature extraction algorithm that looks for the typical indentations caused by robots located inside the field of view, whose relative positions with respect to the sensor is returned. These instantaneous, anonymous measurements (the identity of the detected robots is unknown) are then broadcast to the other robots together with odometric data. Using this information, each robot performs mutual localization using the method of~\cite{2010f-FraOriSte,2013g-FraOriSte}, thus obtaining an estimate of the relative position of all the robots whose data it has received, now labeled with their identity. This localization step is essential for enabling each robot to localize other agents moving in its blind zone, a situation which occurs invariably during encirclement (e.g., at steady-state). Moreover, thanks to the reconstruction of the robot identities, the target can be readily identified.
Altogether, our relative localization module provides all the information needed for implementing the encirclement controllers without requiring an external tracking system.

Coming to the implementation of the controller, we exploit the fact that --- like quadrotors --- differential-drive robots are differentially flat systems, the flat output being the midpoint between the two wheels. This point can therefore track any smooth trajectory. As before, we use the proposed encirclement schemes to produce a reference trajectory $\pv_i(t)$ for the midpoint of the $i$-th robot, and then use the trajectory tracking controller of~\cite{2002-OriDelVen} to track it. The whole framework has been implemented in MIP, a in-house developed software platform specifically aimed at multi-robot systems. 

\begin{figure}[t]
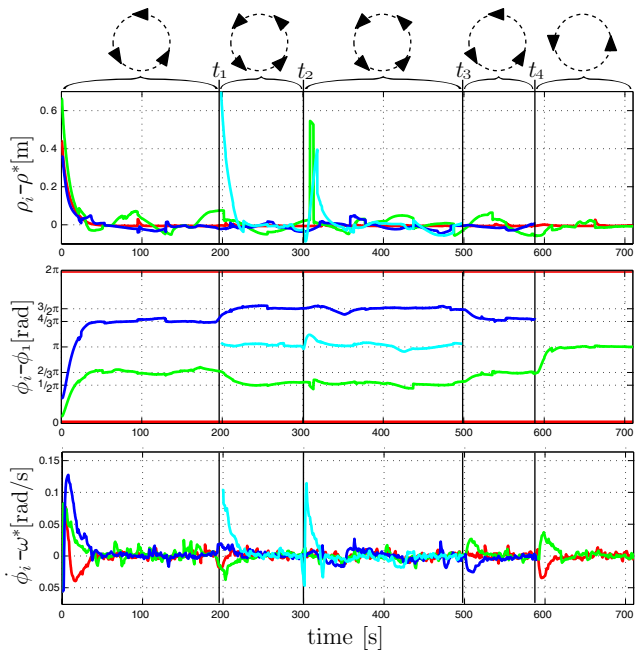

\figKhIIIExpPlots
\caption{Differential-drive robots, first experiment: Encirclement accuracy in the various stages of the experiment, with the desired rotating formation in each stage shown above the plots. Top: for each robot, difference between the current radius and the desired encirclement radius. Center: for each robot, difference between the current phase and the phase of robot 1. Bottom: for each robot, difference between the current angular speed and the desired encirclement speed.}
\label{fig:expPlots}
\end{figure}

In the first experiment, the target is stationary. Controller~1$^*$ (Desired Angular Speed with Collision Avoidance) is used for achieving collision-free encirclement. The desired values for the encirclement radius and angular speed are $\rho^*=0.5$~m and $\omega^*=0.06$~rad/s, respectively, while the control gains are $k_\rho=0.1$ and $k_\phi=0.06$.  At the beginning of the experiment, summarized in Fig.~\ref{fig:expSnap}, the multi-robot system consists of three robots that quickly achieve encirclement in a regular triangular formation. Another robot is then made available, and the group automatically arranges itself in a rotating square formation, which is momentarily lost but  promptly recovered when one of the robots is kidnapped and released at a different location. Two of the robots are then removed in sequence, causing the encirclement formation to become first a triangle and then a dipole.  

A more quantitative evaluation of the first experiment is given in Fig.~\ref{fig:expPlots}. In particular, the performance of the encirclement scheme is evaluated through the behavior of the radius, angular speed and phase of each robot during the various stages of the experiment. Practical convergence of the first two quantities to the desired values is confirmed, while the phase plots show that the appropriate splay state formation is achieved in each stage of the experiment. Note the quickly decaying transients at the start of the experiment and whenever there is a discontinuity in the localization estimates: i.e., at the birth of a new estimate associated to a robot being added to the group (time $t_1$), at a jump in the estimate of a robot kidnapped and released in a different location (time $t_2)$ or at the death of an estimate associated to a robot being removed from the group (times $t_3$ and $t_4$).

This experiment proves the robustness of the proposed encirclement controller, and in particular shows the seamless operation of the overall framework in the presence of a variable number of robots. 

\begin{figure}[t]
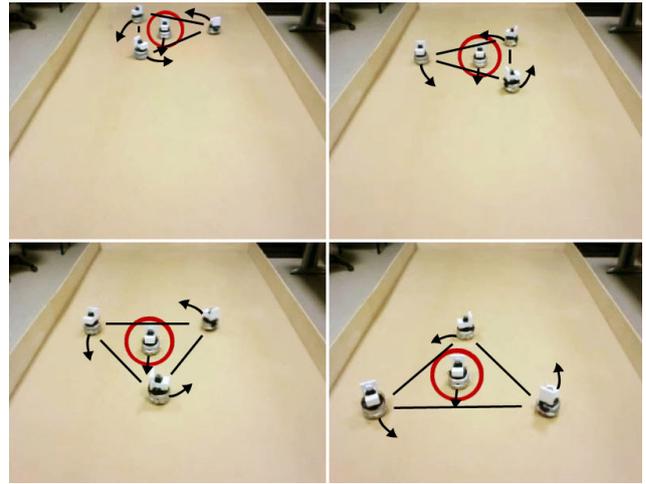

\figExpMovingTarget
\caption{Differential-drive robots, second experiment: Some representative snapshots. The target robot (shown enclosed in a red circle) moves along a rectilinear path. Nevertheless, encirclement is effectively achieved.}
\label{fig:ExpMovingTarget}
\end{figure}

In the second experiment, the target robot moves along a straight line with a constant velocity, with three robots in charge of the encirclement task. As before, this is achieved in a collision-free fashion by using Controller~1$^*$, with the same reference values and gains of the previous experiment. The snapshots shown in Fig.~\ref{fig:ExpMovingTarget} confirm that the robots are effectively able to encircle the moving target while arranging themselves in a rotating regular formation. As a result, each robot moves along a generalized trochoid. 

As for the previous simulations, video clips of these two experiments are contained in the multimedia material attached to the paper.

\fi
\ifx\noconc\undefined
\section{Conclusions}\label{sec:concl}

\noindent
In this paper, we have formulated and solved the problem of encircling a target moving in 3D space using a multi-robot system. In particular, three decentralized controllers have been proposed for different versions of the problem, and their effectiveness has been formally proven. An extension ensuring collision-free motion in the case of finite-size robots has also been proposed. Decentralized schemes for the estimation of the relevant global quantities have also been designed to guarantee that each robot can implement its controller using local information. The proposed strategy has been successfully validated through simulations on kinematic point robots and quadrotor UAVs, as well as experiments on differential-drive wheeled mobile robots. 
 
Future work will include:

\begin{itemize}

\item for the application to robots with complex dynamics, the analysis of a reference trajectory generation scheme based on continuous replanning;

\item the formulation and solution of a 3D encirclement problem on multiple planes, in which the robots should tend to arrange themselves along the vertices of a polyhedron;

\item experimental validation on a team of quadrotors.
 
\end{itemize}

\begin{small}
\bibliographystyle{plainnat}
\bibliography{./alias,./main,./bibCustom}
\end{small}

\end{document}